\documentclass[a4paper,11pt]{article}
\usepackage{blindtext,titlefoot}

\usepackage{graphicx} 
\usepackage{amsthm,amsmath}
\usepackage[numbers]{natbib}
\usepackage{blindtext,titlefoot}
\usepackage{bbm}
\usepackage{amsfonts,amsthm,amssymb,amsmath}
\usepackage[title,titletoc,toc]{appendix}
\usepackage[utf8]{inputenc}
\usepackage[affil-it]{authblk}
\usepackage{xcolor}
\usepackage{tikz}
\usepackage{mathtools}
\usepackage{natbib}
\usepackage{hyperref}
\usepackage{mathrsfs}
\usepackage{enumitem}

\newcommand\COMP{\hbox{C\kern -.58em {\raise .54ex \hbox{$\scriptscriptstyle |$}}
\kern-.55em {\raise .53ex \hbox{$\scriptscriptstyle |$}} }}
\newcommand\NN{\hbox{I\kern-.2em\hbox{N}}}
\newcommand\RR{\hbox{I\kern-.2em\hbox{R}}}
\newcommand\sRR{{\it \hbox{I\kern-.2em\hbox{R}}}}
\newcommand\QQ{\hbox{I\kern-.53em\hbox{Q}}}
\newcommand\PP{\hbox{I\kern-.53em\hbox{P}}}
\newcommand\EE{\hbox{I\kern-.53em\hbox{E}}}
\newcommand\ZZ{{{\rm Z}\kern-.28em{\rm Z}}}
\newcommand\be{\begin{equation}}

\newcommand{\E}{\mathbb E}
\makeatletter
\newcommand*\bigcdot{\mathpalette\bigcdot@{.5}}
\newcommand*\bigcdot@[2]{\mathbin{\vcenter{\hbox{\scalebox{#2}{$\m@th#1\bullet$}}}}}

\makeatother
\newcommand{\is}{\bigcdot }

\numberwithin{equation}{section}

\DeclareMathOperator{\essinf}{ess\;inf}
\DeclareMathOperator{\esssup}{ess\;sup}

\def \Lbrack {[\![}
\def \Rbrack {]\!]}

\newtheorem{theorem}{Theorem}[section]

\newtheorem{remark}[theorem]{Remark}

\newtheorem{lemma}[theorem]{Lemma}

\newtheorem{corollary}[theorem]{Corollary}
\newtheorem{definition}[theorem]{Definition}

\def \Lbrack {[\![}
\def \Rbrack {]\!]}

\hypersetup{
    colorlinks=true,
    linkcolor=blue,
    filecolor=magenta,
    urlcolor=cyan,
    citecolor=red
}
\title{The second-order Esscher martingale densities for continuous-time market models\footnote{The major part of this research was designed and achieved when the first author visited the department of Applied Mathematics, Computer Science and Statistics at Ghent University. The visit was fully funded by the FWO mobility grant V500324N}}
\author{Tahir Choulli and  Ella Elazkany \\
Mathematical and Statistical Sciences Dept.\\
University of Alberta, Edmonton, Canada
\and 
Mich\`ele Vanmaele\\ 
Department of Applied Mathematics, Computer Science and Statistics,\\
Ghent University, Ghent, Belgium}

\begin{document}

\maketitle

\begin{abstract}
In this paper, we introduce the second-order Esscher pricing notion for continuous-time models. Depending whether the stock price $S$ or its logarithm is the main driving noise/shock in the Esscher definition, we obtained two classes of second-order Esscher densities called linear class and exponential class respectively. Using the semimartingale characteristics to parametrize $S$, we characterize the second-order Esscher densities (exponential and linear) using pointwise equations. The role of the second order concept is highlighted in many manners and the relationship between the two classes is singled out for the one-dimensional case. Furthermore, when $S$ is a compound Poisson model, we show how both classes are related to the Delbaen-Haenzendonck's risk-neutral measure. Afterwards, we restrict our model $S$ to follow the jump-diffusion model, for simplicity only, and address the bounds of the stochastic Esscher pricing intervals. In particular, no matter what is the Esscher class, we prove that both bounds (upper and lower) are solutions to the same linear backward stochastic differential equation (BSDE hereafter for short) but with two different constraints. This shows that BSDEs with constraints appear also in a setting beyond the classical cases of constraints on gain-processes or constraints on portfolios. We prove that our resulting constrained BSDEs have solutions in our framework for a large class of claims' payoffs including any bounded claim, in contrast to the literature, and we single out the monotonic sequence of BSDEs that ``{\it naturally}" approximate it as well.
\end{abstract}

\section{Introduction}
The Esscher transform is a time-honoured  concept in actuarial sciences, which was suggested by the Swedish actuary Esscher in \cite{Esscher}, and it is also statistically known as the {\it exponential tilting}.  Then the idea of changing the probability measure to obtain a consistent positive linear pricing rule, which appeared in the actuarial literature in the context of equilibrium reinsurance markets in  \cite{{buhlmann1980economic}}, gave to the Esscher transform an important role in actuarial sciences. Afterwards the Esscher transform was extended by Gerber and co-authors, in  \cite{gerber1994b}, to the case where the logarithm of the stock is a stochastic process with stationary and independent increments. The impact of  latter extension was huge in both actuarial sciences and finance as well,  see \cite{gerber1996}, and the references therein  to cite a few, for more details and related discussions. The third key milestone in the evolution of the Esscher concept was elaborated by B\"uhlmann-Delbaen-Embrechts-Shiryaev in   \cite{buhlmann1996no}. Here, the authors derive the {\it dynamical Esscher transform}, for discrete time models, and baptized it as {\it conditional Esscher}. In fact, for any $t=1,\ldots,T$, the Esscher parameter $\theta_t$ obtained based on the information up to time $t-1$, and hence it is stochastic and the obtained density process $Z$ takes the form of 
\begin{equation*}
Z_n=\exp\left(\sum_{k=1}^n\theta_k(X_k-X_{k-1})+\sum_{i=1}^{n} K_i\right),\quad\mbox{where $K_i$ is observable at time $i-1$},\end{equation*}
$X$ is the logarithm of the stock price, and $K$ is called the exponential cumulant. This interesting formulation of the dynamical Esscher that led Kallsen and Shiryaev  to extend {\it the dynamical Esscher} to the general continuous-time semimartingale setting in \cite{Kallsen}. In this latter paper, which is the first to define the general continuous-time version of the Esscher, the authors singled out two classes called exponential and linear Esscher depending whether $X$ is $S$ itself or represents its logarithm.  For more development about the continuous-time Esscher  and its numerous applications, we refer the reader to \cite{Rheinlander1,hubalek2006,hubalek2009,benth2012risk,elliott2022generalized} and the references therein to cite a few. The paper  \cite{Kallsen} gives various sufficient conditions on the cumulant process to guarantee the uniform integrability of the resulting density process, while it did not address the necessary and/or sufficient for the existence of those Esscher densities nor their relationship.
\\

Certainly the Esscher risk-neutral measures, when they exist, gave a nice and clear linear arbitrage-free pricing rule in incomplete markets besides other optimal neutral-risk measures. However, as noted in \cite{bondi2020comparing} for instance, at the practical and the computational aspects the Esscher rule might not be as efficient as one can wish. In this spirit,  the second-order Esscher transform was introduced in \cite{monfort2012asset}, where they show that this two-parameters Esscher transform adds an important flexibility compared to the one parameter Esscher for pricing derivatives. For given historical dynamics which are estimated using stock return data and under no-arbitrage conditions, the two-parameters Esscher contrary to the one parameter Esscher still provides free parameters to match derivative prices.\\

{\it What are our achievements?} We introduce the Esscher martingale densities/deflators of-order-two for any $d$-dimensional continuous-time semimartingale model $S$. We single out, as in \cite{Kallsen}, that there are the exponential and the linear classes of Esscher densities of-order-two, and we characterize them using pointwise equations involving the characteristics of the underlying model $S$. We elaborate the exact relationship between the first-order-Esscher  and our second-order-Esscher for each class, where the change of priors in the jump-modelling play a central role. For the class of linear-Esscher of-second-order, we give necessary and sufficient conditions for the existence in various manners. In fact, in an equivalent manner, we use some {\it local viability} of the market or equivalently a non-arbitrage condition, or the existence of a solution to an optimization problem which is dual to an optimal portfolio optimization.  For the case of the one-dimensional case of $S$, we single out the exact relationship between the exponential-Esscher and the linear-Esscher. For the case when $S$ is one-dimensional and $\ln(S)$ follows a jump-diffusion model, we describe the upper and lower Esscher price processes, and found that they are the smallest solution to {\it constrained} linear backward stochastic differential equations (BSDE hereafter for short). The obtained constrained linear BSDEs are new, as they are involved with a Skorokhod condition even though the value process does not appear in the constraint.\\

The paper has four sections including this introductory section. The second section presents the mathematical model and gives some of its preliminary analysis that will be useful throughout the paper. The third section introduces the second-order Esscher deflators/densities for the most general $d$-dimensional continuous-time market models. The last section focuses on the class of one-dimensional jump-diffusion models, for simplicity, and address the bounds resulting from the second-order Esscher pricing for any claim that is ``{\it nicely integrable}". The paper has three appendices, where the proofs for the technical lemmas are detailed therein, and some useful results are recalled.  
\section{The mathematical model and preliminaries}\label{Section4ModelNotations}
 Throughout the paper, we are supposed given a filtered probability space $(\Omega, \mathcal{F}, (\mathcal{F}_t)_{t\geq0}, \mathbb{P})$. Here the filtration is supposed to satisfy the usual condition, i.e. it is right-continuous and complete.
\subsection{General notations} 
 Throughout the paper, for any probability $Q$ on $(\Omega,{\cal{F}})$, we denote ${\cal A}(Q)$ (respectively ${\cal M}(Q)$) the set
of right-continuous with left limits (RCLL hereafter) and $\mathbb{F}$-adapted processes with $Q$-integrable variation (respectively that are -uniformly $Q$-integrable martingales). When the probability measure is not mentioned, then by default we are using the probability measure $\mathbb{P}$. The set ${\cal{V}}^+$ denotes the set of all RCLL, nondecreasing, and $\mathbb{F}$-adapted processes with finite values.  
For any process $Y$, we denote by $^{o,\mathbb H}(Y)$  (respectively $^{p,\mathbb H}(Y)$)  the
$\mathbb{F}$-optional (respectively $\mathbb{F}$-predictable) projection of $Y$. For an increasing process $V$, we denote $V^{o,\mathbb{F}}$ (respectively $V^{p,\mathbb{F}}$) its dual $\mathbb{F}$-optional (respectively $\mathbb{F}$-predictable) projection. ${\cal O}$, ${\cal P}$ and  $\mbox{Prog}$ represent the $\mathbb{F}$-optional, the $\mathbb{F}$-predictable and the $\mathbb{F}$-progressive $\sigma$-fields  respectively on $\Omega\times[0,+\infty[$. For an semimartingale $Y$, we denote by $L(Y)$ the set of all $Y$-integrable processes in the Ito's sense, and for $H\in L(Y)$, the resulting integral is one-dimensional semimartingale denoted by $H\is Y:=\int_0^{\cdot} H_udY_u$. If $\mathcal{C}$ is a set of $\mathbb{F}$-adapted processes,
then $\mathcal{C}$ --except when it is stated otherwise-- is the set of processes, $Y$,
for which there exists a sequence of stopping times,
$(T_n)_{n\geq 1}$, that increases to infinity and $X^{T_n}$ belongs to $\mathcal{C}$, for each $n\geq 1$. The set of special semimartingales is denoted by $\mathscr{S}_p$.
\subsection{The model and its parametrization}
  We consider a $d$-dimensional quasi-left-continuous semimartingale $X = \left(X^{(1)},...,X^{(d)}\right)^{tr}$, which represents the {\it driving shock} process for the $d$-risky assets $S$, and is described mathematically using the {\it predictable characteristics}, as follows
\begin{equation}\label{Xcanon}
    X = X_0 + X^c + B + h_{\epsilon}(x) \star (\mu - \nu) + \big(x - h_{\epsilon}(x) \big) \star \mu.
\end{equation}
Here, $\mu$ is the random measure associated to the jumps of $X$ and is defined by
\begin{equation}\label{mu}
    \mu (dt,dx) := \sum_{0\leq s} \delta_{(s,\Delta X_s)} (dt,dx) \mathbbm{1}_{\{\Delta X_s \neq 0\}},
\end{equation}
$X^c$ is the continuous local martingale part of $X$, $B$ is a predictable process of finite variation, the random measure $\nu$ is the compensator of the random measure $\mu$, $h_{\epsilon}(x) := x \mathbbm{1}_{\{|x| \leq\epsilon \}}$ is the canonical truncation function with ${\epsilon}\in(0,1)$ fixed throughout the paper, and $C$ is the matrix with entries $C^{ij} := \langle X^{c,i}, X^{c,j} \rangle$. Furthermore, we can find a version of the characteristics satisfying
\begin{equation}\label{characteristicsDef}
    B = b \bigcdot A, \quad C = c \bigcdot A, \quad \nu(\omega, dt, dx) = dA_t(\omega)F_t(\omega, dx), \quad
    F_t (\{0\}) = 0, \quad \int (|x|^2 \wedge 1)F_t(dx) \leq 1.
\end{equation}
where $A$ is a right-continuous and nondecreasing and predictable process, which we choose to be continuous because we assume that $X$ is quasi-left-continuous, $b$ and $c$ are predictable processes. The decomposition (\ref{Xcanon}) is known as the canonical representation of $X$, while $(b,c,F, A)$ is called predictable characteristics of $X$. \\

Throughout the paper, the following notation will be useful
\begin{equation}
\begin{split}
& {\large\bf{e}}(x):=(e^{x_1},e^{x_2},...,e^{x_d})^{tr}\in\RR^d,\quad \mbox{For any}\quad x\in\RR^d,\\
&{\bf\large{diag}}(x):=\mbox{the diagonal matrix associated to the vector $x\in\RR^d$}.
\end{split}
\end{equation}
For any $d$-dimensional semimartingale $Y$, we define ${\cal{E}}(Y)$ as the unique solution to the following SDE
\begin{equation}\label{Exponnetial4dDimension}
dZ={\bf{diag}}(Z_{-})dY,\quad Z_0={\mathbb{I}}_d:=(1,1,...,1)^{tr}\in\RR^d.
\end{equation}
 The following lemma addresses the exponential stochastic of a $d$-dimensional process.
\begin{lemma} Let $X$ be given by (\ref{Xcanon}), (\ref{mu}) and (\ref{characteristicsDef}). Then the following assertions hold.\\
{\rm{(a)}} There exists unique $d$-dimensional semimartingale $\widetilde{X}$, such that 
\begin{equation}\label{fromXtoXtilde1}
{\bf{e}}(X)={\cal{E}}\left(\widetilde{X}\right)={\bf{diag}}(S_0)^{-1}S=\left({\cal{E}}(\widetilde{X}^{(1)}),...,{\cal{E}}(\widetilde{X}^{(d)})\right)^{tr},\end{equation}
and is given by 
\begin{equation}\label{fromXtoXtilde2}
\widetilde{X}=X+{1\over{2}}\Delta\is{A}+\left({\bf{e}}(x)-{\mathbb{I}}_d-x\right)\star\mu,\quad\mbox{where}\quad \Delta:=\left(c_{11},c_{22},...,c_{dd}\right)^{tr}.
\end{equation}
Furthermore, $\widetilde{X}$ has the following canonical decompositions given by 
\begin{equation}\label{Canonical4Xtilde}
\begin{split}
\widetilde{X}=X_0+X^c+h_{\delta}(x)\star(\widetilde\mu-\widetilde\nu)+\widetilde{b}\is{A}+\left(x-h_{\delta}(x)\right)\star\widetilde\mu.
\end{split}\end{equation}
Here $\delta:=e^{\epsilon}-1$, and $(\widetilde{b},\widetilde\mu,\widetilde{F},\widetilde\nu)$ is given (for any ${\cal{F}}\otimes{\cal{B}}(\mathbb{R}^d)$-measurable $W$) by
\begin{equation}\label{Parameters4Xtilde}
\begin{split}
&\widetilde{b}:=b+{{\Delta}\over{2}}+\int\left(({\bf{e}}(x)-{\mathbb{I}}_d)I_{\{\Vert{\bf{e}}(x)-{\mathbb{I}}_d\Vert\leq\delta\}}-h_{\epsilon}(x)\right)F(dx),\quad W\star\widetilde{\mu}:={W}(\cdot,{\bf{e}}(x)-{\mathbb{I}}_d)\star\mu,\\
&\int W(t,x)\widetilde{F}(dx):=\int{W}(t,{\bf{e}}(x)-{\mathbb{I}}_d(x))F(dx),\quad\quad \widetilde{\nu}(dt,dx):=\widetilde{F}_t(dx)dA_t.
\end{split}\end{equation}
{\rm{(b)}} If $X$ is locally bounded, then we have 
\begin{equation}\label{Xbounded}
\begin{split}
&X=X_0+X^c+x\star(\mu-\nu)+b'\is A\quad\mbox{and}\quad \widetilde{X}=X_0+X^c+x\star(\widetilde\mu-\widetilde\nu)+\widetilde{b}'\is A,\\
&\mbox{where}\quad b':=b+\int xI_{\{\Vert{x}\Vert>\epsilon\}}F(dx),\quad \widetilde{b}':=b'+{{\Delta}\over{2}}+\int ({\bf{e}}(x)-{\mathbb{I}}_d-x)F(dx).
\end{split}
\end{equation}

\end{lemma}
The proof of this lemma follows directly from It\^o's formula, see \cite[]{Kallsen} for the case of one dimensional, and hence it will be omitted herein. We end this preliminary section by defining some sets of local martingale densities/deflators, which appear {\it naturally} in our analysis. 
\begin{definition}\label{DefatorIntergrable} Let $Y$ a $d$-dimensional semimartingale, $Z$ be a process, and $Q$ be a probability. \\
{\rm{(a)}} $Z$ is a local martingale density (deflator) for $(Y,Q)$ if $Z>0$ and both $ZY$ and $Z$ are $Q$-local martingales. The set of local martingale densities (deflators) for $(Y,Q)$ will be denoted by ${\cal{Z}}_{loc}(Y,Q)$, and ${\cal{Z}}^{L\log{L}}_{loc}(Y,Q)$ denotes the set of $Z\in {\cal{Z}}_{loc}(Y,Q)$ satisfying $Z\ln(Z)$ is a $Q$-local submartingale. When  $Q=P$, we simply omit the probability and write $ {\cal{Z}}_{loc}(Y)$ and ${\cal{Z}}^{L\log{L}}_{loc}(Y)$ respectively.\\
{\rm{(b)}}  Let $L$ be a positive local martingale, and let $(T_n)_{n\geq 1}$ be the sequence of stopping times that increases to infinity such that $L^{T_n}\in{\cal{M}}$ and $Q_n:=L_{T\wedge{T}_n}\cdot P$. Then  we define
\begin{equation}
\begin{split}
{\cal{Z}}_{loc}(Y,L)&:=\left\{Z:\ Z^{T_n}\in {\cal{Z}}_{loc}(Y^{T_n},Q_n),\quad n\geq 1\right\},\\
 {\cal{Z}}^{L\log{L}}_{loc}(Y,L)&:=\left\{Z:\ Z^{T_n}\in {\cal{Z}}^{L\log{L}}_{loc}(Y^{T_n},Q_n),\quad n\geq 1\right\}.
\end{split}\end{equation}
\end{definition}
We end this section with the following useful definition
\begin{definition} Let $U$ and $V$ be two real-valued processes with $U_0=V_0=0$. Then we denote $V\preceq U$ if $U-V$ is nondecreasing.
\end{definition}

\section{Esscher pricing densities of-order-two for continuous-time}

The Esscher transform of-order-two was introduced in \cite{monfort2012asset} as the measure 
\begin{equation*}
Q:=\exp(\theta^{tr}{Y_1}+Y_1^{tr}\psi{Y}_1)\left(E[\exp(\theta^{tr}{Y_1}+Y_1^{tr}\psi{Y}_1)]\right)^{-1}\cdot{P},\end{equation*}
 where $Y_1$ is a d-dimensional random variable, $\theta\in{\mathbb{R}}^d$ and $\psi$ is $d\times{d}$-matrix of real numbers. Thus, by extending this notion to the conditional Esscher setting {\it \`a la} Buhlmann-Delbaen-Emprecht-Shirayev, we obtain the following density $Z$ 
$$
Z_n:=\prod_{i=1}^n{{\exp(\theta^{tr}_i\Delta{Y}_i+\Delta{Y}_i^{tr}\psi_i\Delta{Y}_i)}\over{E[\exp(\theta^{tr}\Delta{Y}_i+\Delta{Y}_i^{tr}\psi_i\Delta{Y}_i)|{\cal{F}}_{i-1}]}},\quad n\geq 1,$$
where $\theta_i$ and $\psi_i$ are ${\cal{F}}_{i-1}$-measurable vectors.  This extension sounds tailored-made for the multi-period models. By putting $K_n:=\sum_{i=1}^n \ln\left(E[\exp(\theta^{tr}\Delta{Y}_i+\Delta{Y}_i^{tr}\psi_i\Delta{Y}_i)|{\cal{F}}_{i-1}]\right)$, which is a predictable process, the above equality takes the form of 
\begin{equation*}
Z_n=\exp\left(\sum_{i=1}^n\theta_i^{tr}\Delta{Y}_i+\sum_{i=1}^n(\Delta{Y}_i)^{tr}\psi_i\Delta{Y}_i-K_n\right),\quad n=0,1,....\end{equation*}
This allows us to introduce the notion of {\it Esscher  of-order-two} in the continuous-time setting. To this end, for any $d$-semimartingale $Y$, we define
\begin{equation}\label{ThetaSet}
\Theta(Y):=\left\{ (\theta,\psi)\ \ \mathbb{R}^d\times\mathbb{R}^{d\times d}\mbox{-valued predictable process}:\ \begin{array}{l}\theta \in L(Y)\quad \mbox{and}\\
 \psi\ \mbox{is locally bounded}\end{array}\right\}.
\end{equation} 
\begin{definition}\label{EsscherwithOrders}Let $Z$ be a positive, RCLL and adapted process.\\
{\rm{(a)}} $Z$ is called an Exponential-Esscher density of-order-two for $(S,\mathbb{F})$ if there exist $(\theta,\psi) \in \Theta(X)$  and a RCLL predictable with finite variation $K$ such that 
\begin{equation}
   Z= Z^{(\theta, \psi)} := \exp \bigg(\theta\is{X} + \sum_{0<s\leq\cdot} \Delta{X}_s^{tr}\psi_s\Delta{X}_s - K \bigg)\in {\cal{Z}}_{loc}(S).
\end{equation}
 The pair $(\theta, \psi)$ is called the Esscher parameters of the density $Z^{(\theta,\psi)}$. If $Z^{(\theta,\psi)}$ is uniformly integrable, then $Q:=Z_{\infty}\cdot P$ is called an exponential-Esscher-pricing measure of-order-two for $S$. The set of these pricing densities (respectively measures) will be denoted by ${\cal{Z}}^{EE}(S)$ (respectively ${\cal{Q}}^{EE}(S)$).\\
 {\rm{(b)}}  $Z$ is called Linear-Esscher density of-order-two for $S$  if there exist $(\theta,\psi) \in \Theta(S)$ and a RCLL predictable with finite variation $\widetilde{K}$ such that 
\begin{equation}
   Z= Z^{(\theta, \psi)} := \exp \bigg(\theta \is{S} + \sum_{0<s\leq\cdot} \Delta{S}_s^{tr}\psi_s\Delta{S}_s - \widetilde{K} \bigg)\in {\cal{Z}}_{loc}(S).
\end{equation}
The set of linear-Esscher-pricing  densities (respectively measures) of order-two  for $S$ will be denoted by ${\cal{Z}}^{LE}(S)$ (respectively ${\cal{Q}}^{LE}(S)$).
\end{definition}
In the definition of the set $\Theta(Y)$ we suppose that $\psi$ is locally bounded, for the sake of simplifying the analysis and avoiding technicalities only. Indeed, this assumption on $\psi$ can be replaced by the condition $\sum\vert (\Delta{Y})^{tr}\psi\Delta{Y}\vert\in{\cal{V}}^+$. When $Y$ is locally bounded, the local boundedness assumption of $\psi$ simplifies tremendously the statements of the results and their proofs as well. Furthermore, when looking for the upper and lower prices,  just like the practical framework, this assumption is not restrictive at all, and in this case we will allow $\psi$ to span the set of predictable and bounded processes. 
\begin{lemma}\label{Equivalntform4LinearEsscher}
$Z$ is called Linear-Esscher density of-order-two for $S$  iff there exist $(\overline{\theta},\overline{\psi}) \in \Theta(\widetilde{X})$ and a RCLL predictable with finite variation $\widetilde{K}$ such that 
\begin{equation}
   Z=  \exp \bigg(\overline{\theta}\is\widetilde{X} + \sum_{0<s\leq\cdot} \Delta\widetilde{X}_s^{tr}\overline{\psi}_s\Delta\widetilde{X}_s - \widetilde{K} \bigg), 
\end{equation}
and both $Z$ and $ZS$ are local martingales.
\end{lemma}
\begin{proof} As the two $d\times{d}$-dimensional processes ${\bf{diag}}(S_{-})$ and ${\bf{diag}}(S_{-})^{-1}$ are locally bounded, then it is clear that $L(\widetilde{X})=L(S)$ and $\Theta(S)=\Theta(\widetilde{X})$. Furthermore, for any $(\theta,\psi)\in \Theta(S)$, we put $\overline{\theta}:={\bf{diag}}(S_{-})\theta$ and $\overline{\psi}:={\bf{diag}}(S_{-})\psi{\bf{diag}}(S_{-})$, and derive
$$
\theta\is{S}=\overline{\theta}\is\widetilde{X}\quad\mbox{and}\quad \Delta{S}^{tr}\psi\Delta{S}= \Delta{\widetilde{X}}^{tr}\overline\psi\Delta{\widetilde{X}}.$$
Thus, the proof of the lemma follows immediately from combining these facts.
\end{proof}

\begin{remark} 1) It is clear that our second-order Esscher pricing density generalizes the Esscher pricing density  by putting $\psi\equiv 0$.\\
2) For $(\theta,\psi) \in \Theta(S,\mathbb{F})$ such that $\psi$ is $X$-integrable, we consider the process  
\begin{equation}
\overline{Z}^{(\theta,\psi)}:=\exp \bigg(\theta \bigcdot X + \sum_{i,j=1}^d [{X}^{i},\psi^{ij}\is{X}^{j}] - K' \bigg).
\end{equation}
Then it is clear that we can also defined Esscher pricing density of-order-two any process taking the form of $\overline{Z}^{(\theta,\psi)}$ such that $\overline{Z}^{(\theta,\psi)}\in {\cal{Z}}_{loc}(S)$. In fact these two definitions are equivalent, as one can prove that $\overline{Z}^{(\theta,\psi)}={Z}^{(\theta,\psi)}$ when we put $K=K'-  \sum_{i,j=1}^d [({X}^{i})^{c},\psi^{ij}\is{X}^{(j)}]$, which is a predictable process. \\
3) When $X$ is a continuous process, then second-order Esscher pricing density coincides with the Esscher pricing density. Indeed, in this  case, we write $\exp(\theta\is{X}+\psi\is[X,X]-K)=\exp(\theta\is{X}-K'')$ with $K''$ being a predictable process with finite variation.
\end{remark}
\subsection{Linear-Esscher martingale densities of-order-two}
This subsection addresses the linear Esscher martingale densities of-order-two, gives their characterization, singles out the necessary and sufficient conditions for their existence, and hows how they can be connected to the classical linear-Esscher martingale densities as well. 
\begin{theorem}\label{mainTHM4LinearEsscher}  Let $Z$ be a positive, RCLL and adapted process. Consider $\widetilde{X}$ given by (\ref{fromXtoXtilde2}) or (\ref{Canonical4Xtilde}),  and $(\widetilde{b},\widetilde{c},\widetilde{F},\widetilde{\mu}, \widetilde{\nu})$ defined in (\ref{Parameters4Xtilde}). Then the following assertions are equivalent.\\
{\rm{(a)}} $Z$ is a Linear-Esscher density of-order-two for $S$.\\
{\rm{(b)}}  There exists a pair $(\theta,\psi)\in\Theta(S)$ satisfying 
  \begin{equation}\label{IntegrabilityCondition4Linear}
\theta^{tr}c\theta\is{A}+\exp(\theta^{tr}x+x^{tr}\psi{x})\Bigl(\Vert x\Vert{I}_{\{\Vert{x}\Vert>\epsilon\}}+I_{\{\epsilon<\vert\theta^{tr} x\vert\}}\Bigr) \star\widetilde{\mu}  \in{\cal{A}}_{loc}^+,\quad \mbox{for some}\ \epsilon>0,
\end{equation}
\begin{equation}\label{mgEquation4Linear}
\widetilde{b}+ c\theta + \int \bigg(x\exp(\theta^{tr} x + x^{tr}\psi x)-h_{\epsilon}(x) \bigg)\widetilde{F}(dx) = 0,\quad P\otimes{A}-a.e., \end{equation}
and 
\begin{equation}\label{Zequation4Linear}
Z=\mathcal{E}\Bigl(\theta \is X^c + \left(\exp(\theta^{tr} x+x^{tr} \psi x)-1\right)\star (\widetilde{\mu}-\widetilde{\nu}) \Bigr).
 \end{equation}
{\rm{(c)}}  There exists $(\theta,\psi)\in\Theta(S)$ satisfying (\ref{IntegrabilityCondition4Linear}) and $Z=\mathcal{E}\Bigl(\theta \is X^c + \left(\exp(\theta^{tr} x+x^{tr} \psi x)-1\right)\star (\widetilde{\mu}-\widetilde{\nu}) \Bigr)$ and belongs to ${\cal{Z}}_{loc}(S)$.
\end{theorem} 
The proof will be given in  Subsection \ref{Subsection4Proofs}. 
\begin{remark}\label{Condition(3.6)forBoundedS} Suppose $S$ is locally bounded, and consider a predictable and locally bounded $\psi$. Then the condition  (\ref{IntegrabilityCondition4Linear})  is equivalent to  
\begin{equation}\label{SimplifiedIntreCondi}
\theta^{tr}c\theta\is{A}+{U}_{\epsilon}\in{\cal{A}}^+_{loc},\quad\mbox{where}\quad{U}_{\epsilon}:=e^{\theta^{tr}x}I_{\{\vert\theta^{tr}x\vert>\epsilon\}}\star\widetilde\mu,\quad \epsilon\in(0,\infty).\end{equation}
Indeed, due to the local boundedness of $S$ and that of $\psi$, the process $\vert{x}^{tr}\psi{x}\vert\star\widetilde\mu$ is locally bounded.
\end{remark}
Theorem \ref{mainTHM4LinearEsscher} gives us a complete and explicit characterization, as a solution to a pointwise equation, of the linear-Esscher pricing densities of-order-two. It is clear then that when we put $\psi\equiv 0$, our linear-Esscher of-order-two reduces to the classical linear-Esscher (of first order), which coincides with the minimal entropy-Hellinger martingale density as noticed in \cite[Remark??]{??} for continuous-time setting and in \cite[Theorem 5.1]{Choulli2006} for discrete-time framework. Recall that the minimal entropy-Hellinger martingale density introduced in \cite{Choulli2005}, is the local martingale deflator/density that minimizes the entropy-Hellinger process. Thus, in the spirit of \cite{Choulli2005}, we assume that $S$ is locally bounded and derive various deep characterizations.
 \begin{theorem}\label{mainTHM4LinearEsscherBIS} Suppose that $S$ is locally bounded, and consider $(\widetilde{b},\widetilde{F},\widetilde\mu, ,\widetilde\nu)$ and $\widetilde{b}'$ defined by (\ref{Canonical4Xtilde})-(\ref{Parameters4Xtilde}) and (\ref{Xbounded}) respectively.
 Then the following assertions are equivalent.\\
 {\rm{(a)}} ${\cal{Z}}^{LE}(S)\not=\emptyset$, i.e. $S$ admits a linear-Esscher pricing density of-order-two.\\
   {\rm{(b)}} ${\cal{Z}}^{L\log{L}}_{loc}(S)\not=\emptyset$, i.e. there exists $Z\in{\cal{Z}}_{loc}(S)$ such that $Z\ln(Z)$ is a local submartingale.\\
   {\rm{(c)}} For any locally bounded and predictable $\psi$, the following pointwise minimization problem
  \begin{equation}\label{MinimizationProblem}
  \min_{\theta}\Biggl(\theta^{tr}\widetilde{b}'+{1\over{2}}\theta^{tr}c\theta+ \int \bigg(\exp(\theta^{tr} x + x^{tr}\psi x)-1-\theta^{tr}x \bigg)\widetilde{F}(dx)\Biggr),
  \end{equation}
  admits a solution $\widetilde\theta:=\widetilde\theta(\psi)$ satisfying 
  \begin{equation}\label{Integrability4LinearExponential}
  \widetilde\theta\in L(S)\quad\mbox{and}\quad\widetilde\theta^{tr}c\widetilde\theta\is{A}+f_{L\log{L}}(\exp(\widetilde\theta^{tr}x+x^{tr}\psi{x})-1)\star\widetilde\mu\in{\cal{A}}^+_{loc},\end{equation}
   where
     \begin{equation}\label{functionf(LlogL)}
f_{L\log{L}}(y):=(y+1)\ln(1+y)-y,\quad\forall\ y>-1.\end{equation}
   {\rm{(d)}} There exists a locally bounded and predictable $\psi$ such that the pointwise minimization problem (\ref{MinimizationProblem}) admits a solution $\widetilde\theta$ satisfying (\ref{Integrability4LinearExponential}).
\end{theorem}
The proof of the theorem is relegated to Subsection \ref{Subsection4Proofs}, while herein we discuss the well posedness  of its ingredients and its meaning afterwards. 
\begin{remark}\label{Remark4SlocallyBoundedCase}
{\rm{(a)}} If $S$ is locally bounded, then for any $\theta$ and any $\psi$ the integral  in the right-hand side of (\ref{MinimizationProblem}) is well defined. In fact, by stopping we can assume without loss of generality that $S$ is bounded, and in this case we have $\Vert{x}\Vert\leq C$ $d\widetilde{F}$-a.e. $C\in(0,\infty)$, and hence
\begin{equation*}
\begin{split}
\int \vert\exp(\theta^{tr} x + x^{tr}\psi x)-1-\theta^{tr}x\vert\widetilde{F}(dx)\leq \exp(c\Vert\theta\Vert+C^2\Vert\psi\Vert)\Vert\psi\Vert\int \Vert{x}\Vert^2\widetilde{F}(dx)<\infty,\quad P\mbox{-a.s.}.
\end{split}
\end{equation*}
This proves the claim.\\
{\rm{(b)}} Suppose $S$ is locally bounded and let $(\theta,\psi)$ be a  pair of predictable  processes such that $\psi$ is locally bounded and $\Bigl(\left(\exp(\theta^{tr}x+x^{tr}\psi{x})-1\right)^2\star\widetilde\mu\Bigr)^{1/2}\in{\cal{A}}_{loc}^+$. Then 
$\vert{x}^{tr}\psi{x}\vert\exp(\theta^{tr}x+x^{tr}\psi{x})\star\widetilde\mu\in{\cal{A}}_{loc}^+$. 
\end{remark}
Theorem \ref{mainTHM4LinearEsscherBIS} gives various necessary and sufficient conditions for the existence of a linear-Esscher martingale densities of-order-two (i.e. ${\cal{Z}}^{LE}(S)\not=\emptyset$) when $S$ is locally bounded. On the one hand, the condition ${\cal{Z}}^{LE}(S)\not=\emptyset$ is completely characterized via the arbitrage condition ${\cal{Z}}^{L\log{L}}_{loc}(S)\not=\emptyset$, which conveys the existence of a deflator for $S$, belonging locally to the space of Llog{L}-integrable martingales. This arbitrage condition is equivalent to the existence of ``local" solution to the exponential utility maximization from terminal wealth. For further discussions and details about this claim, we refer the reader to \cite[Lemma 3.3]{ChoulliDengMa}. The second characterization is via a maximization problem which has a striking similarities with the minimization problem intrinsic to the minimization  of Hellinger processes in \cite{Choulli2005}. This conveys the possible intimate link between second order linear-Esscher and the classical linear-Esscher densities somehow. This is the aim of the following theorem.
\begin{theorem}\label{Charaterization4LinearEsscher} 
 Suppose $S$ is locally bounded, and for any locally bounded and predictable $\psi$ we denote
  \begin{equation}\label{ZLlogLandZ(psi)}
  Z^{(\psi)}:={\cal{E}}\Bigl(\left(\exp(x^{tr}\psi{x})-1\right)\star(\widetilde{\mu}-\widetilde{\nu})\Bigr), \quad\mbox{and}\quad \widetilde{\nu}^{(\psi)}(dt,dx):=\exp({x^{tr}\psi{x}})\widetilde{\nu}(dt,dx).
  \end{equation}
  If $Z$ is a RCLL adapted process, then the following assertions hold.\\
  {\rm{(a)}} $Z\in {\cal{Z}}^{LE}(S)$ if and only if there exists a predictable pair $(\theta,\psi)$  such that 
  \begin{equation}\label{Decompositon4SecondEsscher}
  \begin{split}
\psi\ \mbox{is locally bounded, $(\theta,\psi)$ fulfills (\ref{Integrability4LinearExponential}), and }\quad {{Z}\over{Z^{(\psi)}}}\in{\cal{Z}}_{loc}(S, Z^{(\psi)}).
\end{split}
  \end{equation}
  {\rm{(b)}}  Furthermore, $Z/  Z^{(\psi)}={\cal{E}}\left(\theta\is{X}^c+(e^{\theta^{tr}x}-1)\star(\widetilde{\mu}-\widetilde{\nu}^{(\psi)})\right)$, and  for any locally bounded and predictable $\psi$, ${\cal{Z}}^{LE}(S)\not=\emptyset$ if and only if ${\cal{Z}}^{L\log{L}}_{loc}(S, Z^{(\psi)})\not=\emptyset$.
\end{theorem}
The proof of the theorem is relegated to Subsection \ref{Subsection4Proofs}. The theorem conveys simultaneously two main ideas which are intimately realted. On the one hand, it says that any linear-Esscher martingale density of-order-two --when it exists-- coincides with the ``classical" (order one) linear-Esscher martingale density {\it under a specific local change of measure}. In fact, consider any predictable locally bounded $\psi$ such that $Z^{(\psi)}$ given in (\ref{ZLlogLandZ(psi)}) belongs to ${\cal{M}}$. Then put $Q^{(\psi)}:=Z^{(\psi)}_{\infty}\cdot{P}$ and hence $Z$ is a linear-Esscher density of-order-two with parameters $(\theta,\psi)$ if and only if there exists a ``classical" linear-Esscher density $\widetilde{Z}$ (given by $\widetilde{Z}:=Z/Z^{(\psi)}$) for the model $(S,Q^{(\psi)})$ (i.e. linear-Esscher of-order-one for $S$ under the probability $Q^{(\psi)}$). 
On the other hand, in virtue of the equality in (\ref{Decompositon4SecondEsscher}), one can interpret the linear-Esscher of-oder-two as a linear-Esscher under a class of equivalent priors, and this connect second-order Esscher to one-order Esscher via the uncertainty modelling.
\begin{corollary}\label{Corollary4JumpDiffusion} Consider $Z>0$  and  let $(\mu,\sigma,\gamma)$ be a triplet of  predictable and bounded processes such that $1/\vert{\gamma}\vert$ is bounded as well. Suppose $S$ is one-dimensional given by 
    \begin{equation}\label{JumpDiffusion}
   S=S_0\exp(X),\quad X= \int_0^{\cdot}{\mu}_s ds + {\sigma}\is{W}+{\gamma}\is{N}^\mathbb{F},\quad\mbox{and}\quad N^\mathbb{F}_t := N_t -\lambda t,\end{equation}
Then $Z^{LE}(S)\not=\emptyset$, and $Z\in{Z}^{LE}(S)$ if and only if there exists $(\theta, \psi)\in\Theta(\widetilde{X})$ such that 
\begin{equation}
\begin{split}
&\int_0^{\cdot}\exp(\theta_t\widetilde\gamma_t)dt\in{\cal{A}}_{loc}^+,\quad \widetilde{b}+\theta{\sigma}^2+\lambda\widetilde{\gamma}\left(e^{\theta\widetilde{\gamma}+\psi\widetilde{\gamma}^2}-1\right)=0,\ P\otimes dt\mbox{-a.e.},\\
&\mbox{and}\quad Z={\cal{E}}\left(\theta\widetilde{\sigma}\is{W}+(e^{\theta\widetilde{\gamma}+\psi\widetilde{\gamma}^2}-1)\is{N}^{\mathbb{F}}\right).
\end{split}
\end{equation}
Here the triplet $(\widetilde{b}, \widetilde{\sigma},\widetilde{\gamma})$ is given by
\begin{equation}\label{parameters}
\widetilde{b}:=b+{{\sigma^2}\over{2}}+\lambda\left(e^{\gamma}-1-\gamma\right),   \quad \widetilde{\sigma}:={\sigma}   \quad\mbox{and}\quad \widetilde{\gamma}:=e^{\gamma}-1.
\end{equation}
\end{corollary}
\begin{proof} Remark that in this case, as $\gamma$ is bounded, we have 
\begin{equation*}
\begin{split}
\widetilde{X}&=X+{1\over{2}}\int_0^{\cdot}\sigma_s^2 ds+(e^{\gamma}-1-\gamma)\is N=\int_0^{\cdot}\widetilde{b}_s{d}s+\sigma\is{W}+(e^{\gamma}-1)\is{N}^{\mathbb{F}}\\
&=\int_0^{\cdot}\widetilde{b}_s{d}s+\sigma\is{W}+\widetilde\gamma\is{N}^{\mathbb{F}}.
\end{split}
\end{equation*}
Thus, by using the canonical decomposition of $X$, with the truncation function $h_{\epsilon}$, we derive
\begin{equation*}
\begin{split}
&b=\mu-\lambda\gamma{I}_{\{\vert\gamma\vert>\epsilon\}},\quad c=\sigma^2,\quad F_t(dx):=\lambda\delta_{\gamma_t}(dx),\quad A_t=t,\\
&\widetilde{b}=\mu-\lambda\gamma+{{\sigma^2}\over{2}}+\lambda(e^{\gamma}-1){I}_{\{\vert\gamma\vert\leq\epsilon\}},\quad \widetilde{\mu}(dt,dx)=dN_t\delta_{\widetilde\gamma_t}(dx), \quad\widetilde{F}_t(dx)=\lambda\delta_{\widetilde\gamma_t}(dx).
\end{split}
\end{equation*}
Therefore, the proof of the corollary follows immediately from combining these with Theorem \ref{mainTHM4LinearEsscher}.  This ends the proof of the corollary.
\end{proof}
We end this subsection by illustrating the linear-Esscher pricing measure of order two on the compound Poisson process, which is highly used in modelling risks in actuarial sciences. 
\begin{corollary}\label{Corollary4CompoundJumps}  Suppose that $S$ is a one-dimensional compound Poisson process given by 
\begin{equation}\label{LevyModel4X}
 S=S_0\exp(X),\quad X:= \sum_{k=1}^{N_t} J_k,\end{equation}
 where $(J_k)_{k\geq1}$is a sequence of independent and identically distributed satisfying $\mathbb{E}[e^{2J_1}]<\infty$, and is independent of the Poisson process $N$ with rate $\lambda$. \\
 For any real number $\psi$, $Z$ is a linear-Esscher martingale density of-order-two for $S$ with parameter $\psi$ iff there exists a real number $\theta$ satisfying
        \begin{equation}\label{mgEquation1}
       E\Biggl[(e^{J_1}-1)\exp\Bigl(\theta(e^{J_1}-1)+\psi(e^{J_1}-1)^2\Bigr)\Biggr]= 0, 
        \end{equation}
        and
        \begin{equation}\label{Z4CompundJumps}
        \begin{split}
           & Z=\exp\Biggl(\sum_{n=1}^{N_t}\left(\theta(e^{J_n}-1)+\psi(e^{J_n}-1)^2\right)-\lambda\widetilde\kappa(\theta,\psi)t\Biggr),\\
            & \widetilde\kappa(\theta,\psi):=e^{-\theta+\psi}E\exp\left(\theta{e}^{J_1}+\psi{e}^{2J_1}-2\psi{e}^{J_1}\right)-1.
            \end{split}
        \end{equation}
\end{corollary}
The proof of the corollary is a direct consequence of Theorem \ref{mainTHM4LinearEsscher}, and hence it will be omitted.\\

Herein,  for the model (\ref{LevyModel4X}), we will discuss the connection between linear-Esscher measures and the existing pricing measures.  In fact, is clear that the linear-Esscher pricing measure differs tremendously from Gerber-Shui's  risk-neutral measure, which is more related to the exponential-Esscher, and hence this connection will be addressed in the next subsection. However, the linear-Esscher measure of-order-two falls into the family of Delbaen-Haezendonck's pricing measures,  see\cite{DelbaenEsscher} for more details about this latter measure and its application in pricing. Indeed, for the model (\ref{LevyModel4X}), Delbaen and Haezendonck introduced the following risk-neutral measure 
\begin{equation}\label{DelbaenMeasure}
Z^{DH}:=\exp\Biggl(\sum_{n=1}^N \beta(J_n)-\lambda{t}E\left[e^{\beta(J_1)}-1\right]\Biggr),\quad t\geq 0,
\end{equation}
where $\beta$ is a Borel-measurable function. Thus, the linear-Esscher measure of-order-two is  a particular of the Delbaen-Haezendonck's measure by choosing $\beta(x)=\theta(e^{x}-1)+\psi(e^{x}-1)^2$, $x\in\mathbb{R}$.
\subsection{Exponential-Esscher martingale densities of-order-two}\label{Subsection4ExponneialEsscher}
This subsection treats the exponential Esscher pricing densities. 
\begin{theorem}\label{mainTHM4ExponenialEsscher} Suppose $S$ is locally bounded, and let $\widetilde{b}'$ given by (\ref{Xbounded}). Then the following hold.\\
{\rm{(a)}} Let $Z$ be a positive RCLL and adapted process.  Then $Z$ is an Exponential-Esscher density of-order-two for $S$ if and only if there exists $(\theta,\psi)\in\Theta(X)$ satisfying
  \begin{equation}\label{IntegrabilityCondition4EE}
\exp(\theta^{tr}x) I_{\{\theta^{tr}x>\alpha\}}\star\mu \in{\cal{A}}_{loc}^+,\quad \mbox{for some $\alpha>0$},\end{equation}
\begin{equation}\label{mgEquation4Exponential}
    0=
    \widetilde{b}'+ c\theta+ \int \bigg(\exp(\theta^{tr} x + x^{tr}\psi x)-1\bigg)({\bf{e}}(x)-{\mathbb{I}}_d){F}(dx),\quad P\otimes{A}-a.e.,
 \end{equation}
 and 
 \begin{equation}\label{ZFormEquation4Exponential}
Z=\mathcal{E}\Bigl(\theta \is X^c + \left(\exp(\theta^{tr} x+x^{tr} \psi x)-1\right)\star ({\mu}-{\nu}) \Bigr).
 \end{equation}
  {\rm{(b)}} $Z\in {\cal{Z}}^{EE}(S)$ if and only if there exists a predictable and locally bounded $\psi$ such that $Z/ \overline{Z}^{(\psi)}$ is an exponential Esscher density (i.e. of order one) for $S$ under $Z^{(\psi)}$ defined by 
  \begin{equation}\label{Zbar(psi)}
  \overline{Z}^{\psi}:={\cal{E}}\left((e^{x^{tr}\psi{x}}-1)\star(\mu-\nu)\right).
  \end{equation}
Furthermore, in this case, we have
  \begin{equation}\label{Decomposition4Z-Z(psi)}
  Z= \overline{Z}^{(\psi)}{\cal{E}}\left(\theta\is{X}^c+(e^{\theta^{tr}x}-1)\star({\mu}-{\nu}^{(\psi)})\right)\quad\mbox{with}\quad {\nu}^{(\psi)}(dt,dx):=\exp({x^{tr}\psi{x}}){\nu}(dt,dx)
  \end{equation}
\end{theorem}
Similarly as for the linear-Esscher case, the connection between  exponential-Esscher of-order-two and order-one is fully described via ``local" change of probability. However for the other aspect of our analysis to the exponential-Esscher deflators, the  higher dimension (i.e. $d\geq 2$) brings serious difficulties in treating this exponential-Esscher case. In fact when $d\geq 2$, in contrast to the linear-Esscher case, there is difficulties in connecting the condition ${\cal{Z}}^{EE}(S)\not=\emptyset$ to some known non-arbitrage condition.  Again, when $d\geq2$, the equations (\ref{mgEquation4Linear}) and (\ref{mgEquation4Exponential}) sound non-comparable at all. The one-dimensional case, however, allows us to overcome all these obstacles. 
\begin{theorem}\label{OneDimension}
Suppose $X$ is a one-dimensional (i.e. $d=1$) locally bounded semimartingale, and let  $(\widehat{S},\widehat{Z},\widehat\nu)$ be given by 
\begin{equation}\label{ShatZhatNuhat}
\widehat{S}:={\cal{E}}(\widehat{X}),\quad  \widehat{X}:={1\over{2}}c\is{A}+X,\quad \widehat{Z}:={\cal{E}}((f-1)\star(\mu-\nu)),\quad f(x):={{e^x-1}\over{x}}I_{\{x\not=0\}}+I_{\{x=0\}}.
\end{equation}
Then the following assertions hold.\\
{\rm{(a)}} Let  $Z$ be a RCLL and positive adapted process. Then $Z\in{\cal{Z}}^{EE}(S)$ with parameters $(\theta,\psi)\in\Theta(X)$ if and only if the process   
\begin{equation}\label{RelationshipZ2Zhat}
Z':={{Z}\over{\widehat{Z}}}\in {\cal{Z}}^{LE}(\widehat{S},\widehat{Z})\quad\mbox{with the same pair of parameters $(\theta,\psi)$}. \end{equation}
Furthermore, we have 
\begin{equation}\label{Description4Zprime}
Z'={\cal{E}}\Bigl(\theta\is{X}^c+(e^{x\theta+x^2\psi}-1)\star(\mu-\widehat\nu)\Bigr)\quad\mbox{where}\quad \widehat\nu(dt,dx):=f(x)\nu(dt,dx).\end{equation}
 {\rm{(b)}} ${\cal{Z}}^{EE}(S)\not=\emptyset$ if and only if ${\cal{Z}}^{L\log{L}}_{loc}(S)\not=\emptyset$ if and only if ${\cal{Z}}^{LE}(\widehat{S},\widehat{Z})\not=\emptyset.$
\end{theorem}
We end this subsection, by illustrating the above theorem on the particular cases of jump-diffusion models and compound Poisson models in two corollaries.
\begin{corollary}\label{JumpDiffusionCase}
    Let $Z>0$ and RCLL process, and suppose $S$ is one-dimensional given by (\ref{JumpDiffusion}).\\
     Then $Z$ is an exponential Esscher density of-order-two for $S$ iff there exists $(\theta,\psi)\in\Theta(X)$ such that 
        \begin{equation}\label{mgEquation1}
        \begin{split}
           & \widetilde{b} + \theta{\sigma}^2+\lambda\widetilde{\gamma}\bigg(e^{\theta{\gamma} + \psi{\gamma}^2}-1\bigg)=0.
            \end{split}
        \end{equation}
        and
        \begin{equation}\label{Z4JumpdDiffusion}
            Z=\mathcal{E}\Bigl(\theta{ \sigma}\is{ W} + (e^{\theta{ \gamma} + \psi{\gamma}^2} -1)\is{N}^\mathbb{F} \Bigr).
        \end{equation}
        \end{corollary}
   The proof of the corollary follows directly from Theorem \ref{mainTHM4ExponenialEsscher}, and hence it will be omitted.
 \begin{corollary}\label{LevyGeneral}
    Suppose that $S$ is one-dimensional follow the model defined in (\ref{LevyModel4X}). \\
  For any real number $\psi$,  $Z$ is an exponential-Esscher density of-order-two for $S$ iff there exists a real number $\theta$ such that 
        \begin{equation}\label{mgEquation1}
        \mathbb{E}\Bigl[(e^{J_1}-1)\exp\left(\theta{J_1}+\psi{J}_1^2\right)\Bigr]= 0, 
        \end{equation}
        and
        \begin{equation}\label{Z4CompundJumps}
            Z=\exp\Biggl(\sum_{n=1}^N(\theta{J_n}+\psi{J}_n^2)-\lambda\kappa(\theta,\psi)t\Biggr),\quad \kappa(\theta,\psi):=\mathbb{E}\left[\exp(\theta{J_1}+\psi{J}_1^2)\right]-1
        \end{equation}
  \end{corollary}
  The proof of the corollary is  a direct consequence of Theorem \ref{mainTHM4ExponenialEsscher}, and hence will be omitted. \\
 
 On the one hand, it is clear that, for the model (\ref{LevyModel4X}), our exponential-Esscher pricing measure extends the Gerber-Shiu's measure. This measure was introduced in  \cite{gerber1994a}, see also \cite{gerber1994b, gerber1996} for related works, and can be deducted from ours by putting $\psi=0$. On the other hand, the exponential-Esscher ,measure is a particular case of the Delbaen-Haenzendonck's pricing measure introduced in \cite{DelbaenEsscher},  and which we recalled its density in (\ref{DelbaenMeasure}). In fact, by choosing the quadratic form for $\beta(x)=\theta{x}+\psi{x}^2$ instead, the Delbaen-Haenzendonck's measure coincides with the exponential-Esscher. \\
 
We end this section by the proof of its main theorems.
\subsection{Proof of Theorems \ref{mainTHM4LinearEsscher}, \ref{mainTHM4LinearEsscherBIS}, \ref{Charaterization4LinearEsscher}, \ref{mainTHM4ExponenialEsscher} and  \ref{OneDimension}}\label{Subsection4Proofs}
The proof of these theorems relies essentially on the following lemmas which are interesting in themself.
\begin{lemma}\label{Integrability4ThetaProcesses}  Let  $Y$ be a $d$-dimensional semimartingale, $(b^Y, c^Y, F^Y)$ its predictable characteristics, $\mu_Y$ is the random measure of its jumps and $\nu_Y$ its compensator, and $\theta\in L(Y)$. Consider the following RCLL processes with finite variation
\begin{equation}\label{Definition4U(theta)}
\begin{split}
&U^{\theta}(Y):=\sum \Delta{Y}I_{\{\vert\theta^{tr}\Delta{Y}\vert>\epsilon\ \mbox{or}\ \vert\Delta{Y}\vert>\epsilon\}},\quad Z^{\theta}(Y):=Y-U^{\theta}(Y),\\
&U^{\theta,1}(Y):=\sum \Delta{Y}I_{\{\vert\Delta{Y}\vert>\epsilon\geq\vert\theta^{tr}\Delta{Y}\vert \}},\quad U^{\theta,2}(Y):=\sum \Delta{Y}I_{\{\vert\theta^{tr}\Delta{Y}\vert>\epsilon\}}.
\end{split}
\end{equation}
Then the following assertions hold.\\
{\rm{(a)}} We always have
\begin{equation}\label{ConditionIntegrability1}
\int \Vert{h_{\epsilon}(x)}\Vert{I}_{\{\epsilon<\vert\theta^{tr}x\vert \}}F^Y(dx)<\infty\  P\otimes{A^Y}\mbox{-a.e.},\ b^{\theta}(Y):=b^Y-\int{h_{\epsilon}(x)}{I}_{\{\epsilon<\vert\theta^{tr}x\vert \}}F^Y(dx)\in L(A^Y),\end{equation}
and $Z^{\theta}$ admits the following canonical decomposition
\begin{equation}\label{ZdecompositionCannonical}
Z^{\theta}(Y)=Y^c+\underbrace{x{I}_{\{\max(\vert{x}\vert,\vert\theta^{tr}x\vert)\leq\epsilon \}}\star(\mu^Y-\nu^Y)}_{=:M^{\theta}(Y)}+b^{\theta}\is A^Y.
\end{equation}
{\rm{(b)}} $\theta\in L(U^{\theta,1}(Y))\cap{L}(U^{\theta,2}(Y))\cap{L}(Y^c)\cap{L}(M^{\theta}(Y))$, and  $\theta\is Z^{\theta}(Y)$ is a special semimartingale having the following canonical decomposition  
\begin{equation}\label{Z(theta)Decomposition}
\theta\is{Z}^{\theta}(Y)=\theta\is{Y^c}+h_{\epsilon}(\theta^{tr}x){I}_{\{\Vert{x}\Vert\leq\epsilon \}}\star(\mu^Y-\nu^Y)+\left(\theta^{tr}b^Y-\int\theta^{tr}{x}{I}_{\{\Vert{x}\Vert\leq\epsilon<\vert\theta^{tr}x\vert \}}F^Y(dx)\right)\is{A}^Y.
\end{equation}
{\rm{(c)}}  $\theta\is{Y}$ has the following canonical decomposition
\begin{equation}\label{Decomposition4Theta(X)}
\theta\is{Y}=\theta\is{Y^c}+h_{\epsilon}(\theta^{tr}x)\star(\mu^Y-\nu^Y)+\left(\theta^{tr}b^{\theta}(Y)+\int{h}_{\epsilon}(\theta^{tr}x)I_{\{\epsilon<\Vert{x}\Vert\}}F^Y(dx)\right)\is{A}^Y+\theta\is{U}^{\theta,2}(Y).
\end{equation}
\end{lemma}
\begin{lemma}\label{TechnicalLemma2}
 Let $(\psi,\theta)\in\Theta(S)$. Then the following assertions hold.\\
 {\rm{(a)}} The condition (\ref{IntegrabilityCondition4Linear}) is equivalent to
  \begin{equation}\label{IntegrabilityCondition4Linear4Proof}
\Vert x\exp(\theta^{tr}x+x^{tr}\psi{x})-h_{\epsilon}(x)\Vert\star\widetilde{\mu} +\bigg| \exp({\theta^{tr} x + x^{tr} \psi x}) - 1-h_{\epsilon}(\theta^{tr}x) \bigg|\star\widetilde{\mu}  \in{\cal{A}}_{loc}^+.
\end{equation}
 {\rm{(b)}} The condition (\ref{IntegrabilityCondition4Linear})  implies $  \sqrt{(e^{{\theta^{tr} x + x^{tr} \psi x}} - 1)^2\star\widetilde\mu}\in{\cal{A}}_{loc}^+$.\\
  {\rm{(c)}} If (\ref{IntegrabilityCondition4Linear})  and (\ref{mgEquation4Linear}) hold and $\vert{e}^{x^{tr}\psi{x}}-1\vert\star\widetilde\mu\in{\cal{A}}_{loc}^+$, then 
 \begin{equation}\label{Conditions}
  \begin{split}
 \vert\theta^{tr}xe^{{\theta^{tr} x + x^{tr} \psi x}} -\theta^{tr}h_{\epsilon}(x)\vert\star\widetilde\mu\in{\cal{A}}_{loc}^+,\quad\vert\theta^{tr}xe^{{\theta^{tr} x + x^{tr} \psi x}}  -e^{{\theta^{tr} x + x^{tr} \psi x}} + 1\vert\star\widetilde\mu\in{\cal{A}}_{loc}^+.
 \end{split}
 \end{equation}
 {\rm{(d)}} Suppose $S$ is locally bounded.  Then the condition (\ref{IntegrabilityCondition4EE}) is equivalent to 
 \begin{equation}\label{Equa198}
  \begin{split}
&\Vert\left(\exp(\theta^{tr}x+x^{tr}\psi{x})-1\right)({\bf{e}}(x)-\mathbb{I}_d)\Vert\star\mu+\bigg| \exp({\theta^{tr} x + x^{tr} \psi x}) - 1-h_{\epsilon}(\theta^{tr}x) \bigg| \star{\mu} \in{\cal{A}}_{loc}^+\end{split}
\end{equation}
 \end{lemma}
\subsubsection{Proof of Theorems  \ref{mainTHM4LinearEsscher}, \ref{mainTHM4LinearEsscherBIS} and \ref{Charaterization4LinearEsscher}}
Besides Lemmas \ref{Integrability4ThetaProcesses}  and \ref{TechnicalLemma2}, the two theorems require the following third lemma. 
\begin{lemma}\label{Minimization2Root} Suppose $S$ is locally bounded, and let $\widehat\theta$ be a predictable process and $\psi$ a predictable and locally bounded process. Then the following assertions hold.\\
{\rm{(a)}} Let $Q$ be a probability and $(\widetilde{b}^Q,c^Q,\widetilde{F}^Q,A^Q)$ be the predictable characteristics of $\widetilde{X}$ under $Q$. Then $Q\otimes{dA^Q}$-a.e. $\widehat\theta$ is a pointwise solution to
\begin{equation}\label{minimization(Q)}
\min_{\theta}\left(\theta^{tr}\widetilde{b}^Q+{1\over{2}}\theta^{tr}c^Q\theta+\int\left(\exp(\theta^{tr}x)-1-\theta^{tr}x\right)\widetilde{F}^Q(dx)\right)
\end{equation}
 if and only if  $\widehat\theta$ satisfies
\begin{equation}\label{Equation(3.7)}
\widetilde{b}^Q+c^Q\widehat\theta+\int\left(x{e}^{\widehat\theta^{tr}x}-x\right)\widetilde{F}^Q(dx)\equiv 0,\quad Q\otimes{dA^Q}-a.e..
\end{equation}
{\rm{(b)}} Suppose that ${\widehat\theta}^{tr}c\widehat\theta\is{A}+\sqrt{(\exp(\widehat\theta^{tr}x+x^{tr}\psi{x})-1)^2\star\widetilde\mu}\in{\cal{A}}_{loc}^+$. Then\\
$Z:={\cal{E}}\left(\widehat\theta\is{X}^c+(\exp(\widehat\theta^{tr}x+x^{tr}\psi{x})-1)\star(\widetilde\mu-\widetilde\nu)\right)\in{\cal{Z}}_{loc}(S)$ if and only if $\widehat\theta$ solves (\ref{mgEquation4Linear}).
\end{lemma}
\begin{proof}[Proof of Theorem \ref{mainTHM4LinearEsscher}] 1) suppose that $Z\in{\cal{Z}}^{LE}(S)$, and hence there exists a triplet $(\theta,\psi,K)$ of predictable processes such that 
\begin{equation}
(\theta,\psi)\in\Theta(S),\quad K\in {\cal{A}}_{loc},\quad Z=\exp\left(\theta\is\widetilde{X}+\sum(\Delta\widetilde{X})^{tr}\psi\Delta\widetilde{X}-K\right)\in{\cal{Z}}_{loc}(S).
\end{equation}
Then we put $\widetilde{Y}:=\theta\is{\widetilde{X}}+\sum(\Delta{\widetilde{X}})^{tr}\psi\Delta{\widetilde{X}}-\widetilde{K},$ and  we use It\^o to derive the following.
\begin{equation*}
\begin{split}Z_{-}^{-1}\is{Z}&=\widetilde{Y}+{1\over{2}}\langle{\widetilde{Y}}^c\rangle+\sum\left(e^{\Delta{\widetilde{Y}}}-1-\Delta{\widetilde{Y}}\right)\\
&=\theta\is{\widetilde{X}}+\sum(\Delta{\widetilde{X}})^{tr}\psi\Delta{\widetilde{X}}-\widetilde{K}+{{{\theta^{tr}}c\theta}\over{2}}\is A+\left(\exp(\theta^{tr}{x}+{x}^{tr}\psi{{x}})-1-\theta^{tr}{x}-{x}^{tr}\psi{{x}}\right)\star\widetilde{\mu}\\
&=\theta\is{\widetilde{X}}+\left(\exp(\theta^{tr}{x}+{x}^{tr}\psi{{x}})-1-\theta^{tr}{x}\right)\star\widetilde{\mu}+{{{\theta^{tr}}c\theta}\over{2}}\is{A}-\widetilde{K}.
\end{split}
\end{equation*}
By applying Lemma \ref{Integrability4ThetaProcesses} to $Y=\widetilde{X}$ (i.e. (\ref{Decomposition4Theta(X)}), we put $\widetilde{b}^{\theta}:=\widetilde{b}-\int h_{\epsilon}(x)I_{\{\vert\theta^{tr}x\vert>\epsilon\}}$ and we get
\begin{equation}
\begin{split}
Z_{-}^{-1}\is{Z}&=\theta\is{X^c}+\theta^{tr}x{I}_{\{\vert\theta^{tr}x\vert\leq\epsilon \}}\star(\widetilde\mu-\widetilde\nu)+\left(\theta^{tr}b^{\theta}+\int \theta^{tr}xI_{\{\vert\theta^{tr}x\vert\leq\epsilon<\vert{x}\vert\}}\widetilde{F}(dx)\right)\is{A}\\
&+\theta^{tr}x{I}_{\{\vert\theta^{tr}x\vert>\epsilon\}}\star\widetilde\mu+\left(\exp(\theta^{tr}{x}+{x}^{tr}\psi{{x}})-1-\theta^{tr}{x}\right)\star\widetilde{\mu}+{{{\theta^{tr}}c\theta}\over{2}}\is{A}-\widetilde{K}\\
&=\theta\is{X^c}+\theta^{tr}x{I}_{\{\vert\theta^{tr}x\vert\leq\epsilon \}}\star(\widetilde\mu-\widetilde\nu)+\left(\exp(\theta^{tr}{x}+{x}^{tr}\psi{{x}})-1-\theta^{tr}{x}{I}_{\{\vert\theta^{tr}x\vert\leq\epsilon\}}\right)\star\widetilde{\mu}\\
&+\left({{{\theta^{tr}}c\theta}\over{2}}+\theta^{tr}\widetilde{b}^{\theta}+\int \theta^{tr}xI_{\{\vert\theta^{tr}x\vert\leq\epsilon<\vert{x}\vert\}}\widetilde{F}(dx)\right)\is{A}-\widetilde{K}.
\end{split}
\end{equation}
Then, on the one hand, $Z$ is a local martingale if and only if 
\begin{equation}\label{Integrability1Linear}
 \left(\exp(\theta^{tr}{x}+{x}^{tr}\psi{{x}})-1-\theta^{tr}x{I}_{\{\vert\theta^{tr}x\vert\leq\epsilon\}}\right)\star\widetilde{\mu}\in{\cal{A}},\end{equation}
 and 
 \begin{equation}\label{MaringaleFormLinear}
 \begin{split}
 \widetilde{K}&=\left({{{\theta^{tr}}c\theta}\over{2}}+\theta^{tr}\widetilde{b}^{\theta}+\int \theta^{tr}xI_{\{\vert\theta^{tr}x\vert\leq\epsilon<\vert{x}\vert\}}\widetilde{F}(dx)\right)\is{A} +\left(\exp(\theta^{tr}{x}+{x}^{tr}\psi{{x}})-1-\theta^{tr}x{I}_{\{\vert\theta^{tr}x\vert\leq\epsilon\}}\right)\star\widetilde{\nu}\\
 &=\left({{{\theta^{tr}}c\theta}\over{2}}+\theta^{tr}\widetilde{b}\right)\is{A}+\left(\exp(\theta^{tr}{x}+{x}^{tr}\psi{{x}})-1-\theta^{tr}h_{\epsilon}(x)\right)\star\widetilde{\nu},\\
 &Z={\cal{E}}\left(\theta\is{X}^c+\left(\exp(\theta^{tr}{x}+{x}^{tr}\psi{{x}})-1\right)\star(\widetilde{\mu}-\widetilde{\nu})\right)=:{\cal{E}}(\widetilde{N}).
 \end{split}
 \end{equation}
 This proves (\ref{Zequation4Linear}). On the other hand, thanks to (\ref{Canonical4Xtilde}) and for $\delta\in(0,1)$, we derive 
\begin{equation*}
\begin{split}
\widetilde{X}+[\widetilde{X},\widetilde{N}]=&X^c+h_{\delta}(x)\star(\widetilde{\mu}-\widetilde{\nu})+\widetilde{b}\is{A}+(x-h_{\delta}(x))\star\widetilde{\mu}+c\theta\is{A}\\
&+\left(\exp(\theta^{tr}{x}+{x}^{tr}\psi{{x}})-1\right)x\star\widetilde{\mu}\\
=&X^c+h_{\delta}(x)\star(\widetilde{\mu}-\widetilde{\nu})+(\widetilde{b}+c\theta)\is{A}+\left(\exp(\theta^{tr}{x}+{x}^{tr}\psi{{x}})x-h_{\delta}(x)\right)\star\widetilde{\mu}.\\
\end{split}
\end{equation*}
Thus, $ZS$ is a local martingale if and only if $\widetilde{X}+[\widetilde{X},N]$ is a local martingale if and only if 
\begin{equation}\label{Integrability2Linear}
\left(\exp(\theta^{tr}{x}+{x}^{tr}\psi{{x}})x-h_{\delta}(x)\right)\star\widetilde{\mu}\in{\cal{A}},
\end{equation}
 and (\ref{mgEquation4Linear}) holds. Furthermore, by combining (\ref{Integrability2Linear}), (\ref{Integrability1Linear}) and Lemma \ref{TechnicalLemma2}-(a), we obtain the condition (\ref{IntegrabilityCondition4Linear}). This proves that $Z\in{\cal{Z}}^{LE}(S)$ implies the existence of $(\theta,\psi)\in\Theta(S)$ satisfying (\ref{IntegrabilityCondition4Linear})-(\ref{Zequation4Linear})- (\ref{mgEquation4Linear}). To prove the reverse, we assume the existence of such a pair $(\theta,\psi)$, and thanks to It\^o, we derive $Z={\cal{E}}(\widetilde{N})=\exp(L)$, where
 \begin{equation}\label{Equa299}
 \begin{split}
 L=&\widetilde{N}+{1\over{2}}\langle\widetilde{N}^c\rangle+\sum\left(\ln(1+\Delta\widetilde{N})-\Delta\widetilde{N}\right)\\
 =&\theta\is{X}^c+\left(e^{\theta^{tr}{x}+{x}^{tr}\psi{{x}}}-1\right)\star(\widetilde{\mu}-\widetilde{\nu})+{1\over{2}}\theta^{tr}c\theta\is{A}+\left(\theta^{tr}{x}+{x}^{tr}\psi{{x}}-e^{\theta^{tr}{x}+{x}^{tr}\psi{{x}}}+1\right)\widetilde\mu\\
 =&\theta\is{X}^c+\left(e^{\theta^{tr}{x}+{x}^{tr}\psi{{x}}}-1\right)\star(\widetilde{\mu}-\widetilde{\nu})+\left(\theta^{tr}h_{\epsilon}(x)-e^{\theta^{tr}{x}+{x}^{tr}\psi{{x}}}+1\right)\star\widetilde\mu\\
 &+\left(\theta^{tr}{x}I_{\{\Vert{x}\Vert>\epsilon\}}+{x}^{tr}\psi{{x}}\right)\widetilde\mu.
 \end{split}
 \end{equation}
 Thanks to Lemma \ref{TechnicalLemma2}-(b), the condition (\ref{IntegrabilityCondition4Linear}) implies that 
 $\left(\theta^{tr}h_{\epsilon}(x)-e^{\theta^{tr}{x}+{x}^{tr}\psi{{x}}}+1\right)\star\widetilde\mu\in{\cal{A}}_{loc}.$
 Thus, by inserting the compensator of this process in (\ref{Equa299}), arranging terms, and inserting (\ref{mgEquation4Linear}) in the resulting equation afterwards, we obtain 
  \begin{equation*}
 \begin{split}
 L&=\theta\is{X}^c+\theta^{tr}h_{\epsilon}(x)\star(\widetilde{\mu}-\widetilde{\nu})+\left(\theta^{tr}h_{\epsilon}(x)-e^{\theta^{tr}{x}+{x}^{tr}\psi{{x}}}+1\right)\star\widetilde\nu+\left(\theta^{tr}{x}I_{\{\Vert{x}\Vert>\epsilon\}}+{x}^{tr}\psi{{x}}\right)\star\widetilde\mu\\
 &=\theta\is{X}^c+\theta^{tr}h_{\epsilon}(x)\star(\widetilde{\mu}-\widetilde{\nu})+\theta^{tr}\widetilde{b}\is{A}+\theta^{tr}{x}I_{\{\Vert{x}\Vert>\epsilon\}}\star\widetilde\mu\\
 &\quad+\theta^{tr}c\theta\is{A}+\left(\theta^{tr}x\exp(\theta^{tr}x+x^{tr}\psi{x})-\exp(\theta^{tr}x+x^{tr}\psi{x})+1\right)\star\widetilde\nu+\sum(\Delta\widetilde{X})^{tr}\psi\Delta\widetilde{X}\\
 &=\theta\is\widetilde{X}+\sum(\Delta\widetilde{X})^{tr}\psi\Delta\widetilde{X}-\widetilde{K},
  \end{split}
 \end{equation*}
 where $-\widetilde{K}:=\theta^{tr}c\theta\is{A}+\left(\theta^{tr}x\exp(\theta^{tr}x+x^{tr}\psi{x})-\exp(\theta^{tr}x+x^{tr}\psi{x})+1\right)\star\widetilde\nu$ is predictable and belongs to ${\cal{A}}_{loc}$. This ends the proof of the theorem.\end{proof}
\begin{proof}[Proof of Theorem  \ref{mainTHM4LinearEsscherBIS}]  Remark that (c) $\Longrightarrow$ (d) is obvious. Hence the rest of this proof if divided into three parts, where we prove (a) $\Longrightarrow$ (b), (b) $\Longrightarrow$ (c) and (d) $\Longrightarrow$ (a) respectively.\\
{\bf Part 1.} Here we prove the implication (a) $\Longrightarrow$ (b). Thus, we suppose that assertion (a) holds, and in virtue of Theorem  \ref{mainTHM4LinearEsscher} we obtain the existence of $(\theta,\psi)\in\Theta(S)$ such that (\ref{IntegrabilityCondition4Linear}) holds and 
\begin{equation*}
Z=\mathcal{E}\Bigl(\theta \is X^c + \left(\exp(\theta^{tr} x+x^{tr} \psi x)-1\right)\star (\widetilde{\mu}-\widetilde{\nu}) \Bigr)\in{\cal{Z}}_{loc}(S).
\end{equation*}
Thus, assertion (b) follows as soon as we prove that $Z\ln(Z)$ is a special semimartingale. Thanks to \cite[Lemma 3.2 and Proposition 3.5]{Choulli2005}, we have
\begin{equation*}
\begin{split}
&Z\ln(Z)=(\ln(Z_{-})+1)\is\widetilde{N}+Z_{-}\is{H}^E(Z,P),\quad\mbox{and}\\
&{H}^E(Z,P)={1\over{2}}\theta^{tr}c\theta\is{A}+\left((\theta^{tr}x+x^{tr}\psi{x})e^{\theta^{tr}x+x^{tr}\psi{x}}-e^{\theta^{tr}x+x^{tr}\psi{x}}+1 \right)\star\widetilde\mu.
\end{split}
\end{equation*}
Then $Z\ln(Z)$ is a special semimartingale if and only if 
\begin{equation}\label{Condition4Zlog(Z)}
\left((\theta^{tr}x+x^{tr}\psi{x})e^{\theta^{tr}x+x^{tr}\psi{x}}-e^{\theta^{tr}x+x^{tr}\psi{x}}+1 \right)\star\widetilde\mu\in{\cal{A}}_{loc}^+.
\end{equation}
On the one hand, thanks to Lemma \ref{TechnicalLemma2}-(c), we deduce that 
\begin{equation*}
\vert\theta^{tr}x{e}^{\theta^{tr}x+x^{tr}\psi{x}}-e^{\theta^{tr}x+x^{tr}\psi{x}}+1\vert\star\widetilde\mu\in{\cal{A}}_{loc}^+
\end{equation*}
holds due to (\ref{IntegrabilityCondition4Linear}),  (\ref{mgEquation4Linear})  and the local boundedness of $\widetilde{X}$ which implies that $\vert{1}-e^{x^{tr}\psi{x}}\vert\star\widetilde\mu\in{\cal{A}}^+_{loc}$. On the other hand, in virtue of Remark \ref{Condition(3.6)forBoundedS} and the local boundedness of $\widetilde{X}$ again, we derive 
\begin{equation*}
\begin{split}
\vert{x}^{tr}\psi{x}\vert{e}^{\theta^{tr}x+x^{tr}\psi{x}}\star\widetilde\mu&=\vert{x}^{tr}\psi{x}\vert{e}^{\theta^{tr}x+x^{tr}\psi{x}}I_{\{\vert\theta^{tr}x\vert\leq\epsilon\}}\star\widetilde\mu+\vert{x}^{tr}\psi{x}\vert{e}^{\theta^{tr}x+x^{tr}\psi{x}}I_{\{\vert\theta^{tr}x\vert>\epsilon\}}\star\widetilde\mu\\
&\leq e^{\epsilon}\sum\vert(\Delta\widetilde{X})^{tr}\psi\Delta\widetilde{X}\vert{e}^{(\Delta\widetilde{X})^{tr}\psi\Delta\widetilde{X}}+\vert(\Delta\widetilde{X})^{tr}\psi\Delta\widetilde{X}\vert{e}^{(\Delta\widetilde{X})^{tr}\psi\Delta\widetilde{X}}\is{U}_{\epsilon}\in{\cal{A}}_{loc}^+,
\end{split}
\end{equation*}
where $U_{\epsilon}$ is defined in (\ref{SimplifiedIntreCondi}). This proves that $Z\in{\cal{Z}}^{L\log{L}}_{loc}(S)$, and assertion (b) follows immediately. \\
{\bf Part 2.} Hereto we prove (b) $\Longrightarrow$ (c). To this end, we suppose that assertion (b) holds, and consider  a predictable and locally bounded process $\psi$. Then define 
\begin{equation*}
\begin{split}
f(\theta)&:=\theta^{tr}\widetilde{b}'+{1\over{2}}\theta^{tr}c\theta+\int\left( e^{\theta^{tr}x+x^{tr}\psi{x}}-1-\theta^{tr}x\right)\widetilde{F}(dx)\\
f_{\psi}(\theta)&:=\theta^{tr}b^{\psi}+{1\over{2}}\theta^{tr}c\theta+\int\left( e^{\theta^{tr}x}-1-\theta^{tr}x\right)\widetilde{F}^{\psi}(dx)\\
\widetilde{F}^{\psi}(dx)&:=e^{x^{tr}\psi{x}}\widetilde{F}(dx),\quad b^{\psi}:=\widetilde{b}'+\int{x}\left(e^{x^{tr}\psi{x}}-1\right)\widetilde{F}(dx),\quad \widetilde\Gamma^{\psi}:=\int\left(e^{x^{tr}\psi{x}}-1\right)\widetilde{F}(dx)
\end{split}
\end{equation*}
Then, direct calculation shows that $ f(\theta)=f_{\psi}(\theta)+\widetilde\Gamma^{\psi}$, and hence
\begin{equation}\label{equivalence4Optimisation}
\min_{\theta}f(\theta)\quad \mbox{has a solution}\quad\mbox{if and only if}\quad \min_{\theta}f_{\psi}(\theta)\quad \mbox{has a solution}.
\end{equation}
By stopping, we assume without loss of generality that 
\begin{equation}
Z^{\psi}\in{\cal{M}}\quad \mbox{and}\quad Q^{\psi}:=Z^{\psi}_T\cdot P\quad\mbox{is a well defined probability measure}.
\end{equation}
Hence, under $Q^{\psi}$, the process $\widetilde{X}$ admits the following canonical decomposition
 \begin{equation}
 \widetilde{X}=b^{\psi}\is A+x\star(\widetilde\mu-\widetilde\nu^{\psi})+X^c,\quad \widetilde\nu^{\psi}(dt,dx):=e^{x^{tr}\psi{x}}\widetilde\nu(dt,dx)=\widetilde{F}^{\psi}_t(dx)dA_t.
 \end{equation}
 Then, in virtue of the local boundedness of $Z^{\psi}$ and $1/Z^{\psi}$, it is clear that
 \begin{equation}\label{Equivalence4Z(LlogL)}
 {\cal{Z}}^{L\log{L}}_{loc}(S)\not=\emptyset\ \Longleftrightarrow\  {\cal{Z}}^{L\log{L}}_{loc}(S,Q^{\psi})\not=\emptyset,\ \forall\ \psi\ \mbox{predictable and locally bounded}.
  \end{equation}
  Therefore, in virtue of (\ref{Equivalence4Z(LlogL)}) and (\ref{equivalence4Optimisation}), our minimization problem of (\ref{MinimizationProblem}) has a solution if and only if  the minimization problem 
  \begin{equation}\label{minimisationwithpsi}
   \min_{\theta}f_{\psi}(\theta),
   \end{equation}has a solution under the condition ${\cal{Z}}^{L\log{L}}_{loc}(S,Q^{\psi})\not=\emptyset$. Thus, the resulting problem is exactly the minimization problem considered in \cite[Lemma 4.4]{Choulli2005} for the model $(\widetilde{X},Q^{\psi})$ instead, and hence the existence of a unique solution $\widehat\theta^{\psi}=:\widehat\theta$ is guaranteed. Thus, assertion (c) will follow as soon as we prove that this solution $\widehat\theta$ satisfies 
   \begin{equation}\label{Claim2prove}
   {\widehat\theta}^{tr}c{\widehat\theta}\is{A}+f_{L\log{L}}(\exp({\widehat\theta}^{tr}x)-1)\star\widetilde\nu^{\psi}\in{\cal{A}}_{loc}(Q^{\psi}).
   \end{equation}
To prove this latter fact, we use Lemma \ref{Minimization2Root}-(a) for $Q=Q^{\psi}$ and conclude that $\widehat\theta$, the solution  to (\ref{minimisationwithpsi}), satisfies
   \begin{equation}\label{MgEquation4psi}
   b^{\psi}+c\widehat\theta+\int\left(xe^{\widehat\theta^{tr}x}-x\right)\widetilde{F}^{\psi}(dx)=0.
   \end{equation}
   Besides this, a combination of $ {\cal{Z}}^{L\log{L}}_{loc}(S,Q^{\psi})\not=\emptyset$ and Theorem \ref{Characgteristics4Deflator}  yields the existence of a pair $(\beta, f)$ such that
   \begin{equation*}
     b^{\psi}+c\beta+\int{x}\left(f(x)-1)\right)\widetilde{F}^{\psi}(dx)=0,\ \mbox{and}\ \beta^{tr}c\beta\is{A}+\left(f\ln(f)-f+1\right)\star\widetilde\nu^{\psi}\in{\cal{A}}^+_{loc}(Q^{\psi}).
   \end{equation*}
   Thus, by combining the first equality above with (\ref{MgEquation4psi}), we get 
   \begin{equation}\label{MgEquationDifference}
   c(\beta-\widehat\theta)+\int{x}\left(f(x)-e^{\widehat\theta^{tr}x}\right)\widetilde{F}^{\psi}(dx)=0.
   \end{equation}
   Thanks to the convexity of $\theta^{tr}c\theta$ and $g(y):=y\ln(y)-y+1$ for $y>0$, we get 
    \begin{equation*}
    \begin{split}
   \beta^{tr}c\beta\geq    {\widehat\theta}^{tr}c{\widehat\theta}+(\beta-\widehat\theta)^{tr}c{\widehat\theta}\ \mbox{and}\quad
   g(f(x))- g(e^{{\widehat\theta}^{tr}x})\geq {\widehat\theta}^{tr}x(f(x)-e^{{\widehat\theta}^{tr}x}).
   \end{split}
   \end{equation*}
    By integrating the two inequalities above with respect to $A$ and $\widetilde\nu^{\psi}$ respectively, adding them and using (\ref{MgEquationDifference}) afterwards, we get 
        \begin{equation*}
        \beta^{tr}c\beta\is{A}+\left(f\ln(f)-f+1\right)\star\widetilde\nu^{\psi}\geq  {\widehat\theta}^{tr}c{\widehat\theta}\is{A}+f_{L\log{L}}(\exp({\widehat\theta}^{tr}x)-1)\star\widetilde\nu^{\psi},
           \end{equation*}
           and the claim (\ref{Claim2prove}) follows immediately. This proves assertion (c).\\
{\bf Part 3.} This part proves (c) $\Longrightarrow$ (a). Thus, we suppose that assertion (c) holds, and remark that we always have $f_{L\log{L}}(y):=(y+1)\ln(1+y)-y\geq \vert{y}\vert{I}_{\{y>2e-1\}}+{{y^2}\over{4e}}I_{\{\vert{y}\vert\leq 2e-1\}}$ for any $y>-1$. Thus, by combining this with the condition (\ref{Integrability4LinearExponential}) and  \cite[Proposition B.2-(a)]{Choulli2018}, which states that  $\sqrt{(e^{\widehat\theta^{tr}x+x^{tr}\psi{x}}-1)^2\star\widetilde\mu}\in{\cal{A}}_{loc}^+$ if and only if for any $\alpha>0$
$$(e^{\widehat\theta^{tr}x+x^{tr}\psi{x}}-1)^2I_{\{ \vert{e}^{\widehat\theta^{tr}x+x^{tr}\psi{x}}-1\vert\leq \alpha\}}\star\widetilde\mu+ \vert{e}^{\widehat\theta^{tr}x+x^{tr}\psi{x}}-1\vert{I}_{\{ \vert{e}^{\widehat\theta^{tr}x+x^{tr}\psi{x}}-1\vert>\alpha\}}\star\widetilde\mu\in{\cal{A}}_{loc}^+,$$
 we deduce that 
$$
\sqrt{(e^{\widehat\theta^{tr}x+x^{tr}\psi{x}}-1)^2\star\widetilde\mu}\in{\cal{A}}_{loc}^+\quad\mbox{and}\quad \widehat\theta\in L(X^c).$$
Thus, as a result, by combining this with Lemma \ref{Minimization2Root}, we deduce that following process
\begin{equation}\label{Zexpression}
Z:={\cal{E}}\left(\widehat\theta\is{X}^c+(e^{\widehat\theta^{tr}x+x^{tr}\psi{x}}-1)\star(\widetilde\mu-\widetilde\nu)\right)\in {\cal{Z}}_{loc}^{L\log{L}}(S).\end{equation}
Now, thanks to It\^o, we calculate the semimartingale $\ln(Z)$ as follows
\begin{equation}\label{Ito100}
\begin{split}
\ln(Z)=&\widehat\theta\is{X}^c+(e^{\widehat\theta^{tr}x+x^{tr}\psi{x}}-1)\star(\widetilde\mu-\widetilde\nu)-{1\over{2}}\widehat\theta^{tr}c\widehat\theta\is A\\
&+\left(\widehat\theta^{tr}x+x^{tr}\psi{x}-e^{\widehat\theta^{tr}x+x^{tr}\psi{x}}+1\right)\star\widetilde\mu\\
=&\widehat\theta\is{X}^c+(e^{\widehat\theta^{tr}x+x^{tr}\psi{x}}-1)\star(\widetilde\mu-\widetilde\nu)-{1\over{2}}\widehat\theta^{tr}c\widehat\theta\is A\\
&+\left(\widehat\theta^{tr}x-e^{\widehat\theta^{tr}x+x^{tr}\psi{x}}+1\right)\star\widetilde\mu+\sum (\Delta\widetilde{X})^{tr}\psi\Delta\widetilde{X}.
\end{split}
\end{equation}
Remark that for any $n$, it is clear that $I_{\{\Vert\widehat\theta\Vert\leq n\}}\vert\widehat\theta^{tr}x-e^{\widehat\theta^{tr}x+x^{tr}\psi{x}}+1\vert\star\widetilde\mu\in{\cal{A}}_{loc}^+$. Thus, by compensating this process in (\ref{Ito100}) and using (\ref{MgEquation4psi}) afterwards, we derive
\begin{equation}\label{log(Z)}
\begin{split}
&I_{\{\Vert\widehat\theta\Vert\leq n\}}\is\ln(Z)\\
&=I_{\{\Vert\widehat\theta\Vert\leq n\}}\widehat\theta\is{X}^c+I_{\{\Vert\widehat\theta\Vert\leq n\}}\widehat\theta^{tr}x\star(\widetilde\mu-\widetilde\nu)-{1\over{2}}I_{\{\Vert\widehat\theta\Vert\leq n\}}\widehat\theta^{tr}c\widehat\theta\is A\\
&\quad+I_{\{\Vert\widehat\theta\Vert\leq n\}}\left(\widehat\theta^{tr}x-e^{\widehat\theta^{tr}x+x^{tr}\psi{x}}+1\right)\star\widetilde\nu+\sum I_{\{\Vert\widehat\theta\Vert\leq n\}}(\Delta\widetilde{X})^{tr}\psi\Delta\widetilde{X}\\
&=I_{\{\Vert\widehat\theta\Vert\leq n\}}\widehat\theta\is{X}^c+I_{\{\Vert\widehat\theta\Vert\leq n\}}\widehat\theta^{tr}x\star(\widetilde\mu-\widetilde\nu)-{1\over{2}}I_{\{\Vert\widehat\theta\Vert\leq n\}}\widehat\theta^{tr}c\widehat\theta\is A+\sum I_{\{\Vert\widehat\theta\Vert\leq n\}}(\Delta\widetilde{X})^{tr}\psi\Delta\widetilde{X}\\
&\quad+I_{\{\Vert\widehat\theta\Vert\leq n\}}\left(\widehat\theta^{tr}x(1-e^{\widehat\theta^{tr}x+x^{tr}\psi{x}})+ (\widehat\theta^{tr}x)e^{\widehat\theta^{tr}x+x^{tr}\psi{x}}-e^{\widehat\theta^{tr}x+x^{tr}\psi{x}}+1\right)\star\widetilde\nu\\
&=I_{\{\Vert\widehat\theta\Vert\leq n\}}\widehat\theta\is{X}^c+I_{\{\Vert\widehat\theta\Vert\leq n\}}\widehat\theta^{tr}x\star(\widetilde\mu-\widetilde\nu)+I_{\{\Vert\widehat\theta\Vert\leq n\}}\widehat\theta^{tr}\widetilde{b}'\is{A}+\sum I_{\{\Vert\widehat\theta\Vert\leq n\}}(\Delta\widetilde{X})^{tr}\psi\Delta\widetilde{X}\\
&\quad+I_{\{\Vert\widehat\theta\Vert\leq n\}}\left((\widehat\theta^{tr}x)e^{\widehat\theta^{tr}x+x^{tr}\psi{x}}-e^{\widehat\theta^{tr}x+x^{tr}\psi{x}}+1\right)\star\widetilde\nu+{1\over{2}}I_{\{\Vert\widehat\theta\Vert\leq n\}}\widehat\theta^{tr}c\widehat\theta\is A\\
&=I_{\{\Vert\widehat\theta\Vert\leq n\}}\widehat\theta\is\widetilde{X}+\sum I_{\{\Vert\widehat\theta\Vert\leq n\}}(\Delta\widetilde{X})^{tr}\psi\Delta\widetilde{X}+I_{\{\Vert\widehat\theta\Vert\leq n\}}\is\widetilde{K},
\end{split}
\end{equation}
where $\widetilde{K}$ is given by 
$$
\widetilde{K}:=\left((\widehat\theta^{tr}x)e^{\widehat\theta^{tr}x+x^{tr}\psi{x}}-e^{\widehat\theta^{tr}x+x^{tr}\psi{x}}+1\right)\star\widetilde\nu+{1\over{2}}\widehat\theta^{tr}c\widehat\theta\is A,
$$
and clearly is a well defined predictable with finite variation process due to (\ref{Integrability4LinearExponential}) and Remark \ref{Remark4SlocallyBoundedCase}-(b). \\
Now, as the three processes $\ln(Z)$, $\widetilde{K}$ and  $V:=\sum(\Delta\widetilde{X})^{tr}\psi\Delta\widetilde{X}$ are semimartingales, then thanks to \cite[Proposition 1.7]{Stricker}, it is clear that the three processes $I_{\{\Vert\widehat\theta\Vert\leq n\}}\is\ln(Z)$, $I_{\{\Vert\widehat\theta\Vert\leq n\}}\is\widetilde{K}$ and $I_{\{\Vert\widehat\theta\Vert\leq n\}}\is{V}$ converge in the semimartingale topology to $\ln(Z)$, $\widetilde{K}$ and $V$ respectively.  This implies that $I_{\{\Vert\widehat\theta\Vert\leq n\}}\widehat\theta\is\widetilde{X}$ converges in the semimartingale topology. Thus,  in virtue of \cite[Definition 2.1]{Stricker}, we deduce that $\widehat\theta\in{L}(\widetilde{X})$, and its limits is the semimartingale $\widehat\theta\is\widetilde{X}$. Furthermore, by combining these latter remarks and (\ref {log(Z)}), we obtain 
$$
\ln(Z)=\widehat\theta\is\widetilde{X}+\sum(\Delta\widetilde{X})^{tr}\psi\Delta\widetilde{X}+\widetilde{K}.$$
Therefore, assertion (a) follows immediately (i.e.  $Z\in{\cal{Z}}^{LE}(S)$) from combining the above equality with (\ref{Zexpression}). This ends the proof of the theorem. 
\end{proof}
\begin{remark}\label{ImportantRemark4proof}
The last part of the proof above contains an important statement. This confirms the following claim: If $S$ is locally bounded and there exists a pair $(\theta,\psi)$ of predictable processes such that
 \begin{equation}\label{SetofConditions}
\psi\quad\mbox{is locally bounded, (\ref{Integrability4LinearExponential}) holds, and $\theta$ is the solution to (\ref{mgEquation4Linear})},
 \end{equation}
 then $(\theta,\psi)\in\Theta(S)$.
\end{remark}
\begin{proof}[Proof of Theorem \ref{Charaterization4LinearEsscher}] Thanks to Theorem \ref{mainTHM4LinearEsscherBIS} (assertions (a) and (d)), we deduce that $Z\in{\cal{Z}}^{LE}(S)$ if and only if there exists a pair $(\theta,\psi)$ of predictable processes such that 
\begin{equation}\label{1001}
\begin{split}
&\psi\quad\mbox{is locally bounded, (\ref{Integrability4LinearExponential}) holds, and $\theta$ is the solution to (\ref{MinimizationProblem})}\\
&Z={\cal{E}}\left(\theta\is{X}^c+(\exp(\theta^{x}+x^{tr}\psi{x})-1)\star(\widetilde\mu-\widetilde\nu)\right).
\end{split}
\end{equation}
In virtue of Lemma \ref{Minimization2Root}, we have $\theta$ is the solution to (\ref{MinimizationProblem}) if and only if $\theta$ is the solution to (\ref{mgEquation4Linear}) if and only if $Z\in{\cal{Z}}_{loc}(S)$.  Furthermore, direct calculations shows easily that $Z\in{\cal{Z}}_{loc}(S)$ if and only if $Z/Z^{\psi}\in{\cal{Z}}_{loc}(S,Z^{\psi})$. Thus, by combining all these remarks, we conclude that $Z\in{\cal{Z}}^{LE}(S)$ if and only if there exists a pair $(\theta,\psi)$ of predictable processes satisfying 
\begin{equation*}
\begin{split}
\psi\quad\mbox{ is locally bounded, (\ref{Integrability4LinearExponential})  holds and}\ Z/Z^{\psi}\in{\cal{Z}}_{loc}(S,Z^{\psi}).
\end{split}
\end{equation*}
This proves assertion (a).  this proof addresses assertion (b). On the one hand, as $\psi$ and $S$ are locally bounded, we remark that $(\exp(x^{tr}\psi{x})-1)\star(\widetilde\mu-\widetilde\nu)$ (or equivalently $Z^{(\psi)}$) is a well define local martingale. Hence,  using Yor's formula  (i.e. ${\cal{E}}(Y_1){\cal{E}}(Y_2)={\cal{E}}(Y_1+Y_2+[Y_1,Y_2])$ for any pair of semimartingales $Y_1$ and $Y_2$), we derive
\begin{equation}\label{Equa300}
\begin{split}
&Z^{(\psi)}Z^{(\theta,\psi)}:=Z^{(\psi)}{\cal{E}}\left(\theta\is{X}^c+(e^{\theta^{tr}x}-1)\star(\widetilde{\mu}-\widetilde{\nu}^{(\psi)})\right)\\
&={\cal{E}}\left((e^{x^{tr}\psi{x}}-1)\star(\widetilde{\mu}-\widetilde{\nu})+\theta\is{X}^c+(e^{\theta^{tr}x}-1)\star(\widetilde{\mu}-\widetilde{\nu}^{(\psi)})+(e^{x^{tr}\psi{x}}-1)(e^{\theta^{tr}x}-1)\star\widetilde{\mu}\right)\\
&={\cal{E}}\left(\theta\is{X}^c+ (e^{x^{tr}\psi{x}}-1)\star(\widetilde{\mu}-\widetilde{\nu})+(e^{\theta^{tr}x}-1)\star(\widetilde{\mu}-\widetilde{\nu}^{(\psi)})+(e^{x^{tr}\psi{x}}-1)(e^{\theta^{tr}x}-1)\star\widetilde{\mu}\right)\\
&={\cal{E}}\left(\theta\is{X}^c+ (\exp(\theta^{tr}x+x^{tr}\psi{x})-1)\star(\widetilde{\mu}-\widetilde{\nu})\right)=Z.
\end{split}
\end{equation}
The last equality follows from the fact that $$(e^{\theta^{tr}x}-1)\star(\widetilde{\mu}-\widetilde{\nu}^{(\psi)})=(e^{\theta^{tr}x}-1)\star(\widetilde{\mu}-\widetilde{\nu})-(e^{x^{tr}\psi{x}}-1)(e^{\theta^{tr}x}-1)\star\widetilde{\nu}.$$ 
On the other hand, by combining Theorem  \ref{mainTHM4LinearEsscherBIS}, applied to $S$ under $Z^{\psi'}$ ($\psi'$ predictable locally bounded) and the first-order-Esscher (i.e. $\psi=0$), and assertion (a), we deduce that ${\cal{Z}}^{LE}(S)\not=\emptyset$ if and only if $S$ admits the first-order-Esscher density under $Z^{\psi'}$, and this is equivalent to  if and only if  ${\cal{Z}}_{loc}^{L\log{L}}(S,Z^{\psi'})\not=\emptyset$. Thus, by combining this latter fact with (\ref{Equa300}), assertion (b) follows immediately, and the proof of the theorem is complete.
\end{proof}
\subsubsection{Proof of Theorems   \ref{mainTHM4ExponenialEsscher} and \ref{OneDimension}}
We start with the proof of Theorem \ref{mainTHM4ExponenialEsscher}.
\begin{proof}[Proof of Theorem \ref{mainTHM4ExponenialEsscher}] The proof of the theorem is given in three parts, where we prove assertions (a), (b) and (c) respectively.\\
{\bf Part 1.} In this part, we prove assertion (a). To this end, we start by  putting
$$\overline{X}:=\theta\is{X}+\sum(\Delta{X})^{tr}\psi\Delta{X}-K,\quad Z:=\exp(Y).$$ Then, due to It\^o,  we derive 
\begin{equation*}
\begin{split}Z_{-}^{-1}\is{Z}&=\overline{X}+{1\over{2}}\langle{\overline{X}}^c\rangle+\sum\left(e^{\Delta\overline{X}}-1-\Delta\overline{X}\right)\\
&=\theta\is{X}+\sum(\Delta{X})^{tr}\psi\Delta{X}-K+{{{\theta^{tr}}c\theta}\over{2}}\is A+\left(\exp(\theta^{tr}x+x^{tr}\psi{x})-1-\theta^{tr}x-x^{tr}\psi{x}\right)\star\mu\\
&=\theta\is{X}-K+{{{\theta^{tr}}c\theta}\over{2}}\is A+\left(\exp(\theta^{tr}x+x^{tr}\psi{x})-1-\theta^{tr}x\right)\star\mu.
\end{split}
\end{equation*}
Thus, by using (\ref{Decomposition4Theta(X)}) to $Y=X$ and inserting it in the above equality, see Lemma \ref{Integrability4ThetaProcesses}-(c), we derive
\begin{equation*}
\begin{split}
Z_{-}^{-1}\is{Z}=&\theta\is{X^c}+\theta^{tr}x{I}_{\{\vert\theta^{tr}x\vert\leq\epsilon \}}\star(\mu-\nu)+\left(\theta^{tr}b^{\theta}+\int \theta^{tr}xI_{\{\vert\theta^{tr}x\vert\leq\epsilon<\vert{x}\vert\}}F(dx)\right)\is A\\
&+\theta\is{U}_2^{\theta}-K+{{{\theta^{tr}}c\theta}\over{2}}\is A+\left(\exp(\theta^{tr}x+x^{tr}\psi{x})-1-\theta^{tr}x\right)\star\mu\\
=&\theta\is{X^c}+\theta^{tr}x{I}_{\{\vert\theta^{tr}x\vert\leq\epsilon \}}\star(\mu-\nu)-K+\left({{{\theta^{tr}}c\theta}\over{2}}+\theta^{tr}b^{\theta}+\int \theta^{tr}xI_{\{\vert\theta^{tr}x\vert\leq\epsilon<\vert{x}\vert\}}F(dx)\right)\is A\\
&+\left(\exp(\theta^{tr}x+x^{tr}\psi{x})-1-\theta^{tr}x{I}_{\{\vert\theta^{tr}x\vert\leq\epsilon \}}\right)\star\mu.
\end{split}
\end{equation*}
Therefore, on the one hand, $Z$ is a local martingale if and only if 
\begin{equation}\label{Integrability1}
 \left(\exp(\theta^{tr}x+x^{tr}\psi{x})-1-h_{\epsilon}(\theta^{tr}x)\right)\star\mu\in{\cal{A}}_{loc},\end{equation}
 and 
 \begin{equation}\label{MaringaleForm}
 \begin{split}
 &K=\left({{{\theta^{tr}}c\theta}\over{2}}+\theta^{tr}b^{\theta}+\int \theta^{tr}xI_{\{\vert\theta^{tr}x\vert\leq\epsilon<\vert{x}\vert\}}F(dx)\right)\is{A} +\left(\exp(\theta^{tr}x+x^{tr}\psi{x})-1-h_{\epsilon}(\theta^{tr}x)\right)\star\nu\\
 &Z={\cal{E}}\left(\theta\is{X}^c+\left(\exp(\theta^{tr}x+x^{tr}\psi{x})-1\right)\star(\mu-\nu)\right):={\cal{E}}(N).
 \end{split}
 \end{equation}
 This proves (\ref{ZFormEquation4Exponential}). On the other hand, thanks to (\ref{Canonical4Xtilde}), we derive 
\begin{equation*}
\begin{split}
\widetilde{X}+[\widetilde{X},N]=&X^c+({\bf{e}}(x)-\mathbb{I}_d)I_{\{\vert{x}\vert\leq\epsilon\}}\star(\mu-\nu)+\widetilde{b}\is{A}+({\bf{e}}(x)-\mathbb{I}_d)I_{\{\vert{x}\vert>\epsilon\}}\star\mu+c\theta\is{A}\\
&+\left(\exp(\theta^{tr}x+x^{tr}\psi{x})-1\right)({\bf{e}}(x)-\mathbb{I}_d)\star\mu\\
=&X^c+({\bf{e}}(x)-\mathbb{I}_d)I_{\{\vert{x}\vert\leq\epsilon\}}\star(\mu-\nu)+(\widetilde{b}+c\theta)\is{A}\\
&+\left(\exp(\theta^{tr}x+x^{tr}\psi{x})-I_{\{\vert{x}\vert\leq\epsilon\}}\right)({\bf{e}}(x)-\mathbb{I}_d)\star\mu\\
\end{split}
\end{equation*}
Thus, $ZS$ is a local martingale if and only if $\widetilde{X}+[\widetilde{X},N]$ is a local martingale if and only if 
\begin{equation}\label{Integrability2}
\left(\exp(\theta^{tr}x+x^{tr}\psi{x})-I_{\{\vert{x}\vert\leq\epsilon\}}\right)({\bf{e}}(x)-\mathbb{I}_d)\star\mu\in{\cal{A}}_{loc},
\end{equation}
 and (\ref{mgEquation4Exponential}) holds. Furthermore, in virtue of Lemma \ref{TechnicalLemma2}-(d), the conditions  (\ref{Integrability1}) and (\ref{Integrability2}) are equivalent to (\ref{IntegrabilityCondition4EE}). This ends the proof of assertion (a).\\
{\bf Part 2.} In this part we prove assertion (c). To this end, we consider $(\theta,\psi)\in\Theta(S)$. Then we associated to $\psi$, the process $\overline{Z}^{\psi}$ defined in (\ref{Zbar(psi)}), which a a well defined and locally bounded local martingale. Thus, by stopping it, there is no loss of generality in assuming that $\overline{Z}^{\psi}$ is a uniformly integrable martingale and $Q^{\psi}:=\overline{Z}^{\psi}_T\cdot{P}$ is a well defined probability. Furthermore, the $Q^{\psi}$-compensator of $\mu$ and $\widetilde\mu$, denoted by $\nu^{\psi}$ and $\widetilde\nu^{\psi}$ respectively, and the canonical decomposition of $\widetilde{X}$ under $Q^{\psi}$ are given by 
\begin{equation*}
\begin{split}
&\nu^{\psi}(dt,dx)=\exp(x^{tr}\psi{x})\nu(dt,dx)=:F^{\psi}_t(dx)dA_t,\quad F^{\psi}_t(dx):=\exp(x^{tr}\psi{x})F_t(dx),\\
& W\star\widetilde\nu^{\psi}=W(.,\Phi(x))\star\nu^{\psi},\quad\Phi(x):=\mbox{\bf{e}}(x)-\mathbb{I}_d,\quad\forall\ W,\\
&\widetilde{X}=X_0+X^c+x\star(\widetilde\mu-\widetilde\nu^{\psi})+\left(\widetilde{b}'+\int(e^{x^{tr}\psi{x}}-1)({\bf{e}}(x)-\mathbb{I}_d)F(dx)\right)\is{A}.
\end{split}
\end{equation*}
Then, due to direct calculations, it is easy to check that $(\theta,\psi)\in\Theta(S)$ satisfying (\ref{IntegrabilityCondition4EE}) and 
\begin{equation}
Z:={\cal{E}}\left(\theta\is{X}^c+(e^{\theta^{tr}x+x^{tr}\psi{x}}-1)\star(\mu-\nu)\right)\in{\cal{Z}}_{loc}(S,P),
\end{equation}
if and only if 
\begin{equation}
\begin{split}
&\exp(\theta^{tr}x) I_{\{\theta^{tr}x>\alpha\}}\star\mu \in{\cal{A}}_{loc}^+(Q^{\psi}),\quad \mbox{for some $\alpha>0$},\\
&{{Z}\over{\overline{Z}^{\psi}}}={\cal{E}}\left(\theta\is{X}^c+(e^{\theta^{tr}x}-1)\star(\mu-\nu^{\psi})\right)\in{\cal{Z}}_{loc}(S,Q^{\psi}).
\end{split}
\end{equation}
This proves assertion (b), and the proof of theorem is complete.
\end{proof}
The proof of Theorem \ref{OneDimension} relies on the following lemma that is true in general. 
\begin{lemma}\label{Deflator4S2Delfators4Shat} Suppose that assumptions of Theorem \ref{OneDimension} hold.\\
Let $\beta\in L^1_{loc}(X^c)$, $g(t,\omega, x)$ be a positive and ${\cal{P}}\otimes{\cal{B}}(\mathbb{R}^d)$-measurable such that $\sqrt{(g-1)^2\star\mu}\in{\cal{A}}^+_{loc}$, and $(\widehat{S},\widehat\nu,\widehat{Z})$ be given by (\ref{ShatZhatNuhat}). Then $Z:={\cal{E}}(\beta\is{X}^c+(g-1)\star(\mu-\nu))\in {\cal{Z}}_{loc}^{L\log{L}}(S)$ if and only if $Z':={\cal{E}}(\beta\is{X}^c+(g-1)\star(\mu-\widehat\nu))\in{\cal{Z}}_{loc}^{L\log{L}}(\widehat{S},\widehat{Z})$. 
\end{lemma}
The proof of the lemma is relegated to Appendix \ref{Appendix4Section3}, while herein we prove Theorem \ref{OneDimension}.
\begin{proof}[Proof of Theorem \ref{OneDimension}] Due to assertion (a), it is clear that   ${\cal{Z}}^{EE}(S)\not=\emptyset$  if and only if ${\cal{Z}}^{LE}(\widehat{S},\widehat{Z})\not=\emptyset$, and thanks to Theorem \ref{mainTHM4LinearEsscherBIS} this is equivalent to ${\cal{Z}}_{loc}^{L\log{L}}(\widehat{S},\widehat{Z})\not=\emptyset$. Thus, in virtue of Lemma \ref{Deflator4S2Delfators4Shat}, the latter claim is equivalent to ${\cal{Z}}_{loc}^{L\log{L}}({S})\not=\emptyset$. This proves assertion (b), while the rest of the proof focuses on assertion (a). To this end, we consider a RCLL and positive process $\overline{Z}>0$, and remark that by stopping we can assume that $\widehat{Z}$ is a uniformly integrable martingale and hence $\widehat{Q}:=\widehat{Z}_{\infty}\cdot P$ is a well defined probability measure. Furthermore, 
\begin{equation}\label{Xhat}
\begin{split}
\widehat{X}&=\left(b'+{{c}\over{2}}\right)\is{A}+x\star(\mu-\nu)+X^c=\widehat{b}\is{A}+x\star(\mu-\widehat\nu)+X^c\\
\widehat{b}&:= b'+{{c}\over{2}}+\int{x}(f(x)-1)F(dx),\quad\widehat\nu(dt,dx):=\widehat{F}(dx)A_t,\quad \widehat{F}(dx):=f(x)F(dx).
\end{split}
\end{equation}
In virtue of Theorem  \ref{mainTHM4LinearEsscher} (applies to $\widehat{S}$ under $\widehat{Q}$ instead), we deduce that $\overline{Z}\in{\cal{Z}}^{LE}(\widehat{S},\widehat{Z})\not=\emptyset$ if and only if there exists a pair $(\theta,\psi)\in\Theta(\widehat{S})$ such that 
\begin{equation}\label{Equa1000}
\begin{split}
&\theta^2c\is A+e^{x\theta}I_{\{\vert\theta{x}\vert>\epsilon\}}\star\mu\in{\cal{A}}^+_{loc}(\widehat{Q}),\quad\widehat{b}+c\theta+\int{x}(e^{x\theta+x^2\psi}-1)\widehat{F}(dx)=0,\\
&\overline{Z}={\cal{E}}\left(\theta\is{X}^c+(e^{x\theta+x^2\psi}-1)\star(\mu-\widehat\nu)\right).
\end{split}
\end{equation}
Remark that, under the first condition in (\ref{Equa1000}), it is clear that $(\theta,\psi)\in\Theta(\widehat{S})$ iff $(\theta,\psi)\in\Theta({S})$. On the one hand,  due to the boundedness of $S$ (and hence of that of $\widehat{Z}$) 
\begin{equation}\label{Equa1001}
\mbox{the first condition in (\ref{Equa1000}) is equivalent to}\ \theta^2c\is A+e^{x\theta}I_{\{\vert\theta{x}\vert>\epsilon\}}\star\mu\in{\cal{A}}^+_{loc}.
\end{equation}
On the other hand, by using the notation in (\ref{Xhat}), we deduce that the second condition in (\ref{Equa1000}) is equivalent to
\begin{equation}\label{Equa1002}
 b'+{{c}\over{2}}+c\theta+\int(e^{x}-1)(e^{x\theta+x^2\psi}-1){F}(dx)=0.
\end{equation}
Furthermore, using Yor's formula and direct calculations, we conclude that the third condition in (\ref{Equa1000}) holds if and only if 
\begin{equation}\label{Equa1003}
 Z:=\overline{Z}\widehat{Z}={\cal{E}}\left(\theta\is{X}^c+(e^{x\theta+x^2\psi}-1)\star(\mu-\nu)\right).
\end{equation}
Thus, by combining (\ref{Equa1003}), (\ref{Equa1002}),(\ref{Equa1001}) and (\ref{Equa1000}), assertion (a) follows immediately. This also proves (\ref{RelationshipZ2Zhat}), and the proof of the theorem is complete.
\end{proof}
\section{Esscher pricing bounds and linear constraints BSDEs}\label{Section4pricing}
Throughout this section, we suppose $d=1$ and $\mathbb{F}$ is the filtration generated by a Brownian motion $W$ and a non-homogeneous Poisson process $N$, where $W$ and $N$ are independent. Our market model consist of one risky asset $S$ and non-risky asset $S^{(0)}:=\exp(B)$ having the following dynamics
\begin{equation}\label{Model4S}
\begin{split}
&B:=\int_0^{\cdot}r_s ds,\ r\geq 0,\quad  S := S_0 e^X=S_0\mathcal{E} (\widetilde{X}),\ dX:= {b} dt + \sigma{d}W + \gamma{d}\widetilde{N},\ X_0=\widetilde{X}_0=0,\\
   &\widetilde{N} := N -\int_0^\cdot \lambda_s ds,\  d\widetilde{X}:=\widetilde{b}dt + \sigma{d}W + \widetilde\gamma{d}{\widetilde N},\ \widetilde{b} :=b+ \frac{\sigma^2}{2} +\lambda (e^{\gamma} - 1-\gamma)\ \mbox{and}\  \widetilde{\gamma} := e^\gamma - 1.
   \end{split}
\end{equation}
Throughout the rest of the paper, we assume that
\begin{equation}\label{mainassumtpion4S}
 \sigma>0,\ \lambda>0,\ \sigma+\lambda+\sigma^{-1}+\lambda^{-1}+ \vert{r}\vert+\vert{b}\vert+\vert\gamma\vert+\vert\gamma\vert^{-1}\leq C,\ {P}\otimes{dt}-a.e.,\ \mbox{for some}\ C\in(0,\infty),
\end{equation}
and we consider the following sets
\begin{equation}\label{Psi-Psi(n)}
\Psi:=\left\{\psi:\ \psi\ \mbox{is predictable and bounded}\right\},\quad\Psi_n:=\left\{\psi\in\Psi:\ \vert\psi\vert\leq n\right\},\quad n\in\mathbb{N}.
\end{equation}
\begin{remark} For the rest of the paper, we will work for this simple model, (\ref{Model4S}),  and under the assumptions (\ref{mainassumtpion4S}). Our unique leitmotiv for these restrictions on $(S,\mathbb{F})$ lies in avoiding technicalities that might overshadow the key ideas, and we want to present our novel ideas on the second-order Esscher concept and their usefulness in the simplest model possible. It is important to know that all those restrictions on $(S,\mathbb{F})$ can be extended to the most general case possible, as the theory of BSDEs is well developed nowadays.  
\end{remark}
\begin{lemma}\label{Lemma4Z(psi)Densities} For any $\psi\in\Psi$ and $\zeta\in\{\gamma,\widetilde\gamma\}$, we denote by $\eta(\psi)$ the unique root for 
\begin{equation}\label{root4eta(psi)}
 \widetilde{b}-r + \eta{\sigma}^2+\lambda\widetilde{\gamma}\bigg(e^{\eta\zeta + \psi\zeta^2}-1\bigg)=0,
\end{equation}
and $\overline{D}^{\psi}$ is given by  
\begin{equation}\label{Zbar(psi)}
\overline{D}^{\psi}:= {\cal{E}}\left(\eta(\psi)\sigma\is{W}+(e^{\eta(\psi)\zeta+\zeta^2\psi}-1)\is\widetilde{N}\right).
\end{equation}
{\rm{(a)}} If (\ref{mainassumtpion4S}) holds,  then for any $\psi\in\Psi$ we have $\overline{D}^{\psi}\in{\cal{M}}(P)$, and 
\begin{equation}\label{R(psi)}
R_{\psi}:=\overline{D}^{\psi}_T\cdot P\quad\mbox{is a well defined probability measure}.
\end{equation}
 In particular, for $\psi=0\in\Psi$, we get   
\begin{equation}\label{Rzero}
\overline{D}^{0}\in{\cal{M}}\quad\mbox{and}\quad R_0:=\overline{D}^{0}_T\cdot P\quad\mbox{is a well defined probability measure}.
\end{equation}
{\rm{(b)}} Suppose (\ref{mainassumtpion4S}) holds. Then for any $\psi\in\Psi$, we have  
\begin{equation}\label{W(0)Ntilde(o)}
    W^{\psi}:=W- \int_0^\cdot \eta_s(\psi) \sigma_s ds\in{\cal{M}}_{loc}(R_{\psi}) \ \mbox{and} \ \widetilde{N}^{\psi} :=  \widetilde{N} -\int_0^\cdot \frac{r_s - \widetilde{b}_s - \eta_s(\psi)\sigma^2_s}{\widetilde\gamma_s} ds\in{\cal{M}}_{loc}(R_{\psi}).
\end{equation} 
{\rm{(c)}} Suppose (\ref{mainassumtpion4S}) holds. Then for $p\geq 1$ and $n\in\mathbb{N}$, there exists $C_n\in(0,\infty)$ such that  
\begin{equation}
\begin{split}
&\vert\eta(\psi)\vert+\mathbb{E}_0\left[\left({{D^{\psi}_T}\over{D^{\psi}_t}}\right)^p\Big|{\cal{F}}_t\right]\leq C_n,\quad P-a.s.,\quad\mbox{for any}\quad \psi\in\Psi_n,\quad\quad\mbox{where}\\
& {D}^{\psi}:={{\overline{D}^{\psi}}\over{\overline{D}^0}}={\cal{E}}\left(M^{\psi}\right),\ M^{\psi}:=(\eta(\psi)-\eta(0))\sigma\is{W}^0+\left(\exp((\eta(\psi)-\eta(0))\zeta+\zeta^2\psi)-1\right)\is\widetilde{N}^0.
\end{split}
\end{equation}
\end{lemma}
The proof of this lemma is relegated to Appendix \ref{Appendix4Section4}, while herein we define the Esscher price processes bounds for any claim. To this end, we consider ${\cal{Q}}(\zeta)$ and ${\cal{Z}}(\zeta)$, for $\zeta\in\{\gamma,\widetilde\gamma\}$, given by
\begin{equation}\label{Psi(martingales)}
 {\cal{Q}}(\zeta):=\left\{R^{\psi}:=D^{\psi}_T\cdot{R}_0:\quad D^{\psi}\in {\cal{Z}}(\zeta)\right\}\quad\mbox{and}\quad {\cal{Z}}(\zeta):=\left\{D^{\psi}:={{\overline{D}^{\psi}}\over{ \overline{D}^0}}:\quad \psi\in\Psi\right\}.
\end{equation}
\begin{definition}\label{optprice} Let $\xi$ be the payoff of an arbitrary claim satisfying
\begin{equation}\label{ClaimAssumption}
\mathbb{E}_{R_0}\left[\vert\xi\vert\right]\leq \sup_{R^{\psi}\in{\cal{Q}}(\zeta)}\mathbb{E}^{R^{\psi}}\left[\vert\xi\vert\right]=\sup_{Z^{\psi}\in{\cal{Z}}(\zeta)}\mathbb{E}_0\left[Z^{\psi}_T\vert\xi\vert\right]<\infty, \quad \zeta\in \{\gamma,\widetilde\gamma\}.
\end{equation}
We denote by $\Rbrack{Y}^{\rm{inf,exp}},Y^{\rm{up,exp}} \Lbrack$ (respectively $\Rbrack{Y}^{\rm{inf,lin}},Y^{\rm{up,lin}} \Lbrack$) the exponential-Esscher  (respectively the linear-Esscher) pricing stochastic interval for the claim $\xi$, where 
\begin{equation}\label{EsscherInterval}
\begin{split}
&Y^{\rm{inf,exp}}:= \underset{\psi\in\Psi:\ R^{\psi}\in{\cal{Q}}(\gamma)}{\essinf}(Y^\psi),\quad  Y^{\rm{up,exp}}:= \underset{\psi\in\Psi:\ R^{\psi}\in{\cal{Q}}(\gamma)}{\esssup} (Y^\psi),\\
&Y^{\rm{inf,lin}}:= \underset{\psi\in\Psi:\ R^{\psi}\in{\cal{Q}}(\widetilde\gamma)}{\essinf}(Y^\psi),\quad  Y^{\rm{up,lin}}:=  \underset{\psi\in\Psi:\ R^{\psi}\in{\cal{Q}}(\widetilde\gamma)}{\esssup} (Y^\psi),\\
&Y^{\psi}_t:= \mathbb{E}_0 \left[{{D^{\psi}_T}\over{D^{\psi}_t}}e^{-\int_t^Tr_sds}\xi \Big| \mathcal{F}_t \right],\quad{D}^{\psi}\in{\cal{Z}}(\zeta), \quad\zeta\in\{\gamma,\widetilde\gamma\}.
\end{split}
\end{equation}
\end{definition}
Our main goal of this section resides in describing as explicit as possible the four processes, $Y^{\rm{inf,exp}}$, $Y^{\rm{up,exp}}$, $Y^{\rm{inf,lin}}$ and $Y^{\rm{up,lin}}$, and singling out afterwards their precise relationship as well. This can be achieved through BSDEs, which a ``{\it natural}'' stochastic control tool for non Markovian models having more complex dynamics. The rest of this section is divided into three subsections. The first subsection gives some definitions and notation regarding BSDEs with constraints that we {\it naturally} encounter in our analysis. The second subsection is devoted to our main results of this section, while the last subsection proves these main results. 
\subsection{Constrained BSDEs: Definitions  and notation}
Throughout this subsection, we consider a probability measure $Q$ on $(\Omega,{\cal{F}})$,  and we denote by $W^Q$ and $\widetilde{N}^Q$ the $Q$-Brownian motion and the compensated Poisson process under $Q$ respectively. For an y $p\in(1,\infty)$, throughout the rest of the paper we consider spaces, $\mathbb{S}^p(Q)$, $\mathbb{L}^p(W,Q)$, $\mathbb{L}^p(N,Q)$ and ${\cal{A}}^2(Q)$ and their norms given below.  For any unexplained notion, we refer to Section \ref{Section4ModelNotations}.
\begin{equation}\label{Spaces4BSDEs}
\begin{split}
&\mathbb{S}^p(Q):=\Biggl\{X\ \mbox{RCLL $\&$ $\mathbb{F}$-adapted process}:\ \Vert{X}\Vert_{\mathbb{S}^p(Q)}<\infty\Biggr\},\ \Vert{X}\Vert_{\mathbb{S}^p(Q)}^p:=\mathbb{E}^{Q}\Bigl [\sup_{0\leq t\leq T} \vert{X}_t\vert^p\Bigr], \\
&\mathbb{L}^p(N,Q):= \left\{\varphi\in L(\widetilde{N}^Q,Q):\ \Vert\varphi\Vert_{\mathbb{L}^p(N,Q)}<\infty \right\},\quad \Vert\varphi\Vert_{\mathbb{L}^p(N,Q)}^p:=\mathbb{E}^Q \Bigl[(|\varphi|^2\is{N})_T^{p/2}\Bigr],\\
  &\mathbb{L}^p(W,Q):=\Bigl \{\varphi\in L(W^Q,Q):\  \Vert\varphi\Vert_{\mathbb{L}^p(W,Q)}<\infty \Bigr\},\quad \Vert\varphi\Vert_{\mathbb{L}^p(W,Q)}^p:=\mathbb{E}^Q \Bigl[\left(\int_0^T |\varphi_t|^2{dt}\right)^{p/2}\Bigr],\\
 &{\cal{A}}^p(Q):=\Bigl\{V\in{\cal{A}}_{loc}(Q):\ \Vert{V}\Vert_{{\cal{A}}^p(Q)}<\infty\Bigr\},\quad \Vert{V}\Vert_{{\cal{A}}^p(Q)}^p:=\mathbb{E}^Q\left[(\mbox{Var}_T(V))^p\right].
 \end{split}
\end{equation}
Now we can introduce the constrained BSDEs for models with jumps. 
\begin{definition} Let  $\xi\in L^p(Q)$, and $f(t,\omega,y,z,u)$and $\Phi(t,\omega, y,z,u)$ be $\mathbb{F}$-optional functionals and Lipschitz in $(y,z,u)\in\mathbb{R}^3$ such that $\Phi(t,\omega,y,z,u)\geq 0$ $P\otimes dt$-almost every $(\omega,t)$ and for any $(y,z,u)$. \\
{\rm{(a)}} We call $(Y,Z,U,K)$ an $L^p(Q)$-solution to the constrained BSDE $(f,\xi,\Phi)$ (CBSDE$(f,\xi,\Phi)$ hereafter for short), if $(Y,Z,U,K)\in\mathbb{S}^p(Q)\times\mathbb{L}^p(W,Q)\times\mathbb{L}^p(N,Q)\times {\cal{A}}^p(Q)$, $K\in {\cal{V}}^+$ and is predictable, and 
 \begin{equation}\label{BSDEwithCobstraints}
 \begin{split}
 &Y_t=\xi+\int_t^Tf(s,Y_s,Z_s,U_s)ds-\int_t^T Z_s dW^Q_s-\int_t^TU_sd\widetilde{N}^Q_s,\quad{Q}\mbox{-a.s.,}\\
 &\mbox{and}\quad\Phi(t,Y_t,Z_t,U_t)=0,\quad Q\otimes{dt}\mbox{-a.e.}.
 \end{split}
 \end{equation}
 {\rm{(b)}}  $(\widehat{Y},\widehat{Z},\widehat{U},\widehat{K})$ is called the smallest $L^p(Q)$-solution to the CBSDE$(f,\xi,\Phi)$, i.e. (\ref{BSDEwithCobstraints}), if it is a solution to this CBSDE$(f,\xi,\Phi)$, and $\widehat{Y}\leq Y$ for any other solution $(Y,Z,U,K)$ to (\ref{BSDEwithCobstraints}). 
 \end{definition}
It is clear that from this definition (see (a) above), a reflected BSDE is a particular case of constrained BSDE, where the constraint factor $\Phi$ depends on the value process $Y$ only (i.e. $\Phi(t,y,z,u)=\Phi(t,y)$). In this case, the constraint generated naturally the Skorokhod condition $\int_0^T \Phi(s,Y_s)dK_s=0$ $Q$-a.s..\\
Up to our knowledge, CBSDEs appeared for the first time in \cite{cvitanic1998backward}, for the Brownian setting only. The obtained CBSDE, in these papers, was motivated by the problem of supper-replication when there are constrained on portfolio. Hence, naturally the constrained in the original problem of super-replication translate into constrains on the solution of the resulting BSDE.  The main challenge, that was noticed by the early days of the CBSDEs, lies in the existence of a solution to those CBSDEs. In fact, in \cite{cvitanic1998backward}, the authors gave a counter-example showing that the existence might fail in general. Thus, still in the Brownian setting and {\it under the assumption that the CBSDE --under consideration-- has indeed a solution}, Peng proved in \cite{peng1999monotonic} that the smallest solution exists. Many extensions were attempted afterwards, for these we refer the reader to Very recently \cite{peng2010reflected,Xu,elie2014adding,kharroubi2010backward} and the references therein to cite a few. However, in all these extensions, the assumption that the constrained BSDE should admit a solution persists and was not overcome at all, and hence in our setting these are not applicable.
\subsection{Main results on Esscher prices' bounds}
This subsection states our main results on the bounds for the Esscher-pricing intervals. In particular, we prove that the upper and lower bounds of the Esscher-pricing interval are solutions to constrained BSDEs, with different constrains, but having the same driver which is linear in the solution's variables $(y,z,u)$. Below we state our main theorem on the upper bound for the Esscher price processes.
\begin{theorem}\label{BSDE4Y(up)}  Suppose (\ref{mainassumtpion4S}) holds, and consider $p\in(1,\infty)$ and $\xi\in{L}^{0}(P)$ satisfying 
\begin{equation}\label{MainAssumption4Xi}
\mathbb{E}_0\left[(\xi^-)^p\right]+\sup_{\tau\in{\cal{T}},\psi\in\Psi}\mathbb{E}_0\left[{{D^{\psi}_T}\over{D^{\psi}_{\tau}}}(\xi^+)^p\right]<\infty,\quad\mbox{where ${\cal{T}}$ is the set of all stopping times $\tau\leq{T}$}.
\end{equation}
Let $(Y^{(up,exp)}, Y^{(up,lin)},Y^{\psi})$, for any $\psi\in\Psi$, be defined in (\ref{EsscherInterval}), and $Y^{(n)}$ be given by  
\begin{equation}\label{Y(n)4up}
Y^{(n)}:=\underset{\psi\in\Psi_n}{\esssup}(Y^{\psi}),\quad \mbox{for any}\quad n\in\mathbb{N}.\end{equation}
 If $Y^{(up)}\in\{Y^{(up,exp)},Y^{(up,lin)}\}$, then the following assertions hold.\\
{\rm{(a)}} There exists $(Z^{(up)},U^{(up)},K^{(up)})$ such that  the quadruplet $(Y^{(up)},Z^{(up)},U^{(up)},K^{(up)})$ belongs to $ \mathbb{S}^p(R_0) \times \mathbb{L}^p(W,R_0) \times \mathbb{L}^p(N,R_0)\times{\cal{A}}^{+,p}(R_0)$ and is the smallest solution to the constrained BSDE
\begin{equation}\label{ConstrainedBSDE1}
\begin{split}
dY_s&=\left({{\widetilde{b}_s-r_s-\widetilde\gamma_s\lambda_s }\over{ \sigma_s}}Z_s+\lambda_s{U}_s+r_sY_s\right)ds+Z_s{d}W_s+{U}_sd\widetilde{N}_s-dK_s,\quad Y_T=\xi,\quad P-a.s.,\\
&\widetilde\gamma_s{Z}_s-\sigma_s{U}_s\geq 0,\quad P\otimes{ds}\mbox{-a.e.},\quad \int_0^T(\widetilde\gamma_s{Z}_s-\sigma_s{U}_s)dK_s=0,\ P-a.s..
\end{split}
\end{equation}
{\rm{(b)}} There exists a unique pair $(Z^{(n)},U^{(n)})\in \mathbb{L}^p(W,R_0) \times \mathbb{L}^p(N,R_0)$ such that the triplet $(Y^{(n)}, Z^{(n)},U^{(n)})$ belongs to $ \mathbb{S}^p(R_0) \times \mathbb{L}^p(W,R_0) \times \mathbb{L}^p(N,R_0)$ and is the unique solution to the following BSDE
\begin{equation}\label{BSDE(n)}
\begin{split}
Y_t=&\xi+\int_t^T \left({{\eta_s(-n)-\eta_s(0)}\over{\widetilde\gamma_s}}\sigma_s(\widetilde\gamma_s{Z}_s-\sigma_s{U}_s)^+ +  {{\eta_s(0)-\eta_s(n)}\over{\widetilde\gamma_s}}\sigma_s(\widetilde\gamma_s{Z}_s-\sigma_s{U}_s)^--r_s Y_s\right)ds\\
&\hskip 1cm-\int_t^T Z_s{d}W^0_s-\int_t^T{U}_sd\widetilde{N}^0_s.
\end{split}\end{equation}
{\rm{(c)}} For any $n,m\in\mathbb{N}$, ${Y}^{(n+m)}-{Y}^{(n)}$ is a nonnegative $R_0$-supermartingale, 
\begin{equation}\label{Increasiness4Y(n)}
Y^{(n)}\leq{Y}^{(n+1)}\leq{Y}^{up},\quad\mbox{and}\quad Y^{(n)}\ \mbox{converges pointwise to}\ Y^{up}.
\end{equation}
{\rm{(d)}} $(Y^{(n)}, Z^{(n)},U^{(n)})$  converges to $(Y^{(up)},Z^{(up)},U^{(up)})$ in $ \mathbb{S}^p(R_0) \times \mathbb{L}^p(W^0,R_0) \times \mathbb{L}^p(\widetilde{N}^0,R_0)$, and 
$K^{(n)}:=\int_0^{\cdot}(\eta_s(0)-\eta_s(n))\widetilde\gamma_s^{-1}\sigma_s(\widetilde\gamma_s{Z}_s^{(n)}-\sigma_s{U}_s^{(n)})^-ds$ converges to $K^{(up)}$ in the space ${\cal{A}}^{+,p}(R_0)$. As a result, there exists a nonnegative and predictable process $k^{(up)}$ such that $K^{(up)}=\int_0^{\cdot}k^{(up)}_s ds$.
\end{theorem}
The proof of the theorem is relegated to Subsection \ref{Proofs4Section3}. Assertion (a) completely characterizes the upper bound $Y^{up}$ as the smallest solution a {\it constrained} BSDE with  Skorokhod condition. 
\begin{remark}
The discounted upper Esscher-price process $\widetilde{Y}^{up}:=e^{-B}Y^{up}$ is a nonnegative $R_0$-supermartingale, and satisfies 
\begin{equation}
\begin{split}
&\widetilde{Y}^{up}_t=\mathbb{E}_0\left[\xi{e}^{-B_T}+\int_t^T{{\eta^*_s-\eta_s(0)}\over{\widetilde\gamma_s}}\sigma_s(\widetilde{\gamma_s}\widetilde{Z}^{up}_s-\sigma_s\widetilde{U}^{up}_s)ds+\widetilde{K}^{up}_T-\widetilde{K}^{up}_t\Big|{\cal{F}}_t\right]\\
&\widetilde{Y}^{up}_t=\xi{e}^{-B_T}+\int_t^T {{\eta^*_s-\eta_s(0)}\over{\widetilde\gamma_s}}\sigma_s(\widetilde{\gamma_s}\widetilde{Z}^{up}_s-\sigma_s\widetilde{U}^{up}_s)ds-\int_t^T\widetilde{Z}^{up}_sdW^0_s-\int_t^T\widetilde{U}^{up}_sd\widetilde{N}^0_s+\widetilde{K}^{up}_T-\widetilde{K}^{up}_t,
\end{split}
\end{equation}
where $\widetilde{Z}^{up}:=e^{-B}Z^{up}$, $\widetilde{U}^{up}={e}^{-B_T}{U}^{up}$ and $\widetilde{K}^{up}:=e^{-B}\is{K}^{up}$.
\end{remark}
Below, we address the lower Esscher bound process, and we prove that is also fully characterized by a constrained BSDE. Both CBSDEs (i.e. the CBSDEs for upper bound and lower bound) have the same driver, but different constraints and different martingale structures naturally.
\begin{theorem}\label{BSDE4Y(inf)} Suppose (\ref{mainassumtpion4S}) holds, and consider $p\in(1,\infty)$ and $\xi\in{L}^{0}(P)$ satisfying 
\begin{equation}\label{MainAssumption4Xi(inf)}
\mathbb{E}_0\left[(\xi^+)^p\right]+\sup_{\tau\in{\cal{T}},\psi\in\Psi}\mathbb{E}_0\left[{{D^{\psi}_T}\over{D^{\psi}_{\tau}}}(\xi^-)^p\right]<\infty.
\end{equation}
Let $(Y^{(inf,exp)}, Y^{(inf,lin)},Y^{\psi})$, for any $\psi\in\Psi$, be defined in (\ref{EsscherInterval}), and $\overline{Y}^{(n)}$ be given by  
\begin{equation}\label{Ybar(n)}
 \overline{Y}^{(n)}:=\underset{\psi\in\Psi_n}{\essinf}(Y^{\psi}),\quad \mbox{for any}\quad n\in\mathbb{N}.\end{equation}
If $Y^{(inf)}\in\{Y^{(inf,lin)},Y^{(inf,exp)}\}$, then the following assertions hold.\\
{\rm{(a)}}  There exists $(Z^{(inf)},U^{(inf)},K^{(inf)})$ such that the quadruplet $(Y^{(inf)},Z^{(inf)},U^{(inf)},K^{(inf)})$ belongs to $ \mathbb{S}^p(R_0) \times \mathbb{L}^p(W^0,R_0) \times \mathbb{L}^p(\widetilde{N}^0,R_0)\times{\cal{A}}^{+,p}(R_0)$ and is the smallest solution to the constrained BSDE
\begin{equation}
\begin{split}
&dY_s=\left({{\widetilde{b}_s-r_s-\widetilde\gamma_s\lambda_s }\over{ \sigma_s}}Z_s+\lambda_s{U}_s+r_sY_s\right)ds+Z_s{d}W_s+{U}_sd\widetilde{N}_s+dK_s,\quad Y_T=\xi,\quad P-a.s.,\\
&\widetilde\gamma_s{Z}_s-\sigma_s{U}_s\leq 0,\quad P\otimes{ds}\mbox{-a.e.},\quad \int_0^T(\widetilde\gamma_s{Z}_s-\sigma_s{U}_s)dK_s=0,\ P-a.s..
\end{split}
\end{equation}
{\rm{(b)}} There exists a unique pair $(\overline{Z}^{(n)},\overline{U}^{(n)})\in \mathbb{L}^p(W^0,R_0) \times \mathbb{L}^p(\widetilde{N}^0,R_0)$ such that $(\overline{Y}^{(n)}, \overline{Z}^{(n)},\overline{U}^{(n)})$ is the unique solution to the following BSDE
\begin{equation}\label{BSDE(n)}
\begin{split}
Y_t=&\xi-\int_t^T\left({{\eta_s(0)-\eta_s(n)}\over{\widetilde\gamma_s}}\sigma_s(\widetilde\gamma_s{Z}_s-\sigma_s{U}_s)^+ +  {{\eta_s(-n)-\eta_s(0)}\over{\widetilde\gamma_s}}\sigma_s(\widetilde\gamma_s{Z}_s-\sigma_s{U}_s)^--r_sY_s\right)ds\\
&\hskip 1cm-\int_t^TZ_s{d}W^0_s-\int_t^T{U}_sd\widetilde{N}^0_s.
\end{split}\end{equation}
{\rm{(c)}}  For any $n,m\in\mathbb{N}$, $\overline{Y}^{(n)}-\overline{Y}^{(n+m)}$ is a nonnegative $R_0$-submartingale, 
\begin{equation}\label{Decreasiness4Ybar(n)}
\overline{Y}^{(n)}\geq\overline{Y}^{(n+1)}\geq{Y}^{(inf)}\quad\mbox{and}\quad \overline{Y}^{(n)}\ \mbox{converges pointwise to}\ Y^{(inf)}.
\end{equation}
{\rm{(d)}}  $(\overline{Y}^{(n)}, \overline{Z}^{(n)},\overline{U}^{(n)})$ converges to  $(Y^{(inf)},Z^{(inf)},U^{(inf)})$ in $ \mathbb{S}^p(R_0) \times \mathbb{L}^p(W^0,R_0) \times \mathbb{L}^p(\widetilde{N}^0,R_0)$. As a result, there exists a predictable process $k^{(inf)}$ such that $K^{(inf)}:=\int_0^{\cdot}k^{(inf)}_sds\in{\cal{A}}^{+,p}(R_0)$ and  $\overline{K}^{(n)}:=\int_0^{\cdot}\left(\eta_s(0)-\eta_s(n)\right){\widetilde\gamma_s}^{-1}\sigma_s(\widetilde\gamma_s\overline{Z}_s^{(n)}-\sigma_s\overline{U}_s^{(n)})^+ds$ converges to $K^{(inf)}$ in ${\cal{A}}^{+,p}(R_0)$-norm.
\end{theorem}
\begin{proof}[Proof of Theorem \ref{BSDE4Y(inf)}] Let $\xi$ be a claim, and $\psi\in\Psi$, we denote by $Y^{\psi}(\xi)$ and $Y^{\psi}(-\xi)$ the process defined in (\ref{EsscherInterval}) associated to the claim $\xi$ and $-\xi$ respectively. On the one hand, it is clear that the solution to (\ref{YpsiBSDE}) for $\xi$ and $-\xi$, denoted by $( Y^{\psi}(\xi),Z^{\psi}(\xi),U^{\psi}(\xi))$ and $(Y^{\psi}(-\xi),Z^{\psi}(-\xi),U^{\psi}(-\xi))$ respectively, satisfy 
$$
(Y^{\psi}(-\xi),Z^{\psi}(-\xi),U^{\psi}(-\xi))=-( Y^{\psi}(\xi),Z^{\psi}(\xi),U^{\psi}(\xi)).
$$
On the other hand, due to the easy fact that $\underset{i\in{I}}{\essinf}(X_i)=-\underset{i\in{I}}{\esssup}(-X_i)$ for any family of random variable $(X_i)_{i\in{I}}$, we deduce that 
\begin{equation*}
Y^{(inf)}(\xi)=\underset{\psi\in\Psi}{\essinf}(Y^{\psi}(\xi))=-\underset{\psi\in\Psi}{\esssup}(Y^{\psi}(-\xi))=-Y^{up}(-\xi).
\end{equation*}
Therefore, by applying Theorem \ref{BSDE4Y(up)} to the claim $-\xi$ instead, and combining the above remarks afterwards, all assertions of the theorem follow immediately. 
\end{proof}
We end this subsection with the following remark. 
\begin{remark} {\rm{(a)}} For any claim with payoff $\xi$ and any $p\in (1,\infty)$, the condition that we require for both Esscher-pricing bounds is 
\begin{equation}\label{Assumption4Xi}
{\mathbb{E}}_0\left[\vert\xi\vert^p\right]\leq \sup_{\tau\in{\cal{T}},\psi\in\Psi}{\mathbb{E}}_0\left[{{D^{\psi}_T}\over{D^{\psi}_{\tau}}}\vert\xi\vert^p\right]<\infty.
\end{equation}
It is clear that any bounded claim $\xi$ satisfies this assumption. \\
 {\rm{(b)}} For any $\zeta\in\{\gamma,\widetilde\gamma\}$ and any $\xi$ fulfilling (\ref{Assumption4Xi}), the processes $Y^{up}$ and $Y^{inf}$ are $R_0$-supermartingale and $R_0$-submartingale such that 
 \begin{equation*}
Y^{inf}<Y^{\psi}<Y^{up},\quad \forall\quad \psi\in\Psi.
 \end{equation*}
 Furthermore, there exist $(\theta^{up},\theta^{inf})\in\Theta(S)\times\Theta(S)$ and  a unique pair $(L^{up},-L^{inf})$ of strict $R_0$-supermartingales such that the following hold.
 \begin{equation*}
 \begin{split}
 &Y^{inf}=Y^{inf}_0+\theta^{inf}\is{S}-K^{inf}+L^{inf},\quad Y^{up}=Y^{up}_0+\theta^{up}\is{S}-K^{up}-L^{up},\\
 &\left\{Q\ \mbox{equivalent probability}:\ L^{up}\in{\cal{M}}_{loc}(Q)\ \mbox{or}\quad L^{inf}\in{\cal{M}}_{loc}(Q)\right\}=\emptyset.
 \end{split}
 \end{equation*}

\end{remark}
\subsection{Proof of Theorem \ref{BSDE4Y(up)} 
}\label{Proofs4Section3}
This subsection proves the two  main theorems of this section. To this end, we start with some intermediate results that we summarize in four lemmas. These lemmas are interesting in themselves besides conveying the main ideas behind our main results. 
\begin{lemma}\label{EtaProcess} For any $\psi\in\Psi$ and any $\zeta\in\{\gamma,\widetilde\gamma\}$, let $\eta(\psi)$ be the unique root to (\ref{root4eta(psi)}), and consider the functional $\Phi$ given by 
\begin{equation}\label{PhiFunction}
\Phi(x):=\Phi(t,\omega,x):=\sigma_t(\omega)^2x+\lambda_t(\omega)\widetilde\gamma_t(\omega)\zeta_t(\omega)e^x,\quad t\in[0,T],\quad \omega\in\Omega,\quad x\in\mathbb{R}.
\end{equation}
Then the following assertions hold.\\
 {\rm{(a)}}  The functional $\Phi$ is continuous, strictly increasing, $\Phi(\infty)=\infty$, $\Phi(-\infty)=-\infty$, and 
 \begin{equation}\label{Phi2eta}
 \eta(\psi)={1\over{\zeta}}\Phi^{-1}\Bigl(\zeta(r-\widetilde{b}+\lambda\widetilde\gamma)+\sigma^2\zeta^2\psi\Bigr)-\zeta\psi.
 \end{equation}
 As a result, for any $\psi\in\Psi$, there exists $C_{\psi}\in(0,\infty)$ such that $\vert\eta(\psi)\vert\leq C_{\psi}$ $P\otimes{dt}$-a.e..\\
 {\rm{(b)}}  $\eta(\psi) \widetilde\gamma^{-1}$ is decreasing ($P\otimes{dt}$-a.e.) with respect to $\psi$.\\
  {\rm{(c)}}   For any $\zeta\in\{\gamma,\widetilde\gamma\}$, we have 
        \begin{equation}\label{etaLimitsetaStar}
        \lim_{\psi \downarrow -\infty} {{\eta(\psi)}\over{\widetilde\gamma}} = {{\eta(-\infty)}\over{\widetilde\gamma}}= {{\lambda\widetilde\gamma+r-\widetilde{b}}\over{\widetilde\gamma \sigma^2}}=:{{\eta^*}\over{\widetilde\gamma}},\quad \mbox{and}\quad  \lim_{\psi \uparrow \infty} {{\eta(\psi)}\over{\widetilde\gamma}} = -\infty,\quad P\otimes{dt}\mbox{-a.e.}.\end{equation}
        {\rm{(d)}} For any $\zeta\in\{\gamma,\widetilde\gamma\}$ and any $\psi\in\Psi$ such that $\psi\geq0$, we have 
\begin{equation}\label{limite4eta(n)/n}\lim_{n\longrightarrow\infty}{{-\eta(n)}\over{n\widetilde\gamma}}={{\zeta}\over{\widetilde\gamma}},\quad\mbox{and}\quad {{\eta(0)-\eta(\psi)}\over{\widetilde\gamma}}\leq{{\zeta}\over{\widetilde\gamma}}\psi,\quad P\otimes{dt}\mbox{-a.e.}.
\end{equation}
\end{lemma}
The second lemma characterizes the process $Y^{\psi}$, defined in (\ref{EsscherInterval}), in a unique manner by a BSDE, and elaborate some its properties that will be useful throughout the rest of the paper.
\begin{lemma}\label{Lemma1} Suppose that (\ref{mainassumtpion4S}) holds, and let $p\in(1,\infty)$ and $\xi\in{L}^{0}(P)$ satisfying (\ref{MainAssumption4Xi}). For any $\psi\in\Psi$, $\eta(\psi)$ denote the unique root to (\ref{root4eta(psi)}), and  $Y^\psi$ is defined in (\ref{EsscherInterval}). Then  the following hold.\\
{\rm{(a)}} For any $\psi\in\Psi$,there exists a unique pair $( Z^{\psi}, U^{\psi}) \in  \mathbb{L}^p(W,R_0) \times \mathbb{L}^p(N,R_0)$, such that the triple $(Y^{\psi}, Z^{\psi}, U^{\psi})$ belongs to the space $ \mathbb{S}^p(R_0) \times \mathbb{L}^p(W,R_0) \times \mathbb{L}^p(N,R_0)$  and is the unique solution to 
    \begin{equation}\label{YpsiBSDE}
    \begin{split}
        Y_t^\psi &= \xi +\int_t^T f_{\psi}(s,Y^\psi_s, Z^\psi_s, U^\psi_s) ds - \int_t^T Z^\psi_s dW_s - \int_t^T U^\psi_s d\widetilde{N}_s,\\
        &= \xi +\int_t^T  h_{\psi}(s,Y^\psi_s,Z^\psi_s,U^\psi_s) ds- \int_t^T Z^\psi_s dW^0_s - \int_t^T U^\psi_s d\widetilde{N}^0_s.
        \end{split}
    \end{equation}
    Here  the functionals $ f_{\psi}$ and $ h_{\psi}$ are given, for any $ t \in[0,T]$ and $ (y,z,u)\in\mathbb{R}^3$, by
    \begin{equation}\label{generatorf}
    \begin{split}
        f_{\psi}(t, y, z, u) &:= \frac{\eta_t(\psi) \sigma_t}{\widetilde\gamma_t} (\widetilde\gamma_t z - \sigma_t u)- r_t y - \big(\frac{\widetilde{b}_t - r_t}{\widetilde \gamma_t} \big)u,\\
        h_{\psi}(t,y,z,u)&:= \frac{ \sigma_t}{\widetilde\gamma_t} \bigl(\eta_t(\psi) - \eta_s(0)\bigr)\left(\widetilde\gamma_t z - \sigma_t u\right) - r_t y.
        \end{split}
    \end{equation}
 Furthermore, for any $n\geq 1$, there exists $C_n\in(0,\infty)$ such that for any $\psi\in\Psi_n$,
 \begin{equation}\label{Control4(Y,Z,U,A)(psi)}
\Vert {Y}^{\psi}\Vert_{\mathbb{S}^p(R_0)}+\Vert{Z}^{\psi}\Vert_{\mathbb{L}^p(W,R_0) }+\Vert{U}^{\psi}\Vert_{\mathbb{L}^p(N,R_0)}+\Vert\int_0^T\vert{h}_{\psi}(s,Y^\psi_s,Z^\psi_s,U^\psi_s)\vert ds\Vert_{L^p(R_0)}\leq C_n\Vert\xi\Vert_{L^p(R_0)}.
 \end{equation}
 {\rm{(b)}} Let $B$ be defined in (\ref{Model4S}), $\psi\in\Psi$, and $(Y^{\psi},Z^{\psi},U^{\psi})$ be the solution to (\ref{YpsiBSDE}). If we denote
 \begin{equation}\label{g(psi)} 
{g}_{\psi}:={{\eta(\psi)-\eta(0)}\over{\widetilde\gamma}}\sigma(\widetilde\gamma{Z}^{\psi}-\sigma{U}^{\psi}),
\end{equation}
then we have 
 \begin{equation}\label{Y(psi)potentials}
\begin{split}
e^{-B_t}Y^{\psi}_t=\mathbb{E}_0\left[e^{-B_T}\xi+\int_t^T e^{-B_s}g_{\psi}(s)ds\big|{\cal{F}}_t\right],\quad{t}\geq 0.
\end{split}
\end{equation}
  {\rm{(c)}}  Let $\psi_i\in\Psi$, $i=1,2$ and $\Gamma$ be a predictable set. If $\psi_3:=\psi_1{I}_{\Gamma}+\psi_2{I}_{\Gamma^c}$, then
 \begin{equation}\label{psi1psi2}
 \begin{split}
 & d(e^{-B}Y^{\psi_3})={I}_{\Gamma}d(e^{-B}Y^{\psi_1})+{I}_{\Gamma^c}d(e^{-B}Y^{\psi_2}),\\
& Z^{\psi_3}= Z^{\psi_1}{I}_{\Gamma}+Z^{\psi_2}{I}_{\Gamma^c},   \quad\mbox{and}\quad U^{\psi_3}=U^{\psi_1}{I}_{\Gamma}+U^{\psi_2}{I}_{\Gamma^c}.
 \end{split}
 \end{equation}
  {\rm{(d)}} for any $\psi_i\in\Psi$, $i=1,2$, there exists $\psi_3\in\Psi$ such that 
  \begin{equation}\label{Domination4psi}
  \begin{split}
  &\vert\psi_3\vert\leq \max\left(\vert\psi_1\vert,\vert\psi_2\vert\right),\quad Y^{\psi_3}\geq \max\left(Y^{\psi_1},Y^{\psi_2}\right),\quad g^{\psi}=\max\left(g^{\psi_1},g^{\psi_2}\right),\\
    \end{split}
  \end{equation}
  As a results $\{Y^{\psi}:\ \psi\in\Psi\}$, $\{g^{\psi}:\ \psi\in\Psi\}$, $\{Y^{\psi}:\ \psi\in\Psi_n\}$ and $\{g^{\psi}:\ \psi\in\Psi_n\}$ all upward directed.

     \end{lemma}
The proof of the lemma is relegated to Appendix \ref{Appendix4Section4}, while below we address the driver $g_{\psi}$ of the BSDE (\ref{YpsiBSDE}) and show it can be controlled uniformly when $\psi$ spans $\Psi_n$.  
\begin{lemma}\label{lemma4Gtilde} Suppose (\ref{mainassumtpion4S})  holds, and for any $(n,\psi)\in\mathbb{N}\times\Psi$ we consider ${g}_{\psi}$ defined in (\ref{Domination4psi}) and $(\widetilde{g}_n,g_n,\widetilde{g})$ given by 
\begin{equation}\label{Processes(g(n)gTilde(n)gTilde}
\begin{split}
&\widetilde{g}_n(t):=\underset{\psi\in\Psi_n}{\esssup}({g}_{\psi}(t))\geq 0,\quad\widetilde{g}(t):=\sup_{n\geq 0}\widetilde{g}_n(t)=\underset{\psi\in\Psi}{\esssup}({g}_{\psi}(t)),\quad t\geq0,\\
&g_n(t):={{\eta_t(-n)-\eta_t(0)}\over{\widetilde\gamma_t}}\sigma_t(\widetilde\gamma_t{Z}_t^{(n)}-\sigma_t{U}_t^{(n)})^+ +  {{\eta_t(0)-\eta_t(n)}\over{\widetilde\gamma_t}}\sigma_t(\widetilde\gamma_t{Z}_t^{(n)}-\sigma_t{U}_t^{(n)})^-,
\end{split}
\end{equation}
where $(Z^{(n)},U^{(n)})$ with $Y^{(n)}$ constitute the solution to (\ref{BSDE(n)}). Then the following assertions hold.\\
{\rm{(a)}} Let $\psi\in\Psi$, $n\in\mathbb{N}$,  and $\psi_n$  be given by (\ref{Psi(n)}). Then we have 
\begin{equation}\label{Y(psi)Y(n)potentials}
\begin{split}
Y^{\psi}_t=\mathbb{E}_0\left[\xi+\int_t^T g_{\psi}(s)ds\big|{\cal{F}}_t\right],\quad{Y}^{(n)}_t=\mathbb{E}_0\left[\xi+\int_t^T g_n(s)ds\big|{\cal{F}}_t\right]=Y^{\psi_n}_t.
\end{split}
\end{equation}
{\rm{(b)}} For any $n\in\mathbb{N}$, we have 
\begin{equation}\label{Assumption}
  \mathbb{E}_0\left[\int_0^T\widetilde{g}_n(t)dt\right]\leq\mathbb{E}_0\left[\int_0^T\underset{\psi\in\Psi_n}{\esssup}\vert{g}_{\psi}(t)\vert{d}t\right]=\sup_{\psi\in\Psi_n}\mathbb{E}_0\left[\int_0^T\vert{g}_{\psi}(t)\vert{d}t\right]=:C_n<\infty.
\end{equation}
{\rm{(c)}} For $n\in\mathbb{N}$, $\widetilde{g}_n={g}_n,$, $P\otimes{dt}$-a.e., and hence $(g_n)_n$ increases to $\widetilde{g}$, $P\otimes{dt}$-a.e.. \end{lemma}
This lemma is proved Appendix \ref{Appendix4Section4}, and the following lemma constitutes our last main technical step for the proof of Theorem \ref{BSDE4Y(up)}, and it elaborates some inequalities for the process $\widetilde{g}$. 
\begin{lemma}\label{LemmaforgTildeControl}
Consider the process $\widetilde{g}$ defined in (\ref{Processes(g(n)gTilde(n)gTilde}). Then the following assertions hold.\\
{\rm{(a)}} For any stopping time $\tau$, we have 
\begin{equation}\label{Domination4Gtilde}
\begin{split}
\mathbb{E}_0\left[\int_{\tau}^T\widetilde{g}(u)du\big|{\cal{F}}_{\tau}\right]\leq \mathbb{E}_0\left[\Delta+\xi^-\big|{\cal{F}}_{\tau}\right],\quad\mbox{where}\quad \Delta:=\underset{\psi\in\Psi}{\esssup}\sup_{0\leq{t}\leq{T}}\mathbb{E}_0[D^{\psi}_T(D^{\psi}_t)^{-1}\xi^+\big|{\cal{F}}_t].
\end{split}
\end{equation}
{\rm{(b)}} For $p\in(1,\infty)$, there exists a universal constant $C_p\in(0,\infty)$ such that 
\begin{equation}\label{Domination4Gtilde1}
\begin{split}
\mathbb{E}_0\left[\left(\int_0^T\widetilde{g}(u)du\right)^p\right]\leq C_p \mathbb{E}_0\left[(\xi^-)^p\right]+ C_p \sup_{\tau\in{\cal{T}},\psi\in\Psi}\mathbb{E}_0\left[{{D^{\psi}_T}\over{D^{\psi}_{\tau}}}(\xi^+)^p\right].
\end{split}
\end{equation}
\end{lemma}
The proof of this lemma is relegated to \ref{Appendix4Section4}, while below we prove Theorem  \ref{BSDE4Y(up)}.\\
\begin{proof}[Proof of Theorem \ref{BSDE4Y(up)}]
Remark that there is no loss of generality in assuming that $r\equiv0$, which will be enforced throughout this proof. The proof of the theorem is divided into two parts. The first part proves assertion (b) and (c), while the second part proves both assertions (a) and (d).\\
{\bf Part 1.} Hereto we prove assertion (b). To this end, for $s\in [0,T]$, $(y,z,u)\in\mathbb{R}^3$ and $n\in\mathbb{N}$, we put 
\begin{equation}\label{DRiver4Y(n)BSDE}
 \kappa_n(s,y,z,u):={{\eta_s(-n)-\eta_s(0)}\over{\widetilde\gamma_s}}\sigma_s(\widetilde\gamma_s{z}-\sigma_s{u})^+ +  {{\eta_s(0)-\eta_s(n)}\over{\widetilde\gamma_s}}\sigma_s(\widetilde\gamma_s{z}-\sigma_s{u})^-,
\end{equation}
It is clear that $ \kappa_n$ is the driver of the BSDE (\ref{BSDE(n)}). In virtue of (\ref{mainassumtpion4S}) and Lemma \ref{Lemma4Z(psi)Densities}, there exists a constant $C_n\in(0,\infty)$ such that 
\begin{equation*}
\begin{split}
&\vert\kappa_n(t,y_1,z_1,u_1)-\kappa_n(t,y_2,z_2,u_2)\vert\\
&\leq{C_n}\vert (\widetilde\gamma_s{z_1}-\sigma_s{u_1})^+-(\widetilde\gamma_s{z_2}-\sigma_s{u_2})^+\vert+C_n\vert (\widetilde\gamma_s{z_1}-\sigma_s{u_1})^--(\widetilde\gamma_s{z_2}-\sigma_s{u_2})^-\vert\\
&\leq{2C^2_n}\left(\vert{z_1}-{z_2}\vert+\vert{u_1}-{u_2}\vert+\vert{y_1}-{y_2}\vert\right).\end{split}
\end{equation*}
This proves that the driver $\kappa_n$ is Lipschitz. Hence, we deduce that there exists a unique solution to the BSDE (\ref{BSDE(n)}), denoted by $(\widetilde{Y}^{(n)}, \widetilde{Z}^{(n)},\widetilde{U}^{(n)})$ and belongs to the space  $ \mathbb{S}^p(R_0) \times \mathbb{L}^p(W,R_0) \times \mathbb{L}^p(N,R_0)$.
Thus, assertion (b) will follow immediately as soon as we prove that actually  we have $Y^{(n)}=\widetilde{Y}^{(n)}$.\\
To this end, in virtue of Lemma \ref{EtaProcess}, we remark that  
\begin{equation}\label{h(psi)-kappa(n)}
h_{\psi}(t,y,z,u)\leq \kappa_n(t,y,z,u),\quad\mbox{ for any}\ t\in[0,T],\ (y,z,u)\in\mathbb{R}^3,\ \psi\in\Psi_n.\end{equation}
 On the one hand, in virtue of the assumption (\ref{mainassumtpion4S}) and Lemma \ref{Lemma4Z(psi)Densities}, direct It\^o calculations allows us to deduce that for any $\psi\in\Psi_n$, 
 \begin{equation}\label{Dpsi}
{\cal{E}}\Bigl(\sigma(\eta(\psi)-\eta(0))\is{W}^0-{{\sigma^2}\over{\widetilde\gamma\lambda}}(\eta(\psi)-\eta(0))e^{-\eta(0)\zeta}\is\widetilde{N}^0\Bigr)={D}^{\psi}\in{\cal{M}}(R_0),
 \end{equation}
 and hence 
  \begin{equation}\label{Wpsi2}
  W^{\psi}=W^0-\int_0^{\cdot}\sigma_s(\eta_s(\psi)-\eta_s(0))ds\quad \mbox{and}\quad\widetilde{N}^{\psi}=\widetilde{N}^0+\int_0^{\cdot}{{\sigma_s^2}\over{\widetilde\gamma_s}}(\eta_s(\psi)-\eta_s(0))ds.
   \end{equation}
 On the other hand, we derive 
 \begin{equation*}\label{differnce4Y(n)Y(psi)}
 \begin{split}
 &\widetilde{Y}^{(n)}_t-Y^{\psi}_t\\
 &=\int_t^T\left(\kappa_n(s, \widetilde{Z}^{(n)},\widetilde{U}^{(n)})-h_{\psi}(t,Z^{\psi}_s,U^{\psi}_s)\right) ds-\int_t^T(\widetilde{Z}^{(n)}_s-Z^{\psi}_s)dW^0_s-\int_t^T(\widetilde{Z}^{(n)}_s-Z^{\psi}_s)d\widetilde{N}^0_s\\
 &\geq\int_t^T\left(h_{\psi}(s, \widetilde{Z}^{(n)}_s,\widetilde{U}^{(n)}_s)-h_{\psi}(t,Z^{\psi}_s,U^{\psi}_s)\right) ds-\int_t^T(\widetilde{Z}^{(n)}_s-Z^{\psi}_s)dW^0_s-\int_t^T(\widetilde{Z}^{(n)}_s-Z^{\psi}_s)d\widetilde{N}^0_s\\
 &=-\int_t^T(\widetilde{Z}^{(n)}_s-Z^{\psi}_s)\left(dW^0_s-\sigma_s(\eta_s(\psi)-\eta_s(0))ds\right)-\int_t^T(\widetilde{Z}^{(n)}_s-Z^{\psi}_s)\left(d\widetilde{N}^0_s+{{\sigma_s^2}\over{\widetilde\gamma_s}}(\eta_s(\psi)-\eta_s(0))ds\right)\\
 &=-\int_t^T(\widetilde{Z}^{(n)}_s-Z^{\psi}_s)d W^{\psi}_s-\int_t^T(\widetilde{Z}^{(n)}_s-Z^{\psi}_s)d\widetilde{N}^{\psi}_s.
 \end{split}
 \end{equation*}
 Thus, by taking conditional expectation under $R_{\psi}$ in the above inequality, we deduce that $\widetilde{Y}^{(n)}\geq Y^{\psi}$ for any $\psi\in\Psi_n$, and hence this yields $Y^{(n)}\leq\widetilde{Y}^{(n)}$. To prove the converse, we consider the triplet $(\widetilde{Y}^{(n)}, \widetilde{Z}^{(n)},\widetilde{U}^{(n)})$ solution to (\ref{BSDE(n)}), we denote $\widetilde\Omega:=\Omega\times[0,\infty)$, and put 
\begin{equation}
\begin{split}
\eta_n^*:={{\eta_s(-n)-\eta_s(0)}\over{\widetilde\gamma_s}}I_{\Gamma^+_n}+ {{\eta_s(0)-\eta_s(n)}\over{\widetilde\gamma_s}}I_{\widetilde\Omega\setminus\Gamma^+_n},\quad \Gamma^+_n:=\left\{\widetilde\gamma\widetilde{Z}^{(n)}-\sigma\widetilde{U}^{(n)}\geq 0\right\}.
\end{split}
\end{equation}
In virtue of the equation (\ref{Phi2eta}), direct calculations shows that 
\begin{equation}\label{Psi(n)}
\eta^*_n=\eta(\psi_n),\quad\mbox{where}\quad \psi_n:=-nI_{\Gamma^+_n}+nI_{\Gamma^-_n}\in\Psi_n,\end{equation}
and obtain $\kappa_n(s, \widetilde{Y}^{(n)}_s,\widetilde{Z}^{(n)}_s,\widetilde{U}^{(n)}_s)=h_{\psi_n}(s, \widetilde{Y}^{(n)}_s,\widetilde{Z}^{(n)}_s,\widetilde{U}^{(n)}_s)$. This implies that $(\widetilde{Y}^{(n)},\widetilde{Z}^{(n)},\widetilde{U}^{(n)})$ is a solution to the BSDE (\ref{YpsiBSDE}) when $\psi=\psi_n$. Therefore, the uniqueness of the solution of this latter BSDE is guaranteed by Lemma \ref{Lemma1}, and this yields  $\widetilde{Y}^{(n)}={Y}^{\psi_n}\leq Y^{(n)}$. This proves $\widetilde{Y}^{(n)}=Y^{(n)}$, and the proof of assertion (b) is complete. \\
{\bf Part 2.} This parts proves assertion (c). Due to the definition of the essential supremum, it is clear that $(Y^{(n)})_n$ is a nondecreasing sequence and 
\begin{equation}\label{ComparisonY(n)}
Y^{(n)}\leq Y^{(n+1)}\leq Y^{up},\quad\mbox{and}\quad \sup_n Y^{(n)}\leq Y^{up}.
\end{equation}
Furthermore, thanks to Lemma \ref{Lemma1} -(d), there exists a sequence $(\psi_k)_k\subset\Psi$ such that 
$$Y^{up}= \underset{\psi\in\Psi}{\esssup} Y^{\psi}=\displaystyle\lim_{k\longrightarrow\infty}Y^{\psi_k}_t.$$
Thus, as $\psi_k$ is bounded, there exists $n_k\in\mathbb{N}$ such that $\psi_k\in\Psi_{n_k}$, and we deduce 
$$Y^{up}=\lim_{k\longrightarrow\infty}Y^{\psi_k}\leq \sup_k Y^{(n_k)}\leq \sup_n Y^{(n)}.$$
A combination of this with (\ref{ComparisonY(n)}) yields the convergence almost surely of $Y^{(n)}$ to $Y^{up}$. This proves (\ref{Increasiness4Y(n)}).  To prove the remaining claim of assertion (c), for any $n\geq 0$, we consider $(Y^{(n)}, Z^{(n)},U^{(n)})$ the unique solution to (\ref{BSDE(n)}), and  $(Y^{\psi}, Z^{\psi},U^{\psi})$ is the unique solution to (\ref{YpsiBSDE}) for any $\psi\in\Psi$. 
\begin{equation}\label{BSDE4Y(psi)}
\begin{split}
dY^{\psi}_t=- {{\eta_t(\psi)-\eta_t(0)}\over{\widetilde\gamma_t}}\sigma_t(\widetilde\gamma_t{Z}_t^{\psi}-\sigma_t{U}_t^{\psi})dt+{Z}_s^{\psi}dW^0_t+{Z}_s^{\psi}d\widetilde{N}^0_t,\quad Y^{\psi}_T=\xi.
\end{split}
\end{equation}
Then, thanks to assertion (a) of this lemma,  for any $n,m\in\mathbb{N}$, we derive
\begin{equation}\label{GtildeVersusg(n)}
\begin{split}
 Y^{(n+m)}_t-Y^{(n)}_t&=E_0\left[\int_t^T\left(g_{n+m}(u)-g_n(u)\right)du \big|{\cal{F}}_t\right]\\
&=E_0\left[\int_0^T\left(g_{n+m}(u)-g_n(u)\right)du \big|{\cal{F}}_t\right]-\int_0^t\left(g_{n+m}(u)-g_n(u)\right)du.
\end{split}
\end{equation}
Thus, in virtue of Lemma \ref{lemma4Gtilde}-(c), we conclude that is in fact a nonnegative supermartingale under $R_0$. This ends the proof of assertion (c). \\
{\bf Part 3.} In this part, we prove the following properties.
\begin{equation}\label{ConvergenceResults}
\begin{split}
&i)\quad (Y^{(n)}, Z^{(n)}, U^{(n)})\quad\mbox{  converges in}\quad \mathbb{S}^p(R_0)\times\mathbb{L}^p(W,R_0) \times \mathbb{L}^p(N,R_0),\\
&ii)\quad \mbox{the limit of  $(Y^{(n)}$ coincides with the process $Y^{up}$},\\
&iii)\quad \mbox{there exists a subsequence of $(Y^{(n)}, Z^{(n)}, U^{(n)})$ that converges $P\otimes{dt}$-a.e.}.
\end{split}
\end{equation}
On the one hand, by combining Lemma \ref{LemmaforgTildeControl}-(b) and Lemma \ref{lemma4Gtilde}-(c) and the dominated convergence theorem, we deduce that under the assumption  (\ref{MainAssumption4Xi}), we have
\begin{equation}\label{NormConvergence4g(n)}
\sup_{m}\mathbb{E}_0\left[\left(\int_0^T(g_{n+m}(u)-g_{n}(u))du\right)^p\right]\longrightarrow 0\quad\mbox{and}\quad\mathbb{E}_0\left[\left(\int_0^T(\widetilde{g}(u)-g_{n}(u))du\right)^p\right]\longrightarrow 0.
\end{equation}
On the other hand, thanks to (\ref{GtildeVersusg(n)}), we get 
\begin{equation*}
\begin{split}
0\leq Y^{(n+m)}_t-Y^{(n)}_t=\mathbb{E}_0\left[\int_t^T\left(g_{n+m}(u)-g_n(u)\right)du \big|{\cal{F}}_t\right]\leq\mathbb{E}_0\left[\int_0^T\left(g_{n+m}(u)-g_n(u)\right)du \big|{\cal{F}}_t\right].
\end{split}
\end{equation*}
Then by combining this latter inequality with Doob's inequality, we deduce the existence of a universal constant $C_p\in(0,\infty)$ such that 
\begin{equation*}
\begin{split}
\sup_{m}\Vert Y^{(n+m)}-Y^{(n)}\Vert_{{\cal{S}}^p}^p\leq C_p\sup_{m}\mathbb{E}_0\left[\left(\int_0^T(g_{n+m}(u)-g_{n}(u))du\right)^p\right]\longrightarrow 0.
\end{split}
\end{equation*}
This implies that $(Y^{(n)})_n$ is a Cauchy sequence in $ \mathbb{S}^p(R_0)$-norm and hence it converges in this norm to its limit $\widetilde{Y}\in \mathbb{S}^p(R_0)$. By combining this with assertion (c) proved in part 2--which states that $Y^{(n)}$ converges almost surely pointe-wise to $Y^{up}$-- and the uniqueness of the limit, we obtain the convergence of $Y^{(n)}$ almost surely pointe-wise and in $ \mathbb{S}^p(R_0)$-norm to $Y^{up}$. This proves the property  ii)  in (\ref{ConvergenceResults}). \\
Now we address the convergence of the sequence $(Z^{(n)}, U^{(n)})$ in the space $\mathbb{L}^p(W,R_0) \times \mathbb{L}^p(N,R_0)$. Thus, by combining assertion (c) which guarantees that  $Y^{(n)}- Y^{(n+m)}$ is an $R_0$-submartingale,  and \cite[Lemma 1.9]{Stricker} applied to $Y^{(n)}- Y^{(n+m)}$, we obtain 
\begin{equation}\label{Norm4Z(n)U(n)}
\Vert{Z}^{(n+m)}-{Z}^{(n)} \Vert^p_{L^p(W,R_0)}+\Vert{U}^{(n+m)}-{U}^{(n)}\Vert^p_{L^p(N,R_0)}\leq \Vert Y^{(n+m)}-Y^{(n)}\Vert_{{\cal{S}}^p(R_0)}^p.
\end{equation}
As a result, we conclude the existence of a pair $({Z}^{up},{U}^{up})\in\mathbb{L}^p(W,R_0) \times \mathbb{L}^p(N,R_0)$ such that $(Y^{(n)}, Z^{(n)},U^{(n)})$  converges to $({Y}^{up}, {Z}^{up},{U}^{up})$ in the space ${\cal{S}}^p(R_0)\times\mathbb{L}^p(W,R_0) \times \mathbb{L}^p(N,R_0)$, and the property i) in (\ref{ConvergenceResults}) is proved.\\
Furthermore, if $(T_k)_k$ be a sequence of stopping times that increases to infinity and $N^{T_n}\leq k$, then for any $k\geq 1$ we have 
\begin{equation*}
\begin{split}
R_0\otimes{{dt}\over{T}}\left(\lambda{e}^{\eta(0)\zeta}\vert{U}^{up}-{U}^{(n)}\vert{I}_{\Lbrack0,T_k\Rbrack}>\epsilon\right)&\leq(T\epsilon)^{-1}\mathbb{E}_0\left[\int_0^{T_k}\lambda_s{e}^{\eta_s(0)\zeta_s}\vert{U}^{up}_s-{U}^{(n)}_s\vert{d}s\right]\\
&=(T\epsilon)^{-1}\mathbb{E}_0\left[\int_0^{T_k}\vert{U}^{up}_s-{U}^{(n)}_s\vert{d}N_s\right]\\
&\leq(T\epsilon)^{-1}\Vert{U}^{up}-{U}^{(n)}\Vert^p_{L^p(N,R_0)}\Vert\sqrt{N_{T_k}}\Vert_{L^p(R_0)}.
\end{split}
\end{equation*}
This proves that ${U}^{(n)}$ converges to ${U}^{up}$ in probability under $R_0\otimes({dt}/T)$, then there exists a subsequence $(Y^{(n)}, Z^{(n)},U^{(n)})$ that converges to $({Y}^{up}, {Z}^{up},{U}^{up})$ in both norm and $R_0\otimes{dt}$-a.e.. This proves the property iii) in (\ref{ConvergenceResults}), and part 3 is complete.\\
 {\bf Part 4.} Here we prove that $({Y}^{up}, {Z}^{up},{U}^{up})$ satisfies the constraints and the Skorokhod condition, i.e. the last two conditions in (\ref{ConstrainedBSDE1}), and prove assertion (d).
  To this end, consider the subsequence$({Z}^{(n)},{U}^{(n)})$ that converges to $({Z}^{up},{U}^{up})$ $R_0\otimes{dt}$-a.e., which is obtained in Part 3.  Thus, by combining this with the fact that ${{\eta_s(-n)-\eta_s(0)}/{\widetilde\gamma_s}}$ converges to $(\eta_s^*-\eta_s(0))/\widetilde\gamma_s$, $R_0\otimes{dt}$-a.e. due to Lemma \ref{EtaProcess} -(c), we deduce that  $(\eta_s(-n)-\eta_s(0)){\widetilde\gamma_s}^{-1}\sigma_s(\widetilde\gamma_s{Z}_s^{(n)}-\sigma_s{U}_s^{(n)})^+$ converges to $(\eta_s^*-\eta_s(0)){\widetilde\gamma_s}^{-1}\sigma_s(\widetilde\gamma_s\widetilde{Z}_s-\sigma_s\widetilde{U}_s)^+$, $R_0\otimes{dt}$-a.e.. Hence,  as $g_n$ increases to $\widetilde{g}$ in both ${\cal{A}}^p(R_0)$-norm and $R_0\otimes{dt}$-a.e. (see Lemmas \ref{lemma4Gtilde}  and \ref{LemmaforgTildeControl}), and 
\begin{equation*}
\begin{split}
g_n(t)=\underbrace{{{\eta_s(-n)-\eta_s(0)}\over{\widetilde\gamma_s}}\sigma_s(\widetilde\gamma_s{Z}_s^{(n)}-\sigma_s{U}_s^{(n)})^+}_{g_n^{(1)}(t)}+ \underbrace{ {{\eta_s(0)-\eta_s(n)}\over{\widetilde\gamma_s}}\sigma_s(\widetilde\gamma_s{Z}_s^{(n)}-\sigma_s{U}_s^{(n)})^-}_{g_n^{(2)}(t)}\\
\end{split}
\end{equation*}
we obtain 
\begin{equation*}
\begin{split}
\mathbb{E}_0\left[\int_0^T{{g_n^{(2)}(t)}\over{n}}dt\right]\leq {1\over{n}} \mathbb{E}_0\left[\int_0^T{g}_n(u)du\right]\leq {1\over{n}}\mathbb{E}_0\left[\int_0^T\widetilde{g}(u)du\right]\longrightarrow 0.
\end{split}
\end{equation*}
By combining this with Fatou's lemma, the convergence a.e. of $ (\eta_s(0)-\eta_s(n))\sigma_s(\widetilde\gamma_s{Z}_s^{(n)}-\sigma_s{U}_s^{(n)})^-/n\widetilde\gamma_s$ to $\zeta\sigma_s(\widetilde\gamma_s\widetilde{Z}_s-\sigma_s\widetilde{U}_s)^-/\widetilde\gamma$, which is due (\ref{limite4eta(n)/n}) in Lemma \ref{EtaProcess}-(d), and $\zeta\sigma/\widetilde\gamma>0$, we deduce that 
\begin{equation*}
\begin{split}
\widetilde\gamma_s\widetilde{Z}_s-\sigma_s\widetilde{U}_s\geq0 \quad P\otimes {dt}-a.e..
\end{split}
\end{equation*}
As both processes $\int_0^{\cdot}g_n(t)dt$ and $\int_0^{\cdot}{g}_n^{(1)}(t)$ converge in $p$-variation-norm (i.e. ${\cal{A}}^p(R_0)$-norm), we deduce that $K_n(t):=\int_0^t g_n^{(2)}(s)ds$ converges also in $p$-variation-norm. Hence we deduce the existence of a nonnegative and adapted process $k^{up}$ such that 
\begin{equation*}
\begin{split}
K^{(n)}:=\int_0^{\cdot} g_n^{(2)}(s)ds\quad\mbox{converges in ${\cal{A}}^p(R_0)$-norm to }\quad\int_0^{\cdot}{k}^{up}_sds=:K^{up}.\end{split}
\end{equation*}
Therefore, as $I_{\{\widetilde\gamma{Z}^{(n)}-\sigma{U}^{(n)}\geq\epsilon\}}$ converges $R_0\otimes{dt}$-almost everywhere to $I_{\{\widetilde\gamma\widetilde{Z}-\sigma\widetilde{U}\geq\epsilon\}}$ for any $\epsilon>0$, we deduce that for any $\epsilon>0$, we have 
\begin{equation*}
\begin{split}
0=\mathbb{E}_0\left[\left(I_{\{\widetilde\gamma{Z}^{(n)}-\sigma{U}^{(n)}\geq\epsilon\}}\is{K_n}\right)_T\right]\longrightarrow \mathbb{E}_0\left[\left(I_{\{\widetilde\gamma\widetilde{Z}-\sigma\widetilde{U}\geq\epsilon\}}\is{K}^{up}\right)_T\right].
\end{split}
\end{equation*}
This implies, when taking $\epsilon$ to zero, that 
\begin{equation*}
\begin{split}
\int_0^T I_{\{\widetilde\gamma_s\widetilde{Z}_s-\sigma_s\widetilde{U}_s>0\}}d{K}^{up}_s=0\quad P\mbox{-.a.s}.\quad \mbox{or equivalently}\quad\int_0^T (\widetilde\gamma_s{Z}_s^{up}-\sigma_s{U}_s^{up})d{K}^{up}_s=0\quad P\mbox{-.a.s}..
\end{split}
\end{equation*}
This proves simultaneously assertion (d) and the last two conditions in  (\ref{ConstrainedBSDE1}), and Part 4 is complete.\\
{\bf Part 5.} Thanks to Part 3, we consider a subsequence  $(Y^{(n)}, Z^{(n)},U^{(n)})$ that converges to $({Y}^{up}, {Z}^{up},{U}^{up})$ in both norm and $R_0\otimes{dt}$-a.e. and $(Z^{(n)}\is{W}^0_t,U^{(n)}\is\widetilde{N}^0_t)$ converges $P$-almost surely to $(Z^{up}\is{W}^0_t,U^{up}\is\widetilde{N}^0_t)$ for any $t\in[0,T]$. Thus, by taking the limit  $R_0\otimes{dt}$-a.e.in the BSDE (\ref{BSDE(n)}), or equivalently in 
\begin{equation*}
\begin{split}
Y_t^{(n)}=&\xi+\int_t^T {{\eta_s(-n)-\eta_s(0)}\over{\widetilde\gamma_s}}\sigma_s(\widetilde\gamma_s{Z}_s^{(n)}-\sigma_s{U}_s^{(n)})^+ds  +  K^{(n)}_T-K^{(n)}_t-\int_t^T Z_s^{(n)}{d}W^0_s-\int_t^T{U}_s^{(n)}d\widetilde{N}^0_s,
\end{split}\end{equation*}
we deduce that $({Y}^{up}, {Z}^{up},{U}^{up},{K}^{up})\in\mathbb{S}^p(R_0)\times\mathbb{L}^p(W,R_0) \times \mathbb{L}^p(N,R_0)$, fulfills indeed (\ref{ConstrainedBSDE1}) fully. \\
 To prove that this solution is the smallest, we suppose that $(Y,Z,U,K)$ is a solution to the constrained BSDE (\ref{ConstrainedBSDE1}).  Then remark that $(Y,Z,U,K)$ satisfies 
 \begin{equation*}
 dY={{\sigma}\over{\widetilde\gamma}}(\eta(0)-\eta^*)(\widetilde\gamma{Z}-\sigma{U})dt-dK+ZdW^0+Ud\widetilde{N}^0,\quad Y_T=\xi.
 \end{equation*}
 Then by combining this with (\ref{YpsiBSDE}), (\ref{Dpsi}) and (\ref{Wpsi2}), for any $\psi\in\Psi$, we derive
 \begin{equation*}
 \begin{split}
 &d(Y-Y^{\psi})\\
 &={{\sigma}\over{\widetilde\gamma}}(\eta(0)-\eta^*)(\widetilde\gamma{Z}-\sigma{U})dt+{{\sigma}\over{\widetilde\gamma}}(\eta(\psi)-\eta(0))(\widetilde\gamma{Z}^{\psi}-\sigma{U}^{\psi})dt+(Z-Z^{\psi})dW^0+(U-U^{\psi})d\widetilde{N}^0-dK\\
 &={{\sigma}\over{\widetilde\gamma}}(\eta(\psi)-\eta^*)(\widetilde\gamma{Z}-\sigma{U})dt-dK+(Z-Z^{\psi})\Bigl(dW^0-\sigma(\eta(\psi)-\eta(0))dt\Bigr)\\
 &\ +(U-U^{\psi})\Bigl(d\widetilde{N}^0-{{\sigma^2}\over{\widetilde\gamma}}(\eta(\psi)-\eta(0))dt\Bigr)\\
 &={{\sigma}\over{\widetilde\gamma}}(\eta(\psi)-\eta^*)(\widetilde\gamma{Z}-\sigma{U})dt+(Z-Z^{\psi})dW^{\psi}+(U-U^{\psi})d\widetilde{N}^{\psi}.
 \end{split}
 \end{equation*}
In virtue of ${\widetilde\gamma}^{-1}(\eta(\psi)-\eta^*)\leq0$, which is due to assertions (b) and (c) of Lemma \ref{EtaProcess},  and $\widetilde\gamma{Z}-\sigma{U}\geq 0$, we deduce that $Y-Y^{\psi}$ is a $R_{\psi}$-supermartingale with a null terminal value. This implies that 
 $$
 Y_t-Y^{\psi}_t\geq\mathbb{E}^{R_{\psi}}\left[Y_T-Y^{\psi}_T\ \Big|{\cal{F}}_t\right]=0.$$
 Therefore, we get $ Y_t\geq{Y}^{\psi}_t$ $P$-a.s. for any $\psi\in\Psi$ and hence we get $ Y_t\geq\displaystyle\esssup_{\psi\in\Psi}(Y^{\psi}_t)=Y^{(up)}_t.$ This proves that $Y^{up}$ is the smallest solution to the constraint BSDE  (\ref{ConstrainedBSDE1}). Hence assertion (a) follows immediately and the proof of the theorem is complete.\end{proof}
\appendix
\section*{Appendices}
\section{Local martingale deflators via predictable characteristics}
\begin{theorem}\label{Characgteristics4Deflator} Suppose that $S$ is locally bounded. Then the following assertions hold.\\
{\rm{(a)}} ${\cal{Z}}_{loc}(S)\not=\emptyset$ if and only if there exists a pair $(\beta, f)$ such $\beta$ is a $d$-dimensional predictable process and $f$ is real-valued positive ${\cal{P}}\otimes{\cal{B}}(\mathbb{R})$-measurable functional (i.e. $f>0$), 
\begin{equation}\label{DeflatorConditions}
\begin{split}
\beta^{tr}c\beta\is{A}+\sqrt{(f(x)-1)^2\star\widetilde\mu}\in{\cal{A}}^+_{loc},\quad\mbox{and}\quad \widetilde{b}'+c\beta+\int\left(xf(x)-x\right)\widetilde{F}(dx)=0.
\end{split}
\end{equation}
{\rm{(b)}} ${\cal{Z}}_{loc}^{L\log{L}}(S)\not=\emptyset$ if and only if there exists a pair $(\beta, f)$ such that (\ref{DeflatorConditions}) holds and 
\begin{equation}\label{Intregbility4Deflator}
\left(f(x)\ln(f(x))-f(x)+1\right)\star\widetilde\mu\in{\cal{A}}^+_{loc}.
\end{equation}
\end{theorem}
\section{Proofs for Lemmas  \ref{Deflator4S2Delfators4Shat},  \ref{Integrability4ThetaProcesses},  \ref{TechnicalLemma2} and \ref{Minimization2Root}}\label{Appendix4Section3}
\begin{proof}[Proof of Lemma \ref{Deflator4S2Delfators4Shat}] The proof of the lemma is achieved in two parts.\\
1) Here we prove  $Z:={\cal{E}}(\beta\is{X}^c+(g-1)\star(\mu-\nu))\in {\cal{Z}}_{loc}(S)$ iff $Z':={\cal{E}}(\beta\is{X}^c+(g-1)\star(\mu-\widehat\nu))\in{\cal{Z}}_{loc}(\widehat{S},\widehat{Z})$. To this end, on the one hand thanks to \cite[Lemma 2.4-(1)]{Choulli2007}, we deduce that  $Z\in {\cal{Z}}_{loc}(S)$ if and only if 
\begin{equation}\label{EquaDeflator1}
(e^{x}-1)(g-1)\star\mu\in{\cal{A}}_{loc},\quad\mbox{and}\quad \widetilde{b}'+c\beta+\int xf(x)(g(x)-1)F(dx)\equiv 0\quad P\otimes{dA}-a.e..
\end{equation}
On the other hand, by stopping we can assume without loss of generality that $\widehat{Z}\in{\cal{M}}$ and we put $\widehat{Q}:=\widehat{Z}_T\cdot P$. Then, in virtue of \cite[Lemma 2.4-(1)]{Choulli2007} again applied to $(\widehat{S},\widehat{Q})$, we conclude that $Z'\in {\cal{Z}}_{loc}(\widehat{S},\widehat{Q})$ if and only if $x(g-1)\star\mu\in {\cal{A}}_{loc}(\widehat{Q})$ and 
\begin{equation}\label{EquaDeflator2}
 b'+{{c}\over{2}}+\int(e^x-1-x)F(dx)+c\beta+ \int x(g(x)-1)f(x)F(dx)\equiv 0,\quad P\otimes{dA}-a.e..
\end{equation}
On the one hand, thanks to $\widetilde{b}'=b'+c/2+\int(e^x-1-x)F(dx)$, it is clear that the second equations in (\ref{EquaDeflator1}) and (\ref{EquaDeflator2}) coincide. On the other hand,  it is easy to prove that $x(g-1)\star\mu\in {\cal{A}}_{loc}(\widehat{Q})$ if and only if $(e^{x}-1)(g-1)\star\mu\in{\cal{A}}_{loc}$, and both these conditions hold. In fact,  we remark that 
\begin{equation*}
\begin{split}
\vert(e^{x}-1)(g-1)\vert\star\mu&=\sum \vert (e^{\Delta{X}}-1)(g(\Delta{X})-1)\vert\leq \sqrt{\sum  (e^{\Delta{X}}-1)^2}\sqrt{\sum(g(\Delta{X})-1)^2}\\
&\leq C\sqrt{[X,X]}\sqrt{(g-1)^2\star\mu}.\end{split}
\end{equation*}
Thus, as $[X,X]$ is locally bounded and  $\sqrt{(g-1)^2\star\mu}\in{\cal{A}}^+_{loc}$, we get $\sqrt{[X,X]}\sqrt{(g-1)^2\star\mu}\in{\cal{A}}^+_{loc}$, and hence $\vert(e^{x}-1)(g-1)\vert\star\mu\in{\cal{A}}^+_{loc}$. This ends the first part.\\
2) Here we prove that $Z\ln(Z)$ is a special semimartingale if and only if $\widehat{Z}Z'\ln(Z')$ is a special semimartingale. By stopping we can assume without loss of generality that $\widehat{Z}\in{\cal{M}}$ and put $\widehat{Q}:=\widehat{Z}_{\infty}\cdot{P}$. Then the problem reduces to prove that  $Z\ln(Z)$ is a special semimartingale if and only if $Z'\ln(Z')$ is a special semimartingale under $\widehat{Q}$. Thanks to the entropy-Hellinger process, defined in  \cite[Definition 4.3]{Choulli2005}, and \cite[Lemma 3.2]{Choulli2005}, the claim becomes $h^E(Z,P)\in{\cal{A}}^+_{loc}$ if and only if $h^E(Z',\widehat{Q})\in{\cal{A}}^+_{loc}$. Thus,  \cite[Proposition 3.5]{Choulli2005} yields  
\begin{equation}
h^E(Z,P)={1\over{2}}\beta{c}\beta\is A+\left(g\ln(g)-g+1\right)\star\nu\quad\mbox{and}\quad h^E(Z',\widehat{Q})={1\over{2}}\beta{c}\beta\is A+\left(g\ln(g)-g+1\right)\star\widehat\nu.
\end{equation}
Hence, as $\widehat\nu(dt,dx)=f(x)\nu(dt,dx)$ and $f$ is locally bounded and locally bounded away from zero, we deduce that the claim holds. This ends the proof of the lemma.
\end{proof}
\begin{proof}[Proof of Lemma \ref{Integrability4ThetaProcesses}:] Let $\theta\in L(Y)$, and consider the notation given in the lemma. The proof of the lemma will be achieved in three parts. \\
1) Here we prove assertion (a). To this end, on the one hand, we remark that the fact that $\theta\in L(Y)$ yields $I_{\{\vert\theta^{tr}x\vert>\epsilon\}}\star\mu^Y=\sum_{0<s\leq\cdot} I_{\{\vert\theta^{tr}_s\Delta{Y}_s\vert>\epsilon\}}
\in{\cal{A}}_{loc}^+$. Hence, we deduce that 
\begin{equation}\label{Integrability1bis}
b_1^{\theta}:=F^Y\left(\{x:\ \vert\theta^{tr}x\vert>\epsilon\}\right)<\infty\quad P\otimes{dA^Y}\mbox{-a.e.},\quad\mbox{and}\quad b_1^{\theta}\in L(A^Y).\end{equation}
On the other hand, we have 
$$
\int\Vert{x}\Vert{I}_{\{\vert\theta^{tr}x\vert>\epsilon\geq \Vert{x}\Vert\}}F^Y(dx)\leq\epsilon F^Y\left(\{x:\ \vert\theta^{tr}x\vert>\epsilon\}\right).$$
By combining this inequality with (\ref{Integrability1bis}) and $b\in L(A^Y)$,  we deduce that (\ref{ConditionIntegrability1}) follows immediately. To prove (\ref{ZdecompositionCannonical}), we appeal to the canonical decomposition of $Y$  (i.e. (\ref{Xcanon})) and use (\ref{Definition4U(theta)}) and derive
\begin{equation*}\begin{split}
Z^{\theta}(Y)&=X-U^{\theta}(Y)\\
&=Y_0 + Y^c + b^Y\is{A}^Y + h(x) \star (\mu^Y - \nu^Y) + \big(x - h(x) \big) \star \mu^Y-xI_{\{\vert\theta^{tr}x\vert>\epsilon\ \mbox{or}\quad \vert{x}\vert>\epsilon\}}\star\mu^Y\\
& = Y_0 + Y^c + b^Y\is{A}^Y + h(x) \star (\mu^Y - \nu^Y) - xI_{\{\vert{x}\vert\leq\epsilon<\vert\theta^{tr}x\vert\}}\star \mu^Y.
\end{split}
\end{equation*}
Remark that $ xI_{\{\vert{x}\vert\leq\epsilon<\vert\theta^{tr}x\vert\}}\star \mu^Y\in {\cal{A}}_{loc}$ and by compensating it and inserting the compensator process in the above equality we derive 
\begin{equation*}\begin{split}
Z^{\theta}(Y)&=Y_0 + Y^c + b^Y\is{A}^Y + h(x) \star (\mu^Y - \nu^Y) -xI_{\{\vert{x}\vert\leq\epsilon<\vert\theta^{tr}x\vert\}}\star (\mu^Y-\nu^Y)-xI_{\{\vert{x}\vert\leq\epsilon<\vert\theta^{tr}x\vert\}}\star\nu^Y\\
&=Y_0 + Y^c + b^Y\is{A}^Y +xI_{\{\vert{x}\vert\leq\epsilon\ \&\ \vert\theta^{tr}x\vert\leq\epsilon\}}\star (\mu^Y-\nu^Y)-xI_{\{\vert{x}\vert\leq\epsilon<\vert\theta^{tr}x\vert\}}\star\nu^Y.
\end{split}
\end{equation*}
This proves (\ref{ZdecompositionCannonical}), and ends the proof of assertion (a).\\
2) Here we prove assertion (b). To this end, we remark that it is clear that $\theta\in L(U_1^{\theta})\cap{L}(U_2^{\theta})$, and hence we deduce that $\theta\in L(Z^{\theta}(Y))$. Therefore, $\theta\in{L}(Y^c)\cap{L}(M^{\theta})\cap{L}(b^{\theta}\is{A}^Y)$ is a direct consequence of \cite[Corollaire 2.6]{Stricker} and the fact that both local martingales $Y^c$ and $M^{\theta}$ are orthogonal.
Furthermore, in virtue of  (\ref{ZdecompositionCannonical}), we get 
\begin{equation*}
\begin{split}
    \theta\is{Z}^{\theta}=\theta\is{Y^c}+\theta^{tr}xI_{\{\Vert{x}\Vert\leq\epsilon\ \&\ \vert\theta^{tr}x\vert\leq\epsilon\}}\star (\mu^Y-\nu^Y)+\theta^{tr}b^{\theta}\is{A}^Y.
\end{split}
\end{equation*}
This proves  (\ref{Z(theta)Decomposition})  and ends the proof of assertion (b).\\
3) This part proves assertion (c). To this end, we use (\ref{Z(theta)Decomposition}), and derive 
\begin{equation*}
\theta\is{Y}=\theta\is{Z}^{\theta}(Y)+\theta\is{U}^{\theta}(Y)=\theta\is{Y^c}+\theta^{tr}xI_{\{\vert{x}\vert\leq\epsilon\ \&\ \vert\theta^{tr}x\vert\leq\epsilon\}}\star (\mu^Y-\nu^Y)+\theta^{tr}b^{\theta}\is{A}^Y+\theta\is{U}^{\theta,1}(Y)+\theta\is{U}^{\theta,2}(Y).
\end{equation*}
Remark that $\theta\is{U}^{\theta,1}(Y)=\sum \theta^{tr}\Delta{Y}I_{\{\vert\theta^{tr}\Delta{Y}\vert\leq\epsilon<\Vert\Delta{Y}\Vert\}}=\theta^{tr}xI_{\{\vert\theta^{tr}x\vert\leq\epsilon<\Vert{x}\Vert\}}\star\mu^Y\in{\cal{A}}_{loc}$, and by compensating this process in the above equality we obtain
\begin{equation*}\begin{split}
\theta\is{Y}=&\theta\is{Y^c}+\theta^{tr}xI_{\{\Vert{x}\Vert\leq\epsilon\ \&\ \vert\theta^{tr}x\vert\leq\epsilon\}}\star (\mu^Y-\nu^Y)+\theta^{tr}b^{\theta}\is{A}^Y+ \theta^{tr}{x}I_{\{\vert\theta^{tr}x\vert\leq\epsilon<\Vert{x}\Vert\}}\star(\mu^Y-\nu^Y)\\
&+\theta^{tr}{x}I_{\{\vert\theta^{tr}x\vert\leq\epsilon<\vert{x}\vert\}}\star\nu^Y+\theta\is{U}^{\theta,2}(Y)\\
=&\theta\is{Y^c}+\theta^{tr}xI_{\{\vert\theta^{tr}x\vert\leq\epsilon\}}\star (\mu^Y-\nu^Y)
+\theta^{tr}b^{\theta}\is{A}^Y+\theta^{tr}{x}I_{\{\vert\theta^{tr}x\vert\leq\epsilon<\vert{x}\vert\}}\star\nu+\theta\is{U}^{\theta,2}(Y).
\end{split}
\end{equation*}
Thus, this equality yields the decomposition (\ref{Decomposition4Theta(X)}), and the proof of the lemma is complete.
\end{proof}
\begin{proof}[Proof of Lemma \ref{TechnicalLemma2}]1) This part proves assertion (a). Then we consider the following decomposition
\begin{equation*}
\begin{split}
 &e^{{\theta^{tr} x + x^{tr} \psi x}}- 1-h_{\epsilon}(\theta^{tr}x)\\
 &=\left( e^{{\theta^{tr} x + x^{tr} \psi x}}- 1-\theta^{tr}x\right)I_{\{\vert\theta^{tr}x\vert\leq\epsilon\}}+\left( e^{{\theta^{tr} x + x^{tr} \psi x}}- 1\right)I_{\{\vert\theta^{tr}x\vert>\epsilon\}}\\
 &=\left( e^{{\theta^{tr} x + x^{tr} \psi x}}- 1-\theta^{tr}x\right)I_{\{\vert\theta^{tr}x\vert\leq\epsilon\ \&\ \Vert{x}\Vert\leq\epsilon\}}+e^{{\theta^{tr} x + x^{tr} \psi x}}I_{\{\vert\theta^{tr}x\vert\leq\epsilon<\Vert{x}\Vert\}}- (1+\theta^{tr}x)I_{\{\vert\theta^{tr}x\vert\leq\epsilon<\Vert{x}\Vert\}}\\
 &\ - I_{\{\vert\theta^{tr}x\vert>\epsilon\}}+ e^{{\theta^{tr} x + x^{tr} \psi x}}I_{\{\vert\theta^{tr}x\vert>\epsilon\}}
\end{split}
\end{equation*}
Remark that $ I_{\{\vert\theta^{tr}x\vert>\epsilon\}}\star\widetilde\mu$, $ \vert1+\theta^{tr}x\vert{I}_{\{\vert\theta^{tr}x\vert\leq\epsilon<\Vert{x}\Vert\}}\star\widetilde\mu$, and $\vert{e}^{{\theta^{tr} x + x^{tr} \psi x}}- 1-\theta^{tr}x\vert{I}_{\{\vert\theta^{tr}x\vert\leq\epsilon\ \&\ \Vert{x}\Vert\leq\epsilon\}} \star\widetilde\mu$ belong to ${\cal{A}}_{loc}^+$. Therefore, we deduce that $\bigg| \exp({\theta^{tr} x + x^{tr} \psi x}) - 1-h_{\epsilon}(\theta^{tr}x) \bigg|\star\widetilde{\mu}  \in{\cal{A}}_{loc}^+$ if and only if 
\begin{equation}\label{Equa180}
 e^{{\theta^{tr} x + x^{tr} \psi x}}\left(I_{\{\vert\theta^{tr}x\vert>\epsilon\}}+I_{\{\vert\theta^{tr}x\vert\leq\epsilon<\Vert{x}\Vert\}}\right)\star\widetilde{\mu}  \in{\cal{A}}_{loc}^+.
\end{equation}
Similarly, we consider the following decomposition
\begin{equation*}
\begin{split}
 xe^{\theta^{tr}x+x^{tr}\psi{x}}-h_{\epsilon}(x)=&x\left(e^{\theta^{tr}x+x^{tr}\psi{x}}-1\right)I_{\{\Vert{x}\Vert\leq\epsilon\}}+xe^{\theta^{tr}x+x^{tr}\psi{x}}I_{\{\Vert{x}\Vert>\epsilon\}}\\
 =&x\left(e^{\theta^{tr}x+x^{tr}\psi{x}}-1\right)I_{\{\Vert{x}\Vert\leq\epsilon\ \&\ \vert\theta^{tr}x\vert\leq\epsilon\}}-xI_{\{\Vert{x}\Vert\leq\epsilon<\vert\theta^{tr}x\vert\}}\\
 &\ +xe^{\theta^{tr}x+x^{tr}\psi{x}}I_{\{\Vert{x}\Vert\leq\epsilon<\vert\theta^{tr}x\vert\}}+xe^{\theta^{tr}x+x^{tr}\psi{x}}I_{\{\Vert{x}\Vert>\epsilon\}}
\end{split}
\end{equation*}
Thanks to the boundedness of $\psi$, we deduce that $\Vert{x}\left(e^{\theta^{tr}x+x^{tr}\psi{x}}-1\right)\Vert{I}_{\{\Vert{x}\Vert\leq\epsilon\ \&\ \vert\theta^{tr}x\vert\leq\epsilon\}}\star\widetilde\mu$ and 
$\Vert{x}\Vert{I}_{\{\Vert{x}\Vert\leq\epsilon<\vert\theta^{tr}x\vert\}}\star\widetilde\mu$  belong to ${\cal{A}}^+_{loc}$. Hence, we conclude that $\Vert x\exp(\theta^{tr}x+x^{tr}\psi{x})-h_{\epsilon}(x)\Vert\star\widetilde{\mu} \in {\cal{A}}^+_{loc}$ if and only if 
\begin{equation}\label{Equa179}
\Vert{x}\Vert{e}^{\theta^{tr}x+x^{tr}\psi{x}}\left(I_{\{\Vert{x}\Vert>\epsilon\}}+I_{\{\Vert{x}\Vert\leq\epsilon<\vert\theta^{tr}x\vert\}}\right)\star\widetilde{\mu} \in {\cal{A}}^+_{loc}.
\end{equation}
As a result (\ref{IntegrabilityCondition4Linear4Proof}) holds if and only if  (\ref{Equa180}) and (\ref{Equa179}) hold, or equivalently  (\ref{IntegrabilityCondition4Linear}) is fulfilled. This ends the proof of assertion (a).\\
2) Here we prove assertion (b). To this end, throughout this part we suppose that  (\ref{IntegrabilityCondition4Linear})  holds. On the one hand, we combine this latter assumption with $({I}_{\{\Vert{x}\Vert>\epsilon\}}+I_{\{\vert\theta^{tr}x\vert>\epsilon\}})\star\widetilde\mu \in{\cal{A}}_{loc}^+$  and derive
\begin{equation}\label{Piece1}
\begin{split}
&\vert{e}^{\theta^{tr}x+x^{tr}\psi{x}}-1\vert{I}_{\{\Vert{x}\Vert>\epsilon\ \mbox{or}\ \vert\theta^{tr}x\vert>\epsilon\}}\star\widetilde\mu\\
&\leq  \vert{e}^{\theta^{tr}x+x^{tr}\psi{x}}-1\vert{I}_{\{\Vert{x}\Vert>\epsilon\}}\star\widetilde\mu+\vert{e}^{\theta^{tr}x+x^{tr}\psi{x}}-1\vert{I}_{\{\vert\theta^{tr}x\vert>\epsilon\}}\star\widetilde\mu\\
&\leq \left({{\Vert{x}\Vert}\over{\epsilon}}{e}^{\theta^{tr}x+x^{tr}\psi{x}}+1\right){I}_{\{\Vert{x}\Vert>\epsilon\}}\star\widetilde\mu +({e}^{\theta^{tr}x+x^{tr}\psi{x}}+1){I}_{\{\vert\theta^{tr}x\vert>\epsilon\}}\star\widetilde\mu\in{\cal{A}}_{loc}^+.
\end{split}
\end{equation}
On the other hand, as $\theta\in L(\widetilde{X})$ and hence $((\theta^{tr}x)^2 {I}_{\{\vert\theta^{tr}x\vert\leq \epsilon\}}\star\widetilde\mu \leq I_{\{ \vert\theta^{tr}\Delta\widetilde{X}\vert\leq\epsilon\}}\is [\theta\is \widetilde{X},\theta\is \widetilde{X}]\in{\cal{A}}_{loc}^+$, we deduce that 
\begin{equation}\label{Piece2}
({e}^{\theta^{tr}x+x^{tr}\psi{x}}-1)^2{I}_{\{\vert\theta^{tr}x\vert\leq \epsilon\ \&\ \Vert{x}\Vert\}}\star\widetilde\mu\leq C_{\epsilon}\left((\theta^{tr}x)^2+\vert{x}^{tr}\psi{x}\vert\right){I}_{\{\vert\theta^{tr}x\vert\leq \epsilon\ \&\ \Vert{x}\Vert\leq\epsilon\}}\star\widetilde\mu \in{\cal{A}}_{loc}^+.
\end{equation}
Thus, $\sqrt{({e}^{\theta^{tr}x+x^{tr}\psi{x}}-1)^2\star\widetilde\mu} \in{\cal{A}}_{loc}^+$ follows immediately from combining (\ref{Piece2}), (\ref{Piece1}) and the fact that 
$$
\sqrt{({e}^{\theta^{tr}x+x^{tr}\psi{x}}-1)^2\star\widetilde\mu}\leq \sqrt{({e}^{\theta^{tr}x+x^{tr}\psi{x}}-1)^2{I}_{\{\max(\vert\theta^{tr}x\vert, \Vert{x}\Vert)\leq \epsilon\}}\star\widetilde\mu}+\vert{e}^{\theta^{tr}x+x^{tr}\psi{x}}-1\vert{I}_{\{\Vert{x}\Vert>\epsilon\ \mbox{or}\ \vert\theta^{tr}x\vert>\epsilon\}}\star\widetilde\mu.$$
This proves the first claim of assertion (b).\\
3) Here we prove assertion (c). To this end we suppose that (\ref{IntegrabilityCondition4Linear})  and (\ref{mgEquation4Linear}) hold and 
\begin{equation}\label{Assumption4Psi}
\vert{e}^{x^{tr}\psi{x}}-1\vert\star\widetilde\mu\in{\cal{A}}_{loc}^+.\end{equation}
Remark that, in virtue of assertion (a),  it is clear that $\vert\theta^{tr}xe^{{\theta^{tr} x + x^{tr} \psi x}}  -e^{{\theta^{tr} x + x^{tr} \psi x}} + 1\vert\star\widetilde\mu\in{\cal{A}}_{loc}^+$ if and only if  $ \vert\theta^{tr}xe^{{\theta^{tr} x + x^{tr} \psi x}} -\theta^{tr}h_{\epsilon}(x)\vert\star\widetilde\mu\in{\cal{A}}_{loc}^+$ . Thus,  assertion (c) follows immediately as soon as we prove the latter property.

Then, on the one hand, remark that we always  have 
\begin{equation}\label{Decomposition1}
\theta^{tr}x{e}^{\theta^{tr}x+x^{tr}\psi{x}}-\theta^{tr}h_{\epsilon}(x)=e^{x^{tr}\psi{x}}\left((\theta^{tr}x){e}^{\theta^{tr}x}-{e}^{\theta^{tr}x}+1\right)+\left({e}^{\theta^{tr}x+x^{tr}\psi{x}}-1-\theta^{tr}h_{\epsilon}(x)\right)+(1-e^{x^{tr}\psi{x}}),
\end{equation}
and that assertion (a) combined with (\ref{IntegrabilityCondition4Linear}) implies that $\vert{e}^{\theta^{tr}x+x^{tr}\psi{x}}-1-\theta^{tr}h_{\epsilon}(x)\vert\star\widetilde\mu\in{\cal{A}}_{loc}^+$. Thus, this proves that 
thanks to (\ref{Assumption4Psi}) and (\ref{Decomposition1}), we conclude that 
\begin{equation}
 \int \left(\theta^{tr}x{e}^{\theta^{tr}x+x^{tr}\psi{x}}-\theta^{tr}h_{\epsilon}(x)\right)\widetilde{F}(dx)=\theta^{tr}\widetilde{b}+\theta^{tr}c\theta\in L(A)\quad \mbox{iff}\quad \vert{e}^{\theta^{tr}x+x^{tr}\psi{x}}-1-\theta^{tr}h_{\epsilon}(x)\vert\star\widetilde\mu\in{\cal{A}}_{loc}^+.
\end{equation}
On the other hand, by combining $\theta\in L(S)$ and the fact that (\ref{IntegrabilityCondition4Linear}) implies that $\theta\in L^1_{loc}(h_{\epsilon}(x)\star(\widetilde\mu-\widetilde\nu))\cap L((x-h_{\epsilon}(x))\star\widetilde\mu)$, we deduce that $\theta\in{L}(\widetilde{b}\is{A})$ or equivalently $\theta^{tr}\widetilde{b}\in L(A)$. Furthermore, $\theta^{tr}c\theta\is A\leq I_{\{\vert\theta^{tr}\Delta\widetilde{X}\vert\leq \epsilon\}}\is [\theta\is\widetilde{X},\theta\is\widetilde{X}]\in{\cal{A}}_{loc}^+$, we get $\theta^{tr}c\theta\in{L}(A)$. Thus, by combining all these and 
 $$\int \left(\theta^{tr}x{e}^{\theta^{tr}x+x^{tr}\psi{x}}-\theta^{tr}h_{\epsilon}(x)\right)\widetilde{F}(dx)=\theta^{tr}\widetilde{b}+\theta^{tr}c\theta,$$
 which is due to (\ref{mgEquation4Linear}), it clear that assertion (c) follows immediately. \\
4) Here we prove assertion (d). Thus, we remark that 
\begin{equation*}
\begin{split}
&\left(e^{\theta^{tr}x+x^{tr}\psi{x}}-1\right)({\bf{e}}(x)-\mathbb{I}_d)\\
&=\left(e^{\theta^{tr}x+x^{tr}\psi{x}}-1\right)({\bf{e}}(x)-\mathbb{I}_d)I_{\{\vert\theta^{tr}x\vert>\alpha\}}+\left(e^{\theta^{tr}x+x^{tr}\psi{x}}-1\right)({\bf{e}}(x)-\mathbb{I}_d)I_{\{\vert\theta^{tr}x\vert\leq\alpha\}}\\
&=e^{\theta^{tr}x+x^{tr}\psi{x}}({\bf{e}}(x)-\mathbb{I}_d)I_{\{\vert\theta^{tr}x\vert>\alpha\}}-({\bf{e}}(x)-\mathbb{I}_d)I_{\{\vert\theta^{tr}x\vert>\alpha\}}\\
&+\left(e^{\theta^{tr}x+x^{tr}\psi{x}}-1\right)({\bf{e}}(x)-\mathbb{I}_d)I_{\{\vert\theta^{tr}x\vert\leq\alpha\}}
\end{split}
\end{equation*}
As $S$ and $\psi$ are locally bounded, it is clear that 
$$\Vert({\bf{e}}(x)-\mathbb{I}_d\Vert{I}_{\{\vert\theta^{tr}x\vert>\alpha\}}\star\mu+\Vert\left(e^{\theta^{tr}x+x^{tr}\psi{x}}-1\right)({\bf{e}}(x)-\mathbb{I}_d)\Vert{I}_{\{\vert\theta^{tr}x\vert\leq\alpha\}}\star\mu\in{\cal{A}}^+_{loc}.$$
Furthermore, $e^{\theta^{tr}x+x^{tr}\psi{x}}\Vert{\bf{e}}(x)-\mathbb{I}_d\Vert{I}_{\{\vert\theta^{tr}x\vert>\alpha\}}\star\mu\in {\cal{A}}^+_{loc}$ if and only if 
\begin{equation}\label{Equa199}
e^{\theta^{tr}x}\Vert{\bf{e}}(x)-\mathbb{I}_d\Vert{I}_{\{\vert\theta^{tr}x\vert>\alpha\}}\star\mu\in {\cal{A}}^+_{loc}.\end{equation}
\begin{equation*}
\begin{split}
 e^{{\theta^{tr} x + x^{tr} \psi x}} - 1-h_{\epsilon}(\theta^{tr}x)&=\left(e^{{\theta^{tr} x + x^{tr} \psi x}} - 1-\theta^{tr}x\right)I_{\{ \vert\theta^{tr}x\vert\leq\epsilon\}}+\left(e^{{\theta^{tr} x + x^{tr} \psi x}} - 1\right)I_{\{ \vert\theta^{tr}x\vert>\epsilon\}}\\
 &=\left(e^{{\theta^{tr} x + x^{tr} \psi x}} - 1-\theta^{tr}x\right)I_{\{ \vert\theta^{tr}x\vert\leq\epsilon\}}- I_{\{ \vert\theta^{tr}x\vert>\epsilon\}}+e^{{\theta^{tr} x + x^{tr} \psi x}} I_{\{ \vert\theta^{tr}x\vert>\epsilon\}}
 \end{split}
 \end{equation*}
 On the one hand, it is clear that $I_{\{ \vert\theta^{tr}x\vert>\epsilon\}}\star\mu$ and $\left(e^{{\theta^{tr} x + x^{tr} \psi x}} - 1-\theta^{tr}x\right)I_{\{ \vert\theta^{tr}x\vert\leq\epsilon\}}\star\mu$ belongs to ${\cal{A}}_{loc}$. On the other hand, $e^{{\theta^{tr} x + x^{tr} \psi x}} I_{\{ \vert\theta^{tr}x\vert>\epsilon\}}\star\mu\in {\cal{A}}_{loc}^+$ if and only if 
 \begin{equation}\label{Equa200}
 e^{\theta^{tr} x} I_{\{ \vert\theta^{tr}x\vert>\epsilon\}}\star\mu\in {\cal{A}}_{loc}^+.
 \end{equation}
 Thanks to the local boundedness of $S$, it clear that (\ref{Equa200}) implies (\ref{Equa199}). Thus, this proves that (\ref{Equa198}) is equivalent to (\ref{IntegrabilityCondition4EE}), and the proof of assertion (b) is complete.
\end{proof}
\begin{proof}[Proof of Lemma \ref{Minimization2Root}] For any $\theta$ and any $\epsilon\in\mathbb{R}$, we put $\overline\theta:=\widehat\theta+\epsilon(\theta-\widehat\theta)=:\widehat\theta+\epsilon\Delta$. Therefore, as $\widehat\theta$ is solution to the minimization problem almost surely, then we derive
\begin{equation*}
\begin{split}
0\leq f(\overline\theta)-f(\widehat\theta)=\epsilon\Delta^{tr}c\widehat\theta+{{\epsilon^2}\over{2}}\Delta^{tr}c\Delta+\epsilon\Delta^{tr}\widetilde{b}'+\int\left( e^{\widehat\theta^{tr}x+x^{tr}\psi{x}}(e^{\epsilon\Delta^{tr}x}-1)-\epsilon\Delta^{tr}x\right)\widetilde{F}(dx).
\end{split}
\end{equation*}
Then by distinguishing the cases whether $\epsilon$ is positive or negative, dividing with it, and taking the limit afterwards, we obtain
\begin{equation*}
0=\Delta^{tr}\widetilde{b}'+\Delta^{tr}c\widehat\theta+\int\left( \Delta^{tr}x{e}^{\widehat\theta^{tr}x+x^{tr}\psi{x}} -\Delta^{tr}x\right)\widetilde{F}(dx).
\end{equation*}
Thus, as $\Delta=\theta-\widehat\theta$ spans $\mathbb{R}^d$, we conclude that $\widehat\theta$ is a solution to (\ref{Equation(3.7)}).
\end{proof}
\section{Proofs for lemmas of Section \ref{Section4pricing}}\label{Appendix4Section4}
\begin{proof}[Proof of Lemma \ref{Lemma4Z(psi)Densities} ] For any $\psi\in\Psi$, we denote 
$$\overline{M}^{\psi}:=\eta(\psi)\sigma\is{W}+\left(\exp(\eta(\psi)\zeta+\zeta^2\psi)-1\right)\is\widetilde{N}\in{\cal{M}}_{loc}.$$
Then we calculate the predictable quadratic variation of $\overline{M}^{\psi}$, under $P$, as follows
$$
\langle\overline{M}^{\psi}\rangle=\int_0^{\cdot}\left\{\eta(\psi)_s^2\sigma_s^2+\lambda_s\left(\exp(\eta_s(\psi)\zeta_s+\zeta^2_s\psi_s)-1\right)^2\right\}ds.$$
Therefore, thanks to Lemma \ref{EtaProcess}-(a), for any $\psi\in\Psi$, there exists $C_{\psi}\in(0,\infty)$ such that $\langle\overline{M}^{\psi}\rangle\leq C_{\psi}T$, and hence $\langle\overline{M}^{\psi}\rangle$ is bounded. This yields $\overline{Z}^{\psi}\in{\cal{M}}$, for an y $\psi\in\Psi$, and assertion (a) is proved. Assertion (b) follows from combining Girsanov's theorem and (\ref{root4eta(psi)}).\\
Thanks to Lemma \ref{EtaProcess} again, we deduce that there exists $C'_n\in(0,\infty)$ such that 
$$
\sup_{\psi\in\Psi_n}\vert{{\eta(\psi)}\over{\widetilde\gamma}}\vert\leq \sup_{\psi\in\Psi_n}\vert{{\eta(\psi)-\eta(0)}\over{\widetilde\gamma}}\vert+\vert{{\eta(0)}\over{\widetilde\gamma}}\vert\leq {{\eta(-n)-\eta(n)}\over{\widetilde\gamma}}++\vert{{\eta(0)}\over{\widetilde\gamma}}\vert\leq C_n'.
$$
Thus, by combining this with (\ref{mainassumtpion4S}) and the Hellinger process of order $p$
 \begin{equation*}
 \begin{split}
 {{dh^{(p)}_t(D^{\psi},R_0)}\over{dt}}=&{1\over{2}}(\eta_t(\psi)-\eta_t(0))^2\sigma_t^2 \\
 &+\lambda_t{e}^{\eta_t(0)\zeta_t}{{\exp(p(\eta_t(\psi)-\eta_t(0))\zeta_t+p\zeta^2_t\psi_t)-1-p(\exp\left((\eta_t(\psi)-\eta_t(0))\zeta_t+\zeta^2_t\psi_t\right)-1)}\over{p(p-1)}},
 \end{split}
 \end{equation*}
 which is due to \cite[Proposition 3.5]{Choulli2007}, we deduce the existence of $C_n\in(0,\infty)$ such that 
$$
\vert\eta(\psi)\vert+h^{(p)}(D^{\psi},R_0)\leq C_n,\quad \mbox{for any}\quad \psi\in\Psi_n,\quad P\otimes{dt}\mbox{-a.e.}.$$
Hence, thanks to \cite[Theorem 3.9]{Choulli2007}, this latter property implies the existence of $C''_n\in (0,\infty)$ such that 
$$E_0\left[\left({{D^{\psi}_T}\over{D^{\psi}_t}}\right)^p\Big|{\cal{F}}_t\right]\leq C''_n,\quad P\mbox{-a.s..}$$ This proves assertion (c), and the proof of the lemma is complete.
\end{proof}
\begin{proof}[Proof of Lemma \ref{EtaProcess}] 1) the continuity of $\Phi$ is obvious, while $\Phi(\infty)=\infty$, $\Phi(-\infty)=-\infty$ and $\Phi'(x)=\sigma^2+\lambda\widetilde\gamma\zeta{e}^x>0$ follow from $\lambda>0$, $\sigma>0$ and $\zeta\widetilde\gamma>0$. Thus, $\Phi^{-1}$ (the inverse function) exists, and direct calculations afterwards yield (\ref{Phi2eta}). Furthermore, in virtue of (\ref{mainassumtpion4S}), for any $\psi\in\Psi$ there exists $C'_{\psi}\in(0,\infty)$ such that $\vert \zeta\psi\vert+\vert\zeta(r-\widetilde{b}+\lambda\widetilde\gamma)+\sigma^2\zeta^2\psi\vert\leq C'_{\psi}$ $P\otimes{dt}$-a.e., and for any $C_1>0$, there exist $C_2\in(0,\infty)$ such that $-C2\leq\Phi^{-1}(C_1)\leq C_2$. Thus, by combining these, we deduce the existence of $C_{\psi}\in(0,\infty)$ such that $\vert \eta(\psi)\vert\leq C_{\psi}$ $P\otimes{dt}$-a.e.. This proves assertion (a).\\
2) To prove assertion (b), we remark that due to $\zeta\widetilde\gamma>0$, we have 
$${{\partial\eta(\psi)}\over{\widetilde\gamma\partial\psi}}={{\zeta}\over{\widetilde\gamma}}\left({{\sigma^2}\over{\sigma^2+\lambda\widetilde\gamma\zeta\exp(\Phi^{-1}(\Delta_{\psi}))}}-1\right)<0,\quad\mbox{where}\quad\Delta_{\psi}:=\zeta(r-\widetilde{b}+\lambda\widetilde\gamma)+\sigma^2\zeta^2\psi.$$
This proves assertion (b). \\
3) Thanks to assertion (b), we know that $\eta(\psi)/\widetilde\gamma$ is decreasing with respect to $\psi$, thus both 
$$l_{\infty}:=\lim_{x\longrightarrow-\infty}{{\eta(x)}\over{\widetilde\gamma}}={{\eta(-\infty)}\over{\widetilde\gamma}}\quad\mbox{and}\quad \bar{l}_{\infty}:=\lim_{x\longrightarrow\infty}{{\eta(x)}\over{\widetilde\gamma}}={{\eta(\infty)}\over{\widetilde\gamma}},$$
exist  and belong to $[-\infty,\infty]$. On the one hand,  remark that due to (\ref{root4eta(psi)}), we have 
$$\eta(\psi)/{\widetilde\gamma}\leq\left(r-\widetilde{b}+\lambda\widetilde\gamma\right)/{\widetilde\gamma}\sigma^2=:\Delta^{up}<\infty,\quad \mbox{for any}\quad \psi.$$
Thus, by combining this latter fact with assertion (b), we deduce that $-\infty<\eta(0)/{\widetilde\gamma}\leq l_{\infty}\leq \Delta^{up}<\infty$, and by letting $\psi$ to go to $-\infty$ in the equation (\ref{root4eta(psi)}), we obtain the first equality in (\ref{etaLimitsetaStar}). To prove the second equality in (\ref{etaLimitsetaStar}), we remark that $ \bar{l}_{\infty}\leq \Delta^{up}<\infty$, and we use again the equation (\ref{root4eta(psi)}) afterwards to contradict the assumption $\bar{l}_{\infty}>-\infty$. This ends the proof of assertion (c).\\
4) Notice that the second inequality in (\ref{limite4eta(n)/n}) is direct consequence of 
$${{\eta(0)-\eta(\psi)}\over{\widetilde\gamma}}={{\zeta\psi}\over{\widetilde\gamma}}\left(1-{{\sigma^2}\over{\sigma^2+\lambda\widetilde\gamma\zeta\exp(\Delta_{\psi}')}}\right)\leq{{\zeta\psi}\over{\widetilde\gamma}},\quad\mbox{where}\ \Delta_{\psi}'\in[\Delta_{0},\Delta_{\psi}],\quad\mbox{for any}\quad \psi\geq 0.$$
The first equality (\ref{limite4eta(n)/n}) follows from  (\ref{Phi2eta}), which yields 
$${{\eta(n)}\over{n\widetilde\gamma}}=-{{\zeta}\over{\widetilde\gamma}}+{{\Phi^{-1}(\Delta_0+\zeta^2\sigma^2{n})}\over{n\zeta\widetilde\gamma}},$$
and the fact that $\Phi(x)/x$ goes to infinity with $x$ (or equivalently $\Phi^{-1}(y)/y$ goes to zero when $y$ goes to infinity. This ends the proof of the lemma.
\end{proof}
\begin{proof}[Proof of Lemma \ref{Lemma1}] 1) Here that for any $n\geq 1$, there exist $C_n,\overline{C}_n\in(0,\infty)$ such that 
\begin{equation}\label{Control4Y(psi)}
\Vert\sup_{0\leq{t}\leq{T}}\vert{Y}^{\psi}_t\vert\Vert_{L^p(R_0)}\leq C_n\Vert\xi\Vert_{L^p(R_0)},\quad\mbox{for any}\quad \psi\in\Psi_n,\end{equation}
and assertion (a) afterwards. To this end, we consider $n\geq 1$, and use Lemma \ref{Lemma4Z(psi)Densities}-(c) and deduce  $Z^{\psi}\in {\cal{R}}_q(R_0)$, where $q$ is the conjugate of $p$. Thus, by combining this fact with \cite[Theorem 2.10]{Choulli99}, $Y^{\psi}_T=\xi$ and $\exp(-\int_0^{\cdot}r_s ds)Y^{\psi}$ is a martingale under $R_{\psi}$, the inequality (\ref{Control4Y(psi)}) follows immediately. Then due to the assumption (\ref{mainassumtpion4S}), the inequality implies that $Y^{\psi}$ is a special semimartingale under $R_0$, and its canonical decomposition is denoted by 
\begin{equation}\label{Doob4Y(psi)}
Y^{\psi}=Y^{\psi}_0+M^{\psi}+A^{\psi},\quad M^{\psi}\in{\cal{M}}_{loc}(R_0),\quad A^{\psi}\in{\cal{A}}_{loc}(R_0)\ \mbox{and is predictable}.
\end{equation}
Thus, again Lemma \ref{Lemma4Z(psi)Densities}-(c) guarantees that $M^{\psi}\in\mbox{bmo}_q$ and hence thanks to \cite[Theorem 4.5]{Choulli99}, we conclude the existence of $\overline{C}_n\in(0,\infty)$ such that 
\begin{equation}\label{Control4Y(psi)1}
\Vert{Y}^{\psi}_0\Vert_{L^p(R_0)}+\Vert[M^{\psi},M^{\psi}]^{1/2}_T\Vert_{L^p(R_0)}+ \Vert\mbox{Var}_T(A^{\psi})\Vert_{L^p(R_0)}\leq \overline{C}_n  \Vert\sup_{0\leq{t}\leq{T}}\vert{Y}^{\psi}_t\vert\Vert_{L^p(R_0)}.
\end{equation}
By combining (\ref{Control4Y(psi)}) and (\ref{Control4Y(psi)1}), under assumption (\ref{mainassumtpion4S}), we obtain that $M^{\psi}\in{\cal{M}}^p(R_0)$, and due to the martingale representation theorem, we obtain a unique pair $(Z^{\psi},U^{\psi})\in\mathbb{L}^p(W,R_0) \times \mathbb{L}^p(N,R_0)$ such that 
\begin{equation}\label{Y(psi)Representation}
Y^{\psi}=Y^{\psi}_0+Z^{\psi}\is{W}^0+U^{\psi}\is\widetilde{N}^0+A^{\psi}.\end{equation}
Then, as  $\exp(-\int_0^{\cdot}r_s ds)Y^{\psi}\in{\cal{M}}_{loc}(R_{\psi})$, the above equality and  Girsanov's theorem yield
\begin{equation}\label{A(psi)}
A^{\psi}=\int_0^{ \cdot}\left((\eta_s(\psi)-\eta_s(0)){{\sigma_s}\over{\widetilde\gamma_s}}(\widetilde\gamma_s{Z}_s^{\psi}-\sigma{U}_s)-r_s Y^{\psi}_s\right) ds.\end{equation}
On the one hand, by combining this latter claim with (\ref{Control4Y(psi)1}), (\ref{Doob4Y(psi)}) and  (\ref{Control4Y(psi)}), the inequality (\ref{Control4(Y,Z,U,A)(psi)}) follows immediately. On the other hand, by combining  (\ref{A(psi)}), (\ref{Y(psi)Representation}) and $Y^{\psi}_T=\xi$, the second equality in (\ref{YpsiBSDE}) follows immediately, while the first equality is a direct consequence of the second one and Girsanov's theorem. This ends the proof of assertion (a).\\
2) Here we prove assertion (b). Then  remark that for $\psi\in\Psi$ and denoted by $(Y^{\psi},Z^{\psi},U^{\psi})$ the solution to (\ref{YpsiBSDE}), the processes $Z^{\psi}\is{W}^0$  and $U^{\psi}\is\widetilde{N}^0$ are an $R_0$-martingales (due to assertion (a)).  Thus, by taking conditional expectation under $R_0$ in both sides of  (\ref{YpsiBSDE}), the equality (\ref{Y(psi)potentials}) follows immediately.\\
3) This part proves assertions (c) and (d). To this end,  we remark that assertion (d) is a direct consequence of assertion (c). In fact, for any $\psi_i\in\Psi$, $i=1,2$, we put $\Gamma:=\{g^{\psi_1}>g^{\psi_2}\}$, and then $\psi_3:=\psi_1{I}_{\Gamma}+\psi_2{I}_{\Gamma^c}$ satisfies (\ref{Domination4psi}). As a direct consequence of these inequalities, we deduce that all the four sets of assertion (d) are upward directed. This proves assertion (d), and the rest of this proof focuses on proving assertion (c). Therefore, on the one hand, we consider $\psi_i\in\Psi$, $i=1,2$, $\Gamma$ be a predictable set, and define the quadruplet $(\psi_3,\overline{Y},\overline{Z},\overline{U})$ of processes by 
\begin{equation}\label{YbarZbar}
\begin{split}
&\psi_3:=\psi_1{I}_{\Gamma}+\psi_2{I}_{\Gamma^c},\quad d\overline{Y}:=I_{\Gamma}d(e^{-B}Y^{\psi_1})+I_{\Gamma^c}d(e^{-B}Y^{\psi_2}),\quad \overline{Y}_T:=e^{-B_T}\xi,\\
&\overline{Z}:=e^{-B}Z^{\psi_1}I_{\Gamma}+e^{-B}Z^{\psi_2}I_{\Gamma^c},\quad \overline{U}:=e^{-B}U^{\psi_1}I_{\Gamma}+e^{-B}U^{\psi_2}I_{\Gamma^c}.
\end{split}
 \end{equation}
 On the other hand, for any $\psi\in\Psi$, the triplet $(e^{-B}Y^{\psi},e^{-B}Z^{\psi},e^{-B}U^{\psi})$ is the unique solution to the following BSDE
 \begin{equation}\label{BSDE4e(-B)Y(psi)}
 Y_t=e^{-B_T}\xi+\int_t^T(\eta_s(\psi)-\eta_s(0)){{\sigma_s}\over{\widetilde\gamma_s}}\left(\widetilde\gamma_s{Z}_s-\sigma_s{U}_s\right)ds -\int_t^T Z_s dW^0_s-\int_t^T U_s d\widetilde{N}^0_s
 \end{equation}
 Thus, in virtue of the notation in (\ref{YbarZbar}) and using (\ref{BSDE4e(-B)Y(psi)}) for each $e^{-B}Y^{\psi_i}$, $i=1,2$, we easily derive 
 $$
 \overline{Y}_t:=e^{-B_T}\xi+\int_t^T (\eta_s(\psi)-\eta_s(0)){{\sigma_s}\over{\widetilde\gamma_s}}\left(\widetilde\gamma_s\overline{Z}_s-\sigma_s\overline{U}_s\right)ds-\int_t^T \overline{Z}_s dW^0_s-\int_t^T \overline{U}_s d\widetilde{N}^0_s.$$
This proves that $( \overline{Y},\overline{Z} ,\overline{U} ) $ is a solution to (\ref{BSDE4e(-B)Y(psi)}) when $\psi=\psi_3$. Then due to the uniqueness of this BSDE, we deduce that $( \overline{Y},\overline{Z} ,\overline{U} ) $ coincides with $(e^{-B}Y^{\psi_3} ,e^{-B}Z^{\psi_3} ,e^{-B}U^{\psi_3} )$, and  hence (\ref{psi1psi2}) follows immediately. This ends the proof of assertion (c), and completes the proof of the lemma.
\end{proof}
\begin{proof}[Proof of Lemma \ref{lemma4Gtilde} ] It is clear that assertion (a) follows immediately from Lemma \ref{Lemma1} (see assertion (b), (c) and (d)), while assertion (b)  follows from combining assertions (a) , (b) and (d) of Lemma \ref{Lemma1} again. Thus, the rest of this proof addresses assertion (c).\\ 
As $\widetilde{g}\geq 0$, on the other hand, it is clear that 
$$\widetilde{g}=\left(\underset{\psi\in\Psi_n}{\esssup}({g}_{\psi})\right)^+=\underset{\psi\in\Psi_n}{\esssup}({g}_{\psi}^+).$$  On the other hand, in virtue of Lemma \ref{Lemma1}-(d), there exists a sequence $(\psi_k)_k\subset\Psi_n$ such that $(g_{\psi_k})^+$ converges $P\otimes{dt}$-almost everywhere to $\widetilde{g}$, and hence, due to Fatou, we get
\begin{equation*}
\begin{split}
E_0\left[\int_0^T\widetilde{g}_n(t)dt\right]&=E_0\left[\int_0^T\lim_{k\longrightarrow\infty}({g}_{\psi_k}(t))^+dt\right]\\
&\leq \lim_{k\longrightarrow\infty}E_0\left[\int_0^T\vert{g}_{\psi_k}(t)\vert{d}t\right]\leq \sup_{\psi\in\Psi_n}E_0\left[\int_0^T\vert{g}_{\psi}(t)\vert{d}t\right]= C_n.
\end{split}
\end{equation*}
This proves the first property in (\ref{GtildeVersusg(n)}). To prove the second property, we remark that due to Lemma \ref{Lemma1}-(c)-(d) ${g}_n={g}_{\psi_n}\leq \widetilde{g}_n$ obviously follow, and the rest of the proof focuses on the converse inequality. Let  $(\psi_k)_k\subset\Psi_n$ such that $g_{\psi_k}$ converges $P\otimes{dt}$-almost everywhere to $\widetilde{g}$, and in virtue of the first property in (\ref{GtildeVersusg(n)}) and the dominated convergence theorem, this convergence holds in $L^1(P\otimes dt)$. Thus, we derive 
\begin{equation*}
\begin{split}
E\left[\xi+\int_t^T \widetilde{g}(u)du\big|{\cal{F}}_t\right]&=\lim_{k\longrightarrow\infty}E\left[\xi+\int_t^T g_{\psi_k}(u)du\big|{\cal{F}}_t\right]\\
&=\lim_{k\longrightarrow\infty}Y^{\psi_k}_t\leq Y^{(n)}_t=E\left[\xi+\int_t^T g_{n}(u)du\big|{\cal{F}}_t\right]
\end{split}
\end{equation*}
Then by putting $t=0$ and taking expectation under $R_0$ afterwards, we get 
$$E\left[\int_0^T \widetilde{g}(u)du\right]=E\left[\int_0^T g_{n}(u)du\right].$$
Thus, by combining this with  ${g}_n\leq\widetilde{g}_n$ $P\otimes{dt}-a.e.$, we deduce that  ${g}_n= \widetilde{g}_n$, $P\otimes dt$-a.e.. This ends the proof of the lemma.
\end{proof}
\begin{proof}[Proof of Lemma \ref{LemmaforgTildeControl}] Let $\tau$ be a stopping time. Thanks to Lemma \ref{lemma4Gtilde}-(a), for any $n\geq 1$, we derive 
\begin{equation}
\begin{split}
\mathbb{E}_0\left[\int_{\tau}^T g_n(s)ds\big|{\cal{F}}_t\right]&=\mathbb{E}_0\left[-\xi+{Y}^{(n)}_t\big|{\cal{F}}_{\tau}\right]=\mathbb{E}_0\left[-\xi+Y^{\psi_n}_{\tau}\big|{\cal{F}}_{\tau}\right]\\
&\leq \mathbb{E}_0\left[\xi^-+\underset{\psi\in\Psi}{\esssup}\sup_{\tau\in{\cal{T}}}(Y^{\psi}_{\tau})^+\big|{\cal{F}}_{\tau}\right].
\end{split}
\end{equation} 
Then by combining this inequality with Fatou and Lemma \ref{lemma4Gtilde}-(c), assertion (a) follows immediately. \\
To prove assertion (b), we apply Garsia's lemma to the predictable process $\int_0^{\cdot}\widetilde{g}(u)du$, see \cite[Theorem 99]{dellacheriemeyer80}, and use assertion (a) to derive 
\begin{equation}\label{Inequa997}
\begin{split}
\mathbb{E}_0\left[\left(\int_0^T\widetilde{g}(u)du\right)^p\right]&\leq C_p \mathbb{E}_0\left[(\xi^-)^p\right]+ C_p \mathbb{E}_0\left[\left(\underset{\psi\in\Psi}{\esssup}\sup_{\tau\in{\cal{T}}}(Y^{\psi}_{\tau})^+\right)^p\right]\\
&\leq C_p\mathbb{E}_0\left[(\xi^-)^p\right]+ C_p \sup_{\psi\in\Psi} \mathbb{E}_0\left[\left(\sup_{\tau\in{\cal{T}}}(Y^{\psi}_{\tau})^+\right)^p\right]
\end{split}
\end{equation} 
The latter inequality follows from combining  the convergence monotone theorem and Lemma \ref{Lemma1}-(d), which states that the family $\{Y^{\psi}:\ \psi\in\Psi\}$ is upward direct.\\
Now remark that $\sup_{\tau\in{\cal{T}}}(Y^{\psi}_{\tau})^+\leq \sup_{0\leq{t}\leq{T}}E_0[D^{\psi}_T{\xi^+}(D^{\psi}_t)^{-1}\ \Big|{\cal{F}}_t]$, and consider
$$
\tau_a:=\inf\{t\geq 0:\ (Y_t^{\psi})^+ >a\},\quad a\in(0,\infty),\quad \inf(\emptyset):=\infty.$$
Then we have 
\begin{equation*}
\begin{split}
pa^{p-1}R_0\left( \sup_{0\leq{t}\leq{T}}(Y^{\psi}_t)^+>a\right)&\leq pa^{p-1}R_0\left( \tau_a\leq T\right)\\
&\leq pa^{p-2}E_0[(Y_{\tau_a}^{\psi})^+I_{\{\tau_a\leq T\}}]\leq pa^{p-2}E_0\left[{{D^{\psi}_T}\over{D^{\psi}_{T\wedge\tau_a}}}{\xi^+}I_{\{ \sup_{0\leq{t}\leq{T}}(Y_t^{\psi})^+\geq{a}\}}\right]
\end{split}
\end{equation*}
Then by integrating with respect to $da$ and put $q:=p/(p-1)$, we get 
\begin{equation}\label{Inequa998}
\begin{split}
\Vert \sup_{0\leq{t}\leq{T}}(Y^{\psi}_t)^+\Vert_{L^p(R_0)}^p&\leq{q} E_0\left[{{D^{\psi}_T}\over{D^{\psi}_{T\wedge\tau_a}}} {\xi^+}\left(\sup_{0\leq{t}\leq{T}}(Y^{\psi}_t)^+\right)^{p-1}\right]={q} E_{\psi}\left[{{\xi^+}\over{D^{\psi}_{T\wedge\tau_a}}}\left(\sup_{0\leq{t}\leq{T}}(Y^{\psi}_t)^+\right)^{p-1}\right]\\
&\leq {q}\Vert {\xi^+}(D^{\psi}_{T\wedge\tau_a})^{-1/p}\Vert_{L^p(R_{\psi})}\Vert  (D^{\psi}_{T\wedge\tau_a})^{-1/q}\left(\sup_{0\leq{t}\leq{T}}(Y^{\psi}_t)^+\right)^{p-1}   \Vert_{L^q(R_{\psi})}.
\end{split}
\end{equation}
In virtue of  
$$\sup_{0\leq{t}\leq{T}}(Y^{\psi}_t)^+\leq \sup_{0\leq{t}\leq{T}\wedge\tau_a}(Y^{\psi}_t)^++\sup_{{T}\wedge\tau_a\leq{t}\leq{T}}(Y^{\psi}_t)^+,$$ we deduce that 
\begin{equation}\label{Inequa999}
\begin{split}
&2^{1-q}\Vert  (D^{\psi}_{{T}\wedge\tau_a})^{-1/q}\left(\sup_{0\leq{t}\leq{T}}(Y^{\psi}_t)^+\right)^{p-1}   \Vert_{L^q(R_{\psi})}^q=2^{1-q}\E_{\psi}\left[(D^{\psi}_{{T}\wedge\tau_a})^{-1}\left(\sup_{0\leq{t}\leq{T}}(Y^{\psi}_t)^+\right)^{p} \right]\\
&\leq E_{\psi}\left[(D^{\psi}_{{T}\wedge\tau_a})^{-1}\left(\sup_{0\leq{t}\leq{T}\wedge\tau_a}(Y^{\psi}_t)^+\right)^{p} \right]+E_{\psi}\left[(D^{\psi}_{{T}\wedge\tau_a})^{-1}\left(\sup_{T\wedge\tau_a\leq{t}\leq{T}}(Y^{\psi}_t)^+\right)^{p} \right]\\
&\leq \Vert \sup_{0\leq{t}\leq{T}}(Y^{\psi}_t)^+\Vert_{L^p(R_0)}^p+p^pE_{\psi}\left[(D^{\psi}_{{T}\wedge\tau_a})^{-1}(\xi^+)^{p}\right]
\end{split}
\end{equation}
The last inequality follows from conditional Doob's inequality. Thus, by combining (\ref{Inequa997}),  (\ref{Inequa998}) and (\ref{Inequa999}) and Young's inequality, assertion (b) follows immediately. This ends the proof of the lemma.
\end{proof}
{\bf Acknowledgement:} The first author is very grateful to Professor Mich\`ele Vanmaele and the department of Applied Mathematics, Computer Science and Statistics at Ghent University, for their welcome and their hospitality, where the major part of this work was achieved when the first author was visiting Ghent. This visit was exclusively financed by FWO, through the mobility grant V500324N.\\
The first author would like to thank Professor Said Hamadene (Le Mans University) for his discussion about constrained BSDE and the literature about it. \\
 Ella Elazkany is financially supported by NSERC (through grant NSERC RGPIN04987). 



\begin{thebibliography}{1}

 \bibitem{Choulli2018}  {\sc  Aksamit, A., Choulli, T., Deng, J. and Jeanblanc, M.}   
 \newblock  No-arbitrage up to random horizon for quasi-left-continuous models.
 \newblock {\em Finance and Stochastics 21}, 4 (2017), 1103--1139
\bibitem{badescu2009esscher}
{\sc Badescu, A., Elliott, R.~J., and Siu, T.~K.}
\newblock Esscher transforms and consumption-based models.
\newblock {\em Insurance: Mathematics and Economics 45}, 3 (2009), 337--347.

\bibitem{barndorff1997normal}
{\sc Barndorff-Nielsen, O.~E.}
\newblock Normal inverse gaussian distributions and stochastic volatility modelling.
\newblock {\em Scandinavian Journal of Statistics 24}, 1 (1997), 1--13.

\bibitem{benth2012risk}
{\sc Benth, F.~E., and Sgarra, C.}
\newblock The risk premium and the Esscher transform in power markets.
\newblock {\em Stochastic Analysis and Applications 30}, 1 (2012), 20--43.

\bibitem{bernhart2012swing}
{\sc Bernhart, M., Pham, H., Tankov, P., and Warin, X.}
\newblock Swing options valuation: A BSDE with constrained jumps approach.
\newblock In: {\em Numerical Methods in Finance: Bordeaux, June 2010}. Springer, 2012, pp.~379--400.

\bibitem{bondi2020comparing}
{\sc Bondi, A., Radoji{\v{c}}i{\'c}, D., and Rheinl{\"a}nder, T.}
\newblock Comparing two different option pricing methods.
\newblock {\em Risks 8}, 4 (2020), 108.

\bibitem{boughamoura2014two}
{\sc Boughamoura, W., and Trabelsi, F.}
\newblock On two-parametric Esscher transform for geometric CGMY L{\'e}vy processes.
\newblock {\em International Journal of Operational Research 19}, 3 (2014), 280--301.

\bibitem{briand2003lp}
{\sc Briand, P., Delyon, B., Hu, Y., Pardoux, E., and Stoica, L.}
\newblock $L^p$ solutions of backward stochastic differential equations.
\newblock {\em Stochastic Processes and their Applications 108}, 1 (2003), 109--129.


\bibitem{buhlmann1980economic}
{\sc B{\"u}hlmann, H.}
\newblock An economic premium principle.
\newblock {\em ASTIN Bulletin: The Journal of the IAA 11}, 1 (1980), 52--60.

\bibitem{buhlmann1996no}
{\sc B{\"u}hlmann, H., Delbaen, F., Embrechts, P., and Shiryaev, A.~N.}
\newblock No-arbitrage, change of measure and conditional Esscher transforms.
\newblock {\em CWI quarterly 9}, 4 (1996), 291--317.

\bibitem{buhlmann1998esscher}
{\sc B{\"u}hlmann, H., Delbaen, F., Embrechts, P., and Shiryaev, A.~N.}
\newblock On Esscher transforms in discrete finance models.
\newblock {\em ASTIN Bulletin: The Journal of the IAA 28}, 2 (1998), 171--186.


 \bibitem{ChoulliDengMa}  {\sc  Choulli, T., Deng, J. and Ma, J.F..}   \newblock  How non-arbitrage, viability and num\'eraire portfolio are related. \newblock {\em Finance  and Stochastics 19}, 4 (2015), 719--741.
  \bibitem{Choulli2007} {\sc  Choulli, T., Stricker, Ch. and Li, J.}   \newblock  Minimal Hellinger martingale measures of order $q$. \newblock {\em Finance  and Stochastics 11}, 3 (2007), 399--427.  
  \bibitem{Choulli2006} {\sc  Choulli, T. and Stricker, Ch.}   \newblock  More on minimal entropy-Hellinger martingale measure. \newblock {\em Mathematical Finance 16},  1 (2006), 1--19.  
   \bibitem{Choulli2005} {\sc  Choulli, T. and Stricker, Ch.}   \newblock  Minimal entropy-Hellinger martingale measure in incomplete markets. \newblock {\em Mathematical Finance 15}, 3 (2005), 465--490.
\bibitem{Choulli99} {\sc  Choulli, T., Stricker, Ch., and Krawczyk, L.}:  On Fefferman and Burkholder-Davis-Gundy inequalities for $\cal{E}$-martingales. \newblock {\em Probability Theory and Related Fields 113}, 4 (1999),  571--597.

\bibitem{cvitanic1998backward}
{\sc Cvitanic, J., Karatzas, I., and Soner, H.~M.}
\newblock Backward stochastic differential equations with constraints on the gains-process.
\newblock {\em The Annals of Probability 26}, 4 (1998), 1522--1551.



\bibitem{dellacheriemeyer92} {\sc Dellacherie, M., Maisonneuve, B. and  Meyer  P-A.} \newblock {\em Probabilit\'es et Potentiel}, Volume 5, chapitres XVII-XXIV: Processus de Markov, compl\'ements aux calculs stochastiques. Hermann, Paris, 1992.

\bibitem{dellacheriemeyer80}{\sc Dellacherie, C. and Meyer P-A.} \newblock {\em Probabilit\'es et Potentiel}, Volume 2, chapitres V-VIII: Th\'eorie des martingales.  Hermann, Paris, 1980.

\bibitem{delbaen2010harmonic}
{\sc Delbaen, F., and Tang, S.}
\newblock Harmonic analysis of stochastic equations and backward stochastic differential equations. {\em Probability Theory and Related Fields 146},  1-2 (2010), 291--336.

\bibitem{DelbaenEsscher}
{\sc Delbaen, F., and Haezendonck}
\newblock A martingale approach to premium calculation principles in an arbitrage free market. {\em Insurance: Mathematics and Economics 8},  4 (1989), 269--277.

\bibitem{elie2014adding}
{\sc Elie, R., and Kharroubi, I.}
\newblock Adding constraints to BSDEs with jumps: an alternative to multidimensional reflections.
\newblock {\em ESAIM: Probability and Statistics 18}, (2014), 233--250.

\bibitem{elliott2022generalized}
{\sc Elliott, R., and Siu, T.}
\newblock A generalized Esscher transform for option valuation with regime switching risk.
\newblock {\em Quantitative Finance 22}, 4 (2022), 691--705.

\bibitem{essaky2008reflected}
{\sc Essaky, E.}
\newblock Reflected backward stochastic differential equation with jumps and RCLL obstacle.
\newblock {\em Bulletin des Sciences Math\'ematiques 132}, 8 (2008), 690--710.

\bibitem{essaky2011general}
{\sc Essaky, E., and Hassani, M.}
\newblock General existence results for reflected BSDE and BSDE.
\newblock {\em Bulletin des Sciences Math{\'e}matiques 135}, 5 (2011), 442--466.

\bibitem{Esscher} {\sc  Esscher, F.} \newblock  On the probability function in collective theory of risk. \newblock {\em Skandinavisk Aktuarietidskrift 15},  (1932), 175--195.

\bibitem{Follmer-Kramkov} {\sc  F\"ollmer, H. and Kramkov, D.} \newblock Optional decompositions under constraints. \newblock {\em Probability Theory and Related Fields 109}, 1 (1997), 1--25.

\bibitem{fujiwara2003minimal}
{\sc Fujiwara, T., and Miyahara, Y.}
\newblock The minimal entropy martingale measures for geometric L{\'e}vy processes.
\newblock {\em Finance and Stochastics 7}, 4 (2003), 509--531.

\bibitem{gerber1994a} {\sc Gerber, H.~U., and Shiu, E.~S.} \newblock Martingale approach to pricing perpetual American options.
\newblock {\em Astin Bulletin 24}, 2 (1994), 195--220.

 \bibitem{gerber1994b}{\sc Gerber, H.~U., and Shiu, E.~S.} \newblock Option pricing by Esscher transform. 
\newblock {\em  Transaction of the Society of Actuaries 46},  (1994), 99--140.

\bibitem{gerber1996}
{\sc Gerber, H.~U., and Shiu, E.~S.}
\newblock Actuarial bridges to dynamic hedging and option pricing.
\newblock {\em Insurance: Mathematics and Economics 18}, 3 (1996), 183--218.

\bibitem{grandits1999p}
{\sc Grandits, P.}
\newblock The $p$-optimal martingale measure and its asymptotic relation with the minimal-entropy martingale measure.
\newblock {\em Bernoulli 5}, 2 (1999), 225--247.

\bibitem{hamadene2012lp}
{\sc Hamad{\`e}ne, S., and Popier, A.}
\newblock $L^p$-solutions for reflected backward stochastic differential equations.
\newblock {\em Stochastics and Dynamics 12}, 02 (2012), 1150016.

\bibitem{he2019semimartingale}
{\sc He, S.-w., Wang, J.-g., and Yan, J.-a.}
\newblock {\em Semimartingale Theory and Stochastic Calculus}.
\newblock Routledge, New York, 2019.

\bibitem{hubalek2006}
{\sc Hubalek, F., and Sgarra, C.}
\newblock On then Esscher transforms and other equivalent martingale measures for Barndorff-Nielsen and Shephard stochastic volatility models with jumps.
\newblock {\em Stochastic Processes and their Applications 119}, 7 (2009),  2137--2157.

\bibitem{hubalek2009}
{\sc Hubalek, F., and Sgarra, C.}
\newblock Esscher transforms and the minimal entropy martingale measure for exponential L{\'e}vy models.
\newblock {\em Quantitative Finance 6}, 02 (2006), 125--145.

\bibitem{jacod2006calcul}
{\sc Jacod, J.}
\newblock {\em Calcul Stochastique et Probl\`emes de Martingales}, vol.~714.
\newblock Springer, 2006.

\bibitem{jacod2013limit}
{\sc Jacod, J., and Shiryaev, A.}
\newblock {\em Limit Theorems for Stochastic Processes}, vol.~288.
\newblock Springer Science \& Business Media, 2013.

\bibitem{Kallsen}
{\sc Kallsen, J., and Shiryaev, A.~N.}
\newblock The cumulant process and Esscher's change of measure.
\newblock {\em Finance and Stochastics 6}, 4 (2002), 397--428.

\bibitem{kharroubi2010backward}
{\sc Kharroubi, I., Ma, J., Pham, H., and Zhang, J.}
\newblock Backward SDEs with constrained jumps and quasi-variational inequalities.
\newblock {\em The Annals of Probability 38}, 2 (2010), 794--840.

\bibitem{lau2008option}
{\sc Lau, J.~W., and Siu, T.~K.}
\newblock On option pricing under a completely random measure via a generalized Esscher transform.
\newblock {\em Insurance: Mathematics and Economics 43}, 1 (2008), 99--107.

\bibitem{lepeltier2005penalization}
{\sc Lepeltier, J.-P., and Xu, M.}
\newblock Penalization method for reflected backward stochastic differential equations with one r.c.l.l.~barrier.
\newblock {\em Statistics \& Probability Letters 75}, 1 (2005), 58--66.

\bibitem{li2018exchange}
{\sc Li, W., Liu, L., Lv, G., and Li, C.}
\newblock Exchange option pricing in jump-diffusion models based on Esscher transform.
\newblock {\em Communications in Statistics-Theory and Methods 47}, 19 (2018), 4661--4672.

\bibitem{madan1998variance}
{\sc Madan, D.~B., Carr, P.~P., and Chang, E.~C.}
\newblock The Variance Gamma process and option pricing.
\newblock {\em Review of Finance 2}, 1 (1998), 79--105.

\bibitem{madan1991option}
{\sc Madan, D.~B., and Milne, F.}
\newblock Option pricing with V.G. martingale components.
\newblock {\em Mathematical Finance 1}, 4 (1991), 39--55.

\bibitem{miyahara2001geometric}
{\sc Miyahara, Y.}
\newblock [Geometric L{\'e}vy process \& MEMM] pricing model and related estimation problems.
\newblock {\em Asia-Pacific Financial Markets 8}, 1 (2001), 45--60.

\bibitem{monfort2012asset}
{\sc Monfort, A., and Pegoraro, F.}
\newblock Asset pricing with second-order Esscher transforms.
\newblock {\em Journal of Banking \& Finance 36}, 6 (2012), 1678--1687.

\bibitem{monoyios2007minimal}
{\sc Monoyios, M.}
\newblock The minimal entropy measure and an Esscher transform in an incomplete market model.
\newblock {\em Statistics \& Probability Letters 77}, 11 (2007), 1070--1076.

\bibitem{ouknine1998reflected}
{\sc Ouknine, Y.}
\newblock Reflected backward stochastic differential equations with jumps.
\newblock {\em Stochastics: An International Journal of Probability and Stochastic Processes 65}, 1-2 (1998), 111--125.

\bibitem{peng1999monotonic}
{\sc Peng, S.}
\newblock Monotonic limit theorem of BSDE and nonlinear decomposition theorem of Doob--Meyer's type.
\newblock {\em Probability Theory and Related Fields 113}, 4 (1999), 473--499.

\bibitem{peng2010reflected}
{\sc Peng, S., and Xu, M.}
\newblock Reflected BSDE with a constraint and its applications in an incomplete market.
\newblock {\em Bernoulli 6}, 3 (2010), 614--640.

\bibitem{Rheinlander1}
{\sc Rheinl\"ander, Th. and Sexton, J.}
\newblock {\em Hedging Derivatives.} Advanced series on Statistical Science \& Applied Probability, Volume 15.
\newblock  World Scientific Publishing Company, 2011.


\bibitem{rong1997solutions}
{\sc Rong, S.}
\newblock On solutions of backward stochastic differential equations with jumps and applications.
\newblock {\em Stochastic Processes and their Applications 66}, 2 (1997), 209--236.

\bibitem{schachermayer2008close}
{\sc Schachermayer, W., and Teichmann, J.}
\newblock How close are the option pricing formulas of Bachelier and Black--Merton--Scholes?
\newblock {\em Mathematical Finance 18}, 1 (2008), 155--170.

\bibitem{schoutens2003levy}
{\sc Schoutens, W.}
\newblock {\em L{\'e}vy Processes in Finance: Pricing Financial Derivatives}.
\newblock Wiley Series in Probability and Statistics. John Wiley \& Sons Ltd., Chichester, 2003.


\bibitem{siu2001bayesian}
{\sc Siu, T.~K., Tong, H., and Yang, H.}
\newblock Bayesian risk measures for derivatives via random Esscher transform.
\newblock {\em North American Actuarial Journal 5}, 3 (2001), 78--91.

\bibitem{Stricker}
{\sc Stricker, C.}
\newblock Quelques remarques sur la topologie des semimartingales. Applications aux int{\'e}grales stochastiques.
\newblock In: Az\'ema, J., Yor, M. (eds), {\em S{\'e}minaire de Probabilit{\'e}s XV 1979/80}, Lecture Notes in mathematics, volume 850. Springer,Berlin, Heidelberg, 1981, pp.~499--522.

\bibitem{wang2007normalized}
{\sc Wang, S.}
\newblock Normalized exponential tilting: pricing and measuring multivariate risks.
\newblock {\em North American Actuarial Journal 11}, 3 (2007), 89--99.

\bibitem{wang2000class}
{\sc Wang, S.~S.}
\newblock A class of distortion operators for pricing financial and insurance risks.
\newblock {\em Journal of Risk and Insurance 67}, 1 (2000), 15--36.

\bibitem{wang2003equilibrium}
{\sc Wang, S.~S.}
\newblock Equilibrium pricing transforms: new results using Buhlmann’s 1980 economic model.
\newblock {\em ASTIN Bulletin: The Journal of the IAA 33}, 1 (2003), 57--73.

\bibitem{Xu} Xu, M. \newblock Superhedging problem under ratio constraint: BSDE approaches with Malliavin clculus.
\newblock {\em Numerical Algebra, Control and Optimization 13}, 3 $\&$ 4 (2023), 664--680.
\end{thebibliography}
\end{document}